\documentclass[aps,physrev,twocolumn,superscriptaddress,longbibliography, 10pt]{revtex4-2}
\usepackage[T1]{fontenc}
\usepackage[utf8]{inputenc}
\usepackage[british]{babel} 
\usepackage{lmodern}
\usepackage{microtype}

\usepackage{graphicx}
\usepackage{dcolumn}
\usepackage[table,dvipsnames]{xcolor}

\usepackage{amsmath}
\usepackage{amssymb}
\usepackage{amsfonts}
\usepackage{mathtools}
\usepackage{dsfont}
\usepackage{amsthm}

\usepackage[caption=false]{subfig}
\usepackage{paralist}
\usepackage{booktabs}
\usepackage{fontawesome5}

\usepackage{silence} 
\usepackage[
    colorlinks,
    linkcolor={black!30!blue},
    citecolor={black!30!blue},
    urlcolor={black!30!blue}
]{hyperref} 

\usepackage{cleveref} 
\crefformat{appendix}{\textup{appendix~}#2\textup{#1}#3}
\Crefformat{appendix}{\textup{Appendix~}#2\textup{#1}#3}

\usepackage{tikz}
\usepackage{quantikz}
\usetikzlibrary{positioning}
\usetikzlibrary{arrows.meta}
\usetikzlibrary{calc}

\tikzset{
  pulse/.style={
    fill=tunared!75,
    minimum width=0.8cm,
    minimum height=0.8cm
  },
  native/.style={
    fill=sand!60,
    inner sep=2pt
  },
  target/.style={
    fill=targetblue!60,
    inner sep=2pt
  }
}

\usepackage[nolist]{acronym}
\newacro{LP}{linear program} 
\newacro{PBC}{periodic boundary condition}

\usepackage{cancel}

\definecolor{tunared}{RGB}{214,170,170}
\definecolor{sand}{RGB}{210,190,150}
\definecolor{targetblue}{RGB}{180,205,235}

\definecolor{martin}{rgb}{0,.4,1}

\definecolor{ozgun}{RGB}{74,173,116}

\definecolor{claude}{RGB}{176,48,160}

\newtheorem{theorem}{Theorem}

\newtheorem{lemma}[theorem]{Lemma} 
\newtheorem{proposition}[theorem]{Proposition}
\newtheorem{definition}[theorem]{Definition}
\AddToHook{env/lemma/begin}{\crefalias{theorem}{lemma}} 

\newcommand{\e}{\ensuremath\mathrm{e}} 
\renewcommand{\i}{\ensuremath\mathrm{i}} 
\DeclareMathOperator{\Tr}{Tr} 
\renewcommand{\Re}{\operatorname{Re}} 
\renewcommand{\Im}{\operatorname{Im}} 
\DeclareMathOperator{\rank}{rank} 
\DeclareMathOperator{\supp}{supp} 
\DeclareMathOperator{\wt}{wt} 
\newcommand{\Herm}{\mathrm{Herm}}
\newcommand{\diag}{\mathrm{diag}}

\DeclareMathOperator{\U}{U}

\newcommand{\CC}{\mathbb{C}}
\newcommand{\RR}{\mathbb{R}}

\newcommand{\ZZ}{\mathbb{Z}}

\newcommand{\1}{\mathds{1}} 
\newcommand{\PP}{\mathbb{P}} 
\newcommand{\TT}{\mathbb{T}}

\newcommand{\mc}[1]{\mathcal{#1}}

\renewcommand{\vec}[1]{\mathbf{#1}} 

\DeclareMathOperator{\LandauO}{\mathrm{O}} 

\newcommand{\func}[1]{{\ensuremath{\mathsf{#1}}}}

\newcommand{\poly}{\func{poly}}

\DeclarePairedDelimiterX{\abs}[1]{\lvert}{\rvert}{%
  \ifblank{#1}{\,\cdot\,}{#1}
}   

\DeclarePairedDelimiterX\norm[1]\lVert\rVert{%
  \ifblank{#1}{\,\cdot\,}{#1}
}   

\DeclarePairedDelimiterX{\iiiNorm}[1]{\lvert}{\rvert}{%
  \delimsize\lvert\delimsize\lvert#1\delimsize\rvert\delimsize\rvert%
}

\newcommand{\ind}[1]{{\boldsymbol{#1}}}

\newcommand{\cdag}{c^{\dagger}}

\DeclarePairedDelimiterXPP\snorm[1]{}\lVert\rVert{_\infty}{\ifblank{#1}{\,\cdot\,}{#1}}   

\DeclarePairedDelimiterXPP\twonorm[1]{}\lVert\rVert{_2}{\ifblank{#1}{\,\cdot\,}{#1}}   

\DeclarePairedDelimiterXPP\trnorm[1]{}\lVert\rVert{_1}{\ifblank{#1}{\,\cdot\,}{#1}}   

\DeclarePairedDelimiterXPP\fnorm[1]{}\lVert\rVert{_{\fro}}{\ifblank{#1}{\,\cdot\,}{#1}}   

\DeclarePairedDelimiterXPP\dnorm[1]{}\lVert\rVert{_\diamond}{\ifblank{#1}{\,\cdot\,}{#1}}   

\DeclarePairedDelimiterXPP\cbnorm[1]{}\lVert\rVert{_\mathrm{cb}}{\ifblank{#1}{\,\cdot\,}{#1}}   
\DeclarePairedDelimiterXPP\onenorm[1]{}\lVert\rVert{_{1\rightarrow 1}}{\ifblank{#1}{\,\cdot\,}{#1}}   
\DeclarePairedDelimiterXPP\ddnorm[1]{}\lVert\rVert{_{\diamond\rightarrow \diamond}}{\ifblank{#1}{\,\cdot\,}{#1}}   
\DeclarePairedDelimiterXPP\ssnorm[1]{}\lVert\rVert{_{\infty\rightarrow\infty}}{\ifblank{#1}{\,\cdot\,}{#1}}   

\providecommand\given{}
\newcommand\SetSymbol[1][]{%
  \nonscript\:#1\vert
  \allowbreak
  \nonscript\:
  \mathopen{}}
\DeclarePairedDelimiterX\Set[1]\{\}{%
  \renewcommand\given{\SetSymbol[\delimsize]}
  #1
}

\DeclarePairedDelimiterX\innerp[2]{\langle}{\rangle}{%
  \ifblank{#1}{\,\cdot\,}{#1} , \ifblank{#2}{\,\cdot\,}{#2}%
}

\let\bra\relax
\let\ket\relax
\let\braket\relax
\let\ketbra\relax

\DeclarePairedDelimiter{\bra}{\langle}{\vert}
\DeclarePairedDelimiter{\ket}{\vert}{\rangle}

\DeclarePairedDelimiterX\braket[2]{\langle}{\rangle}%
  {#1\kern0.15ex\delimsize\vert\kern0.15ex\mathopen{}#2}

\DeclarePairedDelimiterX\ketbra[2]{\vert}{\vert}%
  {#1\kern0.15ex\delimsize\rangle\delimsize\langle\kern0.15ex\mathopen{}#2}

\DeclarePairedDelimiterX\sandwich[3]{\langle}{\rangle}%
  {#1\,\delimsize\vert\kern0.15ex\mathopen{}#2\kern0.15ex\delimsize\vert\kern0.15ex\mathopen{}#3}

\DeclarePairedDelimiterX\obraket[2]{(}{)}%
  {#1\kern0.15ex\delimsize\vert\kern0.15ex\mathopen{}#2}

\DeclarePairedDelimiterX\oketbra[2]{\vert}{\vert}%
  {#1\kern0.15ex\delimsize)\delimsize(\kern0.15ex\mathopen{}#2}

\DeclarePairedDelimiterX\osandwich[3]{(}{)}%
  {#1\,\delimsize\vert\kern0.15ex\mathopen{}#2\kern0.15ex\delimsize\vert\kern0.15ex\mathopen{}#3}

\newcommand{\cone}{\operatorname{cone}}
\newcommand{\col}{\operatorname{col}}
\newcommand{\conv}{\operatorname{conv}}
\newcommand{\FH}{{\mathrm{FH}}}
\newcommand{\Pb}{{\bar{P}}}

\definecolor{fluxblue}{RGB}{52, 89, 149}

\begin{document}

\title{Fermionic Hamiltonian engineering with local control}

\newcommand{\tuhh}{Hamburg University of Technology, Institute for Quantum Inspired and Quantum Optimization, Hamburg, Germany}

\author{Özgün Kum}
\email{oezguen.kum@tuhh.de}
\affiliation{\tuhh}
\author{Matthias Zipper}
\affiliation{\tuhh}
\author{Ludwig Mathey}
\affiliation{Zentrum für Optische Quantentechnologien and Institut für Quantenphysik, Universität Hamburg, 22761 Hamburg, Germany}
\affiliation{The Hamburg Centre for Ultrafast Imaging, Luruper Chaussee 149, Hamburg 22761, Germany}
\author{Martin Kliesch}
\email{martin.kliesch@tuhh.de}
\affiliation{\tuhh}

\begin{abstract}
  Quantum simulators enable the exploration of complex quantum phenomena in condensed-matter systems by reproducing their dynamics on controllable quantum devices.
  However, experimental constraints often restrict the class of Hamiltonians that can be realized natively.
  Hamiltonian engineering addresses this limitation by expanding the set of accessible Hamiltonians beyond the fixed system Hamiltonian provided by the hardware.
  We introduce a framework for fermionic Hamiltonian engineering based on interleaving free evolution under the system Hamiltonian with a sequence of experimentally feasible local fermionic unitaries.
  The required sequences and free-evolution times are obtained efficiently through a linear program, enabling effective time evolution under a broad class of target Hamiltonians.
  In particular, our method realizes arbitrary complex tunnelling coefficients, constrained only by the connectivity of the underlying system Hamiltonian.
  We further show that symmetries of the Hamiltonian engineering task can be exploited to drastically reduce the cost of computing the control parameters, and that errors arising from finite pulse times can be compensated within the same framework through a simple modification of the linear program.
  To our knowledge, this is the first automated method to realize locally tunable artificial gauge fields in interacting fermionic lattice models using only site-resolved energy shifts as control.
  We demonstrate the method by engineering the dynamics of interacting Fermi--Hubbard chains with complex tunnelling coefficients and of the Harper--Hofstadter model on a 1088-mode lattice.
  For the latter, we exploit the translation symmetry of the model to determine a provably optimal six-pulse sequence that produces the desired dynamics independent of the lattice size.
\end{abstract}

\maketitle

\hypersetup{pdftitle = {Fermionic Hamiltonian engineering with local control},
	     pdfauthor = {Özgün Kum, Matthias Zipper, Ludwig Mathey, Martin Kliesch},
	     pdfsubject = {Quantum simulation},
       pdfkeywords = {quantum simulation, quantum, fermionic, Fermi-Hubbard,
                           Harper--Hofstadter, Hofstadter--Hubbard, Hamiltonian engineering,
                           artificial gauge fields, complex tunnelling, optical lattices,
                           ultracold atoms, analogue quantum simulator, Floquet,
                           topological phases, local pulses, local control, local gates, linear program, linear programming, average Hamiltonian theory, quantum advantage}
	    }

\section{Introduction}
\label{sec:intro}

The simulation of quantum systems is widely regarded as the primary domain in which quantum devices are expected to offer a practical computational advantage over classical algorithms \cite{georgescuQuantumSimulation2014}. 
In particular, fermionic quantum many-body systems are a central target for this advantage due to their role in describing high-temperature superconductivity \cite{arovasHubbardModel2022,qinHubbardModelComputational2022}, exotic topological phases  \cite{cooperTopologicalBandsUltracold2019}, and complex electronic structures \cite{McArdle2020,bauerQuantumAlgorithmsQuantum2020}. 
Classically, these systems are expected to be largely intractable due to the exponential growth of the Hilbert space dimension and the pervasive fermionic sign problem  \cite{troyerComputationalComplexityFundamental2005}.
While it has been argued that analogue quantum simulators  \cite{blochQuantumSimulationsUltracold2012, altmanQuantumSimulatorsArchitectures2021} have already demonstrated a form of quantum advantage  \cite{daleyPracticalQuantumAdvantage2022,flanniganPropagationErrorsQuantitative2022, trivediQuantumAdvantageStability2023}, such devices often lack the programmability required to fit within the formal complexity-theoretic frameworks used to certify and verify computational speed-ups.

Tailoring quantum dynamics to implement desired unitary operations has a long history. 
Around 2000, the first methods for engineering arbitrary two-local Hamiltonians using a different two-local Hamiltonian in conjunction with handcrafted local control pulses \cite{Nielsen2002,Dodd2002} were developed. 
Methods from parallel works \cite{Leung2000,Leung2002} could effectively decouple all but one interaction in a two-local Hamiltonian, also resulting in a Hamiltonian engineering scheme. 
To automate the engineering, linear programming was first used by \citet{Hayes14ProgrammableQuantumSimulation} for spin chain dynamics.
Again in parallel, substantial progress on making the engineered dynamics robust against various errors using average Hamiltonian theory was made in the context of global control pulses \cite{choi_optimal_2014, Choi2020}. 
Then, making use of local pulses in the context of two-local Ising-type all-to-all interactions for trapped ion platforms, 
the efficient, arbitrary, simultaneously entangling (EASE) gate design was developed and implemented \cite{grzesiak_efficient_2020}. 
Later, the robustness ideas were extended to two-local XY model Hamiltonian engineering with local pulses \cite{Votto23UniversalQuantumProcessors}. 
In a similar context, linear programming was used to minimize the total gate time \cite{Bassler22MAGICal}. 
This method was then first made efficient \cite{Bassler24Time-optimal} and then extended to arbitrary Hamiltonians with the ability to change coupling types, made more efficient and, based on \cite{Votto23UniversalQuantumProcessors}, also made robust against various errors \cite{basslerGeneralEfficientRobust2025a}. 

However, all of these works address qubit systems, and simulating fermionic dynamics even on digital qubit hardware is challenging because it requires the use of fermion-to-qubit encodings.
These incur an overhead either in operator weight, as for the Jordan--Wigner or Bravyi--Kitaev transformations \cite{JordanWigner1928, BravyiKitaev2002}, or in the number of qubits, as for local encodings \cite{Chien2026}. 
To circumvent this overhead, programmable quantum processors operating natively on fermionic modes have been proposed \cite{gonzalez-cuadraFermionicQuantumProcessing2023a}, and high-fidelity entangling gates between fermionic atoms have recently been demonstrated \cite{bojovicHighfidelityCollisionalQuantum2026}, an important experimental step toward their realization. 
Nevertheless, this route inherits the central difficulty of digital
quantum simulation, as the entangling-gate count grows with system size and simulation time, and without error correction the accumulated infidelity limits the accessible dynamics.  

On the algorithmic side, natively fermionic analogue hardware has been considered mainly for variational ground-state preparation, via fermionic circuits for quantum chemistry on optical superlattices \cite{gkritsisSimulatingChemistryFermionic2025}, and via variational circuits generated by quenches of the native Fermi--Hubbard dynamics \cite{tabaresProgrammingOpticalLatticeFermiHubbard2025a}.
In a concurrent and independent work \cite{weiUniversalEfficientHybrid2026}, a variational protocol was applied to the Hofstadter--Hubbard model, where the target's complex tunnelling coefficients are not realized as physical couplings but enter through classical postprocessing.
Common to all of these variational approaches is that they extract static properties of approximate ground states rather than realizing a target time evolution.
Furthermore, variational algorithms are heuristic and come with challenges related to trainability \cite{Bittel21TrainingVariationalQuantum,Bittel22OptimizingTheDepth, Larocca2025} and the potentially large number of quantum circuit evaluations required \cite{Gonthier22MeasurementsAsARoadblock}.
Complementary to these state-preparation schemes, universal dynamics under purely global control were recently established for a broad class of analogue simulators, including fermionic platforms \cite{hu2025universaldynamicsgloballycontrolled}, although the required control parameters are obtained by direct optimal control, which is not scalable to large systems.

In this work, we introduce an efficient framework for fermionic Hamiltonian engineering (FHE) that enhances the programmability of natively fermionic analogue quantum simulators, taking an important step toward practical quantum advantage in the simulation of strongly correlated fermionic matter.
Building upon the conjugation-based framework for qubit systems \cite{basslerGeneralEfficientRobust2025a}, our method brings the programmability of digital simulators to analogue hardware. This is achieved through synthesizing target evolutions by interleaving the evolution under a native fermionic system Hamiltonian $H_S$ with sequences of local fermionic unitaries. 
The local unitaries we use are generated by number operators, i.e.\ by local shifts of the mode energies, which makes them among the least demanding local operations a platform has to provide.
We demonstrate that this approach enables the simulation of a broad class of target Hamiltonians featuring locally tunable, complex tunnelling coefficients, constrained only by the connectivity of $H_S$.
This provides a systematic route to engineering artificial gauge fields and exploring interacting topological phases without requiring complex hardware modifications. 
The required unitary sequences are synthesized efficiently by solving a linear program.
We extend this framework to account for implementation errors caused by the finite duration of pulses in the presence of an always-on system Hamiltonian. Using average Hamiltonian theory, we derive a corrective modification to the linear constraints that systematically mitigates these errors.

Our framework is motivated by the rapid experimental progress in the single-site control of ultracold atoms in optical lattices \cite{esslingerFermiHubbardPhysicsAtoms2010,weitenbergSinglespinAddressingAtomic2011}. 
In particular, the development of quantum gas microscopes \cite{bakrQuantumGasMicroscope2009, grossQuantumGasMicroscopy2021} has enabled site-resolved imaging and, more recently, the local readout and manipulation of individual modes \cite{impertroLocalReadoutControl2024}, while optical superlattices make the Hubbard couplings themselves tunable \cite{chalopinOpticalSuperlatticeEngineering2025}.
Furthermore, the advent of optical tweezer arrays \cite{barredoSyntheticThreedimensionalAtomic2018, bernienProbingManybodyDynamics2017, manetschTweezerArray61002024} allows atoms to be arranged in highly programmable geometries.
These technologies provide the toolkit of local operations that our method leverages.

As a concrete application of our method, we consider quantum quenches under the Harper--Hofstadter \cite{cooperTopologicalBandsUltracold2019} and Fermi--Hubbard Hamiltonians \cite{arovasHubbardModel2022, qinHubbardModelComputational2022}, the latter of which is a popular model for exploring strongly correlated electronic systems.
While standard analogue simulators are typically restricted to the real and homogeneous native tunnelling and interaction terms of the lattice, our approach enables the engineering of complex-valued, site-dependent tunnelling coefficients.
Quenches in large interacting systems are among the most promising candidates for achieving early practical quantum advantage \cite{flanniganPropagationErrorsQuantitative2022}, because the linear growth of entanglement quickly renders classical tensor-network simulations intractable \cite{schuchEntropyScalingSimulability2008, schollwoeckDensitymatrixRenormalizationGroup2011}.
We provide a numerical demonstration of our protocol, showing that the synthesized dynamics accurately track the target evolution.
For the non-interacting case, where the dynamics is classically tractable and serves as a benchmark for the protocol, we reach a 2D lattice of 1088 modes. 
Importantly, the efficiency of the pulse synthesis is not affected when we consider interacting systems. 
We restrict our numerical demonstrations on interacting systems to small system sizes only because classical benchmarking relies on exact diagonalization, the cost of which grows exponentially with the number of modes.

For the past two decades, Floquet driving has been a state-of-the-art approach to Hamiltonian engineering in fermionic systems \cite{eckardtAtomicQuantumGases2017}, enabling the realization of synthetic gauge fields and topological phases in otherwise inaccessible systems \cite{struckEngineeringIsingXYSpin2013,jotzuExperimentalRealizationTopological2014,aidelsburgerExperimentalRealizationStrong2011,aidelsburgerRealizationHofstadterHamiltonian2013,aidelsburgerMeasuringChernNumber2015,miyakeRealizingHarperHamiltonian2013}.
By periodically modulating system parameters, Floquet schemes reshape the effective stroboscopic dynamics, and it has been shown that tunnelling coefficients can even be tuned individually using periodic driving \cite{nixonIndividuallyTunableTunnelling2024}.
However, periodically driven interacting systems are susceptible to Floquet heating, whereby energy absorption from the drive can eventually invalidate the effective-Hamiltonian description.
Furthermore, Floquet-based approaches typically require a target-specific driving protocol to be designed and its parameters to be determined, even in schemes offering local control over individual tunnelling coefficients \cite{nixonIndividuallyTunableTunnelling2024}.
In contrast, our approach replaces continuous modulation with optimized sequences of discrete local unitaries and systematically constructs the required pulse sequence directly from descriptions of the system and target Hamiltonians.
This avoids relying on continuous periodic driving to generate the desired couplings, while providing an automated route to Hamiltonian synthesis.

The remainder of this paper is organized as follows. \Cref{sec:Preliminaries} introduces the necessary notation and preliminaries.
In \cref{sec:tunnelling}, we demonstrate the core capabilities of the method by locally tuning tunnelling parameters in both quadratic and interacting Hamiltonians. 
Here, we also demonstrate the power of our method by providing some applications.
First, in \cref{sec:quadratic}, we study the Harper--Hofstadter model on a 1088-mode lattice and a triangular lattice with artificial gauge fields, then in \cref{sec:interactions}, we focus on the Fermi--Hubbard model, using the fermionic Hamiltonian engineering framework to tune the relative strength of interactions to tunnellings, and implementing interacting Hamiltonians with arbitrary complex tunnelling coefficients.
\Cref{sec:FHE_general} develops the technical foundations of our approach and presents the  technical main results.
We begin the section by outlining the limits of the framework when applied to interaction terms in \cref{sec:nogo}.
In \cref{sec:tritstring_Hamiltonians}, we introduce a notation that allows fermionic Hamiltonians without number operators to be expressed in close analogy to Pauli decompositions of qubit Hamiltonians. 
Using this notation, \cref{sec:he_tritstring} shows that any Hamiltonian sharing the same connectivity graph as the system Hamiltonian can be realized and derives the corresponding linear program used to compute the local operations required for Hamiltonian engineering. 
In \cref{sec:implementation_errors}, we analyze the dominant sources of error, namely Trotter and finite-pulse-time errors, quantify the quantum resource requirements in terms of runtime and pulse count, and show how finite-pulse-time errors can be systematically addressed.
We conclude with an outlook in \cref{sec:conclusion}.

\begin{figure*}[t]
  \centering

  \begin{minipage}[c]{0.19\textwidth}
  \centering
  \vspace{0.7cm}
  \begin{tikzpicture}[
    font=\sffamily,
    box/.style={
      draw,
      rounded corners=4pt,
      line width=1pt,
      align=left,
      inner sep=8pt,
      minimum width=3.5cm
    },
    arrow/.style={-{Latex[length=3mm]}, line width=1pt}
  ]

  \node[box] (input) {
    \textbf{Input:} \\[3pt]
   Hamiltonians: $H_S, H_T$
  };

  \node[
    below=0.5cm of input,
    font=\fontsize{40}{40}\selectfont
  ] (laptop) {\faLaptop};

  \node[
    box,
    below=0.5cm of laptop,
  ] (output) {
    \textbf{Output:}\\[3pt]
   Pulses, times: $\{V_\ind{b},\lambda_\ind{b}\}_\ind{b}$
  };

  \draw[arrow] (input.south) -- (laptop.north);
  \draw[arrow] (laptop.south) -- (output.north);

  \end{tikzpicture}
  \end{minipage}
  \hfill
  \begin{minipage}[c]{0.78\textwidth}
  \centering
  \vspace{0pt}
  \[
  \begin{quantikz}[row sep=0.30cm, column sep=0.40cm]
      & \gate[style=pulse]{V_1}\gategroup[5,steps=1,style={dashed,rounded corners,fill=red!12},background,label style={label position = above, anchor=north,yshift = 0.4cm}]{$V_{\ind{b}_1}$}
      & \gate[5, style = native]{e^{-i H_S \lambda_{\ind{b}_1} t}}
      & \gate[style=pulse]{V_1^{\dagger}}\gategroup[5,steps=1,style={dashed,rounded corners,fill=red!12},background,label style={label position = above, anchor=north,yshift = 0.4cm}]{$V_{\ind{b}_1}^\dagger$}
      & \ \cdots \
      & \gate[style=pulse]{V_1}\gategroup[5,steps=1,style={dashed,rounded corners,fill=red!12},background,label style={label position = above, anchor=north,yshift = 0.4cm}]{$V_{\ind{b}_{2r}}$}
      & \gate[5, style = native]{e^{-i H_S \lambda_{\ind{b}_{2r}} t}}
      & \gate[style=pulse]{V_1^{\dagger}}\gategroup[5,steps=1,style={dashed,rounded corners,fill=red!12},background,label style={label position = above, anchor=north,yshift = 0.4cm}]{$V_{\ind{b}_{2r}}^\dagger$}
      & \qw\\
      & \gate[style=pulse]{V_2^\dagger}
      &
      & \gate[style=pulse]{V_2}
      & \ \cdots \
      & \ghost{V_1}
      &
      & \ghost{V_1^\dagger}
      & \qw \\
      & \ghost{V_1}
      &
      & \qw
      & \ \cdots \
      & \ghost{V_1}
      &
      & \qw
      & \qw \\
      & \ghost{V_1}
      &
      & \qw
      & \ \cdots \
      & \gate[style=pulse]{V_4^\dagger}
      &
      & \gate[style=pulse]{V_4}
      & \qw \\
      & \gate[style=pulse]{V_5}
      &
      & \gate[style=pulse]{V_5^{\dagger}}
      & \ \cdots \
      & \ghost{V_1^\dagger}
      &
      & \qw
      & \qw
  \end{quantikz}
   \approx 
   \begin{quantikz}[row sep=0.35cm, column sep=0.35cm]
      & \ghost{V_1}
      & \gate[5, style = target]{e^{-i H_T t}}
      & \ghost{V_1^\dagger}
      & \qw \\
      & \qw
      &
      & \ghost{V_1^\dagger}
      & \qw \\
      & \ghost{V_1}
      &
      & \qw
      & \qw \\
      & \ghost{V_1^\dagger}
      &
      & \qw
      & \qw \\
      & \ghost{V_1^\dagger}
      &
      & \ghost{V_1}
      & \qw
  \end{quantikz}\]
  \end{minipage}

  \caption{
    Schematic overview of the fermionic Hamiltonian engineering protocol.
    Given a description of a \emph{system} $H_S$ and a \emph{target} Hamiltonian $H_T$, the algorithm runs in $\poly(n)$ time and returns local pulses $V_\ind{b}$ and quantum evolution times $\lambda_\ind{b}$.
    This output is then used to construct the experimental implementation.
    Here the implementation is depicted using a quantum circuit for simplicity, where each line corresponds to a fermionic mode, but in general the framework is applicable to arbitrary geometries.
    The red boxes depict local unitaries from the set $\{\1,V_j,V_j^\dagger\}$, for mode $j$.
    The product of such local unitaries comprises a local pulse, shown here as vertical boxes with dashed borders. 
    Using these local pulses, we conjugate time evolution under the native system Hamiltonian.
    Applying the pulses at the right times simulates the target evolution.
    The accuracy of the simulation is governed by the standard Trotter error.
    }
  \label{fig:fhe_schematic}
\end{figure*}
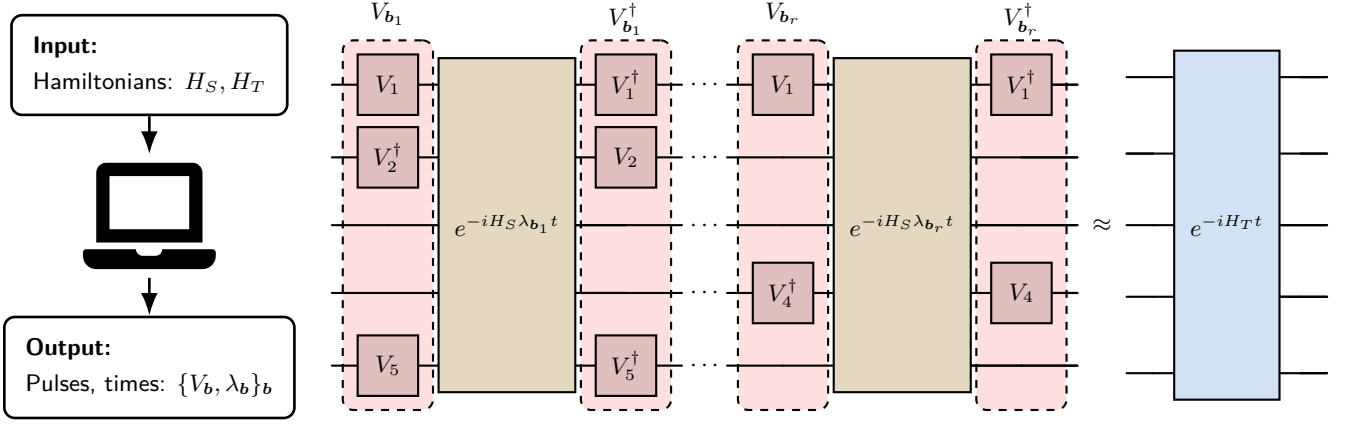

\section{Notation and preliminaries}
\label{sec:Preliminaries}
For a positive integer $n\in\ZZ_+$, we define $[n]\coloneqq \{1,2,\dots,n\}$.
We consider $n$ fermionic modes, which for concreteness can be thought of as different lattice sites in an optical lattice. 
The relevant Hilbert space is the fermionic Fock space $\mathcal{F}\coloneqq \bigoplus_{j=0}^n \mathcal{F}_j$, where $\mathcal{F}_j$ is the fully antisymmetric subspace of the $j$-particle tensor product space built from the single particle Hilbert space $\mathcal{H} \cong \mathbb{C}^n$, $\mathcal{F}_j \coloneqq \bigwedge^j \mathcal{H}\subset \mathcal{H}^{\otimes j}$, where $j\in\{0,1,\dots,n\}$. 
The operators $\Set{c_j}_{j\in [n]}$ and $\Set{\cdag_j}_{j\in [n]}$ are the annihilation and creation operators respectively, and satisfy the canonical anticommutation relations $\{c_j,c_k^\dagger\}=\delta_{jk}$, and $\{c_j,c_k\}= \{c_j^\dagger,c_k^\dagger\}=0$ for all $j,k\in [n]$. 
The number operator on the mode $j$ is defined by $n_j \coloneqq c_j^\dagger c_j$.
It can be checked straightforwardly that
\begin{equation}
    \label{eq:commutator}
  [n_k,c_j] = -\delta_{jk} c_k.
\end{equation} 

We will extensively use \emph{local unitaries} generated by the number operator, 
\begin{align}
	\label{eq:V_definition}
	V_j(\theta)\coloneqq \exp{\left(-\i \theta n_j\right)},
\end{align}
where $\theta \in [0,2\pi)$.
These operators are local in the sense that they act on single modes.

A direct calculation using the Campbell identity yields
\begin{align}
  \begin{split}
    \label{eq:conjugations}
  V_j^\dagger(\theta) c_k V_j(\theta) 
  &= 
  \e^{-\i\theta \delta_{jk}} c_k, 
  \\
  V_j^\dagger(\theta) \cdag_k V_j(\theta) 
  &= 
  \e^{\i\theta \delta_{jk}} \cdag_k, 
  \end{split}
\end{align}
see \cref{eq:Campbell} for a derivation.
In other words, under conjugation with these local
unitaries, the annihilation and creation operators acquire a phase $\e^{\pm\i\theta}$ if the operators act on the same mode. 
This simple fact provides the basic mechanism for effectively changing the coefficients of a Hamiltonian under time evolution.
By using it repeatedly, it allows simulating time evolution under a wide range of Hamiltonians as repeated applications of the local unitaries allow us to incrementally modify the phases acquired by the annihilation and creation operators, i.e., 
\begin{align}
  \begin{split}
    (V_j^\dagger(\theta))^m c_k (V_j(\theta))^m & =  V_j(m\theta)^\dagger c_k V_j(m\theta) \\
    & = \e^{-\i m\theta  \delta_{jk}} c_k.  
  \end{split}
\end{align}
Frequently, we will make use of multiple local unitaries applied in parallel, which we call a \emph{local pulse}
\begin{align}
  \label{eq:defn_pulse}
  V_\ind{b}(\theta) \coloneqq \prod_{j\in[n]} V_j^{b_j}(\theta),
\end{align}
where the index string $\ind{b} \in \ZZ^n$ indicates the modes on which the pulse acts and how many times the local operators (or sometimes which local operators) $V_j(\theta)$ are applied.

\section{Local tuning of tunnelling coefficients}
\label{sec:tunnelling}

Our framework grants independent control over each tunnelling coefficient, enabling simulation of interacting fermionic Hamiltonians with arbitrary complex-valued tunnelling coefficients.
The main idea is to exploit the interplay between native time evolution under a fixed system Hamiltonian and conjugation by local unitaries in order to modify effective tunnelling coefficients.
We first illustrate the method in the setting of quadratic Hamiltonians, where the effect of the conjugations can be analyzed transparently, before extending the construction to interacting Hamiltonians, focusing in particular on the Fermi--Hubbard model.
\subsection{Quadratic Hamiltonians}
\label{sec:quadratic}

Our Hamiltonian engineering method is best illustrated when considering Hamiltonians that are quadratic in the annihilation and creation operators, which are known to be efficiently simulable on classical computers \cite{bravyiLagrangianRepresentationFermionic2004,jozsaMatchgatesClassicalSimulation2008a}. 
Later in this section and more properly in \cref{sec:FHE_general}, we extend these results to non-quadratic Hamiltonians.

The goal is to simulate the time evolution under a \emph{target Hamiltonian} $H_T$. 
For now, we assume it to be a quadratic Hamiltonian on $n$ fermionic modes, i.e., to be of the form 
\begin{align}
  \label{eq:quadratic_target_Hamiltonian}
  H_T = \sum_{\{j,k\}\in E_T} \beta_{jk}  c_j^\dagger c_k + \mathrm{H. c.},
\end{align}
where $\beta_{jk}\in\CC$ and $E_T\subset \binom{[n]}{2}$ 
is the set of mode pairs with a non-vanishing tunnelling term, with $\binom{[n]}{2}\coloneqq~\{e~\subset~[n]~\mid~|e|=2\}$ and $+\mathrm{H.c.}$ denotes that we are adding the Hermitian conjugate of the terms left of it. 
For a quadratic Hamiltonian $H$ of the form \cref{eq:quadratic_target_Hamiltonian}, we call the undirected graph $G_T = ([n],E_T)$ the \emph{connectivity graph} 
of $H$. 
Thus, $|E_T|$ corresponds to the number of couplings between modes.
For instance, $E_T=\Set{\{i,i+1\}\given i\in [n-1]}$ corresponds to a 1D system consisting of $|E_T|=n-1$ nearest-neighbour hopping terms.

We assume we have access to a quantum simulator that
\begin{compactenum}[(i)]
\item realizes the time evolution under the \emph{system Hamiltonian}
\begin{align}
  \label{eq:quadratic_system_Hamiltonian}
  H_S = \alpha \sum_{\{j,k\}\in E_S}  c_j^\dagger c_k + \mathrm{H. c.},
\end{align}
defined on the connectivity graph $G_S = ([n],E_S)$, with $r=|E_S|$ coupling terms, where $\alpha\in \RR$ is the \emph{tunnelling constant}, and 
\item can additionally apply the single-mode unitaries $V_j\coloneqq V_j(2\pi/3)=\e^{-2\pi\i n_j/3 }$, as a special case of the local unitaries \eqref{eq:V_definition},
to any mode $j$ at any time. 
\end{compactenum}
The single-mode unitaries give us access to the local pulses $V_\ind{b}\coloneqq V_\ind{b}(2\pi/3)$ from \cref{eq:defn_pulse}.
To be able to generate arbitrary complex coefficients in the target Hamiltonian we have chosen the phase angle $2\pi/3$ corresponding to third roots of unity, on which we elaborate in \cref{sec:he_tritstring}.

Given a description of the system and target Hamiltonians, we efficiently compute the pulses and the evolution times for the simulation task. These parameters are then used to implement the target evolution, as depicted in \cref{fig:fhe_schematic}.
The classical algorithm used to compute the pulses follows from the following theorem.

\begin{theorem}[Simulation of quadratic Hamiltonians]
\label{thm:quadratic}
The time evolution under any quadratic target Hamiltonian \eqref{eq:quadratic_target_Hamiltonian}, whose connectivity graph $G_T$ is a subgraph of the connectivity graph of the system Hamiltonian, i.e.\ $E_T\subseteq E_S$, can be simulated by applying at most $2r=2|E_S|$ conjugations under local pulses $V_\ind{b}$ at suitably chosen times. 
\end{theorem}

Below we sketch the proof of this theorem, the full proof for a generalized class of Hamiltonians can be found in \cref{sec:FHE_general}.
A more pedagogical treatment based on simple examples is provided in \cref{sec:simple_examples_appendix}.

In order to simulate the time evolution $\e^{-\i H_T t}$, we decompose the target Hamiltonian as a sum of conjugations of the system Hamiltonian under pulses, 
\begin{align}
  \label{eq:decomposition}
  H_T = \sum_{\ind{b}\in\mathcal{B}}\lambda_{\ind{b}}V_\ind{b}^\dagger H_S V_\ind{b},
\end{align} 
with nonnegative coefficients $\lambda_\ind{b}\in\RR_{\geq 0}$, where $\mathcal{B} \subset \ZZ^n$ determines the set of pulses we consider and we set $c\coloneqq \abs{\mc B}$.
Such a decomposition can always be found as long as the connectivity graph of the target Hamiltonian is a subgraph of the system Hamiltonian.
This decomposition allows us to approximate the time evolution under $H_T$ by the product
\begin{align}
    \label{eq:product}
    U_\mathrm{eng}(t) = \prod_{\ind{b}\in\mathcal{B}} V_\ind{b}^\dagger \e^{-\i \lambda_{\ind{b}}H_S t} V_\ind{b}  \approx \e^{-\i H_T t},
  \end{align}
where the approximation is controlled via Trotterization \cite{trotterProductSemiGroupsOperators1959,childsTheoryTrotterError2021,dalzellQuantumAlgorithmsSurvey2023}. 
In the rest of the paper, we use the second-order Trotter formula \eqref{eq:engineered_evolution_operator} to construct the engineered time evolution operator \eqref{eq:product} for all numerical applications.

The time evolution $U_\mathrm{eng}(t)$ can now be implemented by the quantum simulator, where the nonnegative coefficients $\lambda_\ind{b}$ correspond to the duration of the \emph{free evolution} under the system Hamiltonian between pulses, indicating the times at which the local pulses need to be applied.
This finishes the sketch for the proof of \cref{thm:quadratic}.

The same decomposition \eqref{eq:decomposition}, allows us to construct an efficient classical algorithm for computing the pulses and evolution times necessary for a given target evolution.
We do this by expressing the condition \cref{eq:decomposition} as a matrix equation $F\boldsymbol{\lambda} = \boldsymbol{\beta}/\alpha$, where the \emph{constraint matrix} $F\in\CC^{r\times c}$ captures the effect of all possible conjugations on the system Hamiltonian with $c = |\mathcal{B}|$ pulses, $\boldsymbol{\lambda}\in\RR_{\geq 0}^c$ is the vector of all free evolution times, and $\boldsymbol{\beta}\in \CC^r$ is the vector of target Hamiltonian coefficients.
This allows us to phrase the Hamiltonian engineering problem as a \ac{LP}
\begin{align}
\begin{split}
  \label{eq:LP_quadratic}
  \text{minimize} \quad & \boldsymbol{1}^\top \boldsymbol{\lambda} \\
  \text{subject to} \quad &  F \boldsymbol{\lambda} = \boldsymbol{\beta}/\alpha, \\
                          &  \boldsymbol{\lambda} \in \RR_{\geq 0}^{c},
  \end{split}
\end{align}
where $\boldsymbol{1}\in\RR^c$ is the vector of all ones, and  we minimize the total evolution time $\boldsymbol{1}^\top \boldsymbol{\lambda} = \sum_{\ind{b}\in\mathcal{B}}\lambda_\ind{b} = \norm{\ind{\lambda}}_{\ell_1}$, such that the optimal solution corresponds to the shortest quantum runtime.
In \cref{sec:FHE_general}, we show that this \ac{LP} has an optimal solution for all $\ind{\beta}$, and argue that we can efficiently find a set of pulses which allows us to simulate any Hamiltonian defined on the same connectivity graph.

Many physically relevant Hamiltonians, such as the Harper--Hofstadter Hamiltonian, exhibit translation symmetry.
If the system Hamiltonian shares this symmetry, it can be exploited to reduce the cost of computing the decomposition \eqref{eq:decomposition}.
The idea is to choose a repeating cell of the lattice, solve the \ac{LP} \eqref{eq:LP_quadratic} on this cell with periodic boundary conditions (which can be identified with a torus) and extend the resulting pulse patterns to the full lattice, while retaining their weights.
For a fixed cell, the classical optimisation cost is independent of the total number of modes $n$, and the extended sequence has the same pulse count and evolution time as the cell solution.
We describe the details of the construction in \cref{sec:symmetry_reduction} and illustrate the approach by computing an optimal set of pulses for the Harper--Hofstadter model in \cref{sec:harper_hofstadter}.

In the rest of this section, we consider concrete applications of our method to quadratic Hamiltonians.
Note that all linear programs in this work are modelled with the CVXPY Python package \cite{agrawalRewritingSystemConvex2019,diamondCVXPYPythonEmbeddedModeling2016}, and solved with the MOSEK solver \cite{MOSEKOptimizerAPI}.
Numerical simulations of non-interacting systems are performed using free fermionic simulation techniques (see \cref{sec:Harper--Hofstadter_appendix}), and interacting systems are simulated using exact diagonalization and the QuSpin Python package \cite{weinbergQuSpinPythonPackage2017,weinbergQuSpinPythonPackage2019}.
The code used for the simulations is available at \cite{kum_fhe_code_2026}.

\subsubsection{Harper--Hofstadter model} 
\label{sec:harper_hofstadter}

\begin{figure*}[t]
  \centering
  \begin{minipage}{0.39\textwidth}
    \includegraphics[width = \columnwidth]{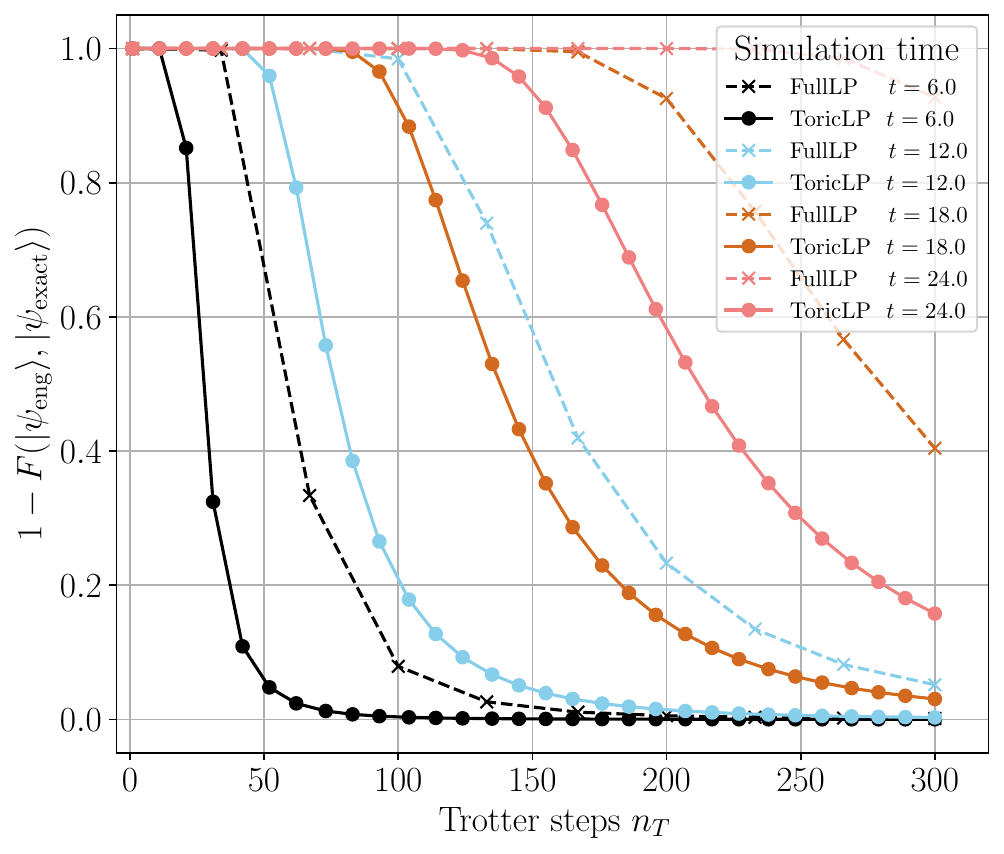}
    
    (a)
  \end{minipage}
  \begin{minipage}{0.6\textwidth}
    \includegraphics[width = \columnwidth]{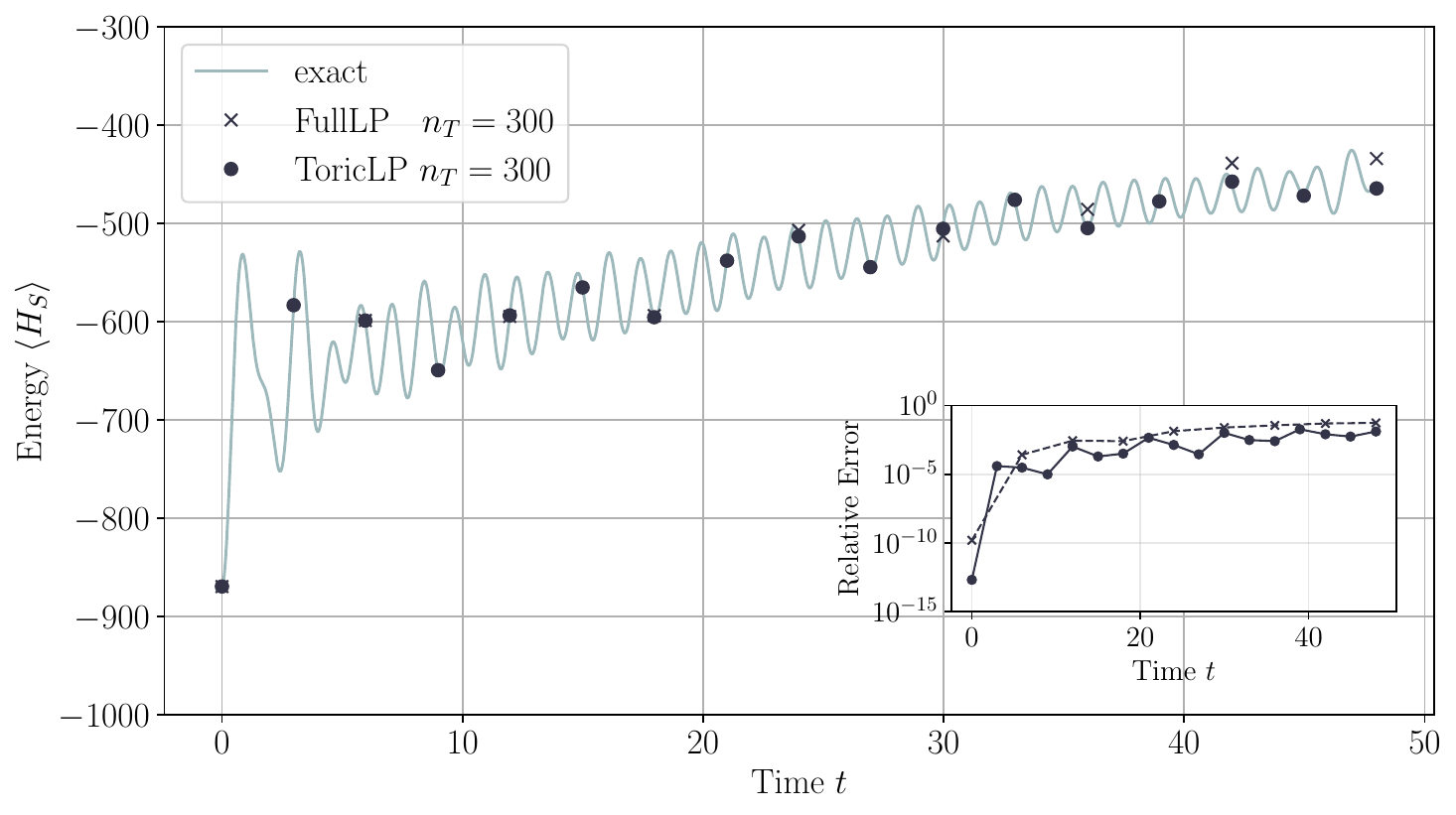}

    (b)
  \end{minipage}
\caption{
  Numerical simulation of the Harper--Hofstadter dynamics on a $32\times34$ lattice with $n=1088$ modes at half-filling $N=544$. (a) Infidelity $1-F(\ket{\psi_\mathrm{eng}},\ket{\psi_\mathrm{exact}})$ between the engineered state $\ket{\psi_\mathrm{eng}}$ and the exact target state $\ket{\psi_\mathrm{exact}}$ as a function of the number of Trotter steps. 
  The dots and solid lines correspond to the data obtained from extended toric ($\mathrm{ToricLP}$) pulses, and crosses and dashed lines correspond to data obtained from time evolution using pulses given by the \ac{LP} on the full lattice using efficient relaxation ($\mathrm{FullLP}$), see \cref{sec:relaxation}.
  Both engineered dynamics converge to the exact evolution for a wide range of simulation times as the number of Trotter steps is increased.
  We observe that the dynamics generated by the $\mathrm{ToricLP}$ pulses converge quicker than the dynamics generated by the LP pulses, as is expected from the bounds on the Trotter formula \eqref{eq:trotter_bound}.
  (b) The energy $\langle H_S\rangle$ of the time evolved state simulated exactly using direct exponentiation of the target Hamiltonian, and using FHE with $\mathrm{FullLP}$ (crosses) and $\mathrm{ToricLP}$ (dots) pulses, and with $n_T = 300$ Trotter steps.
  The inset plot shows the evolution of the relative error $|\langle H_S\rangle_\mathrm{exact}-\langle H_S\rangle_\mathrm{eng}|/|\langle H_S\rangle_\mathrm{exact}|$.
  We observe that the engineered dynamics track the exact evolution under $H_T$ throughout the simulated window.
}
  
  \label{fig:HH_fidelity_energy}
\end{figure*}

\begin{figure}[t]
    \centering
    \begin{tikzpicture}[>=stealth, line cap=round, line join=round]

        \begin{scope}[xshift=0.5cm, yshift=0cm, scale=1.0]
            \foreach \x in {0,1} {
                \foreach \y in {0,1,2,3} {
                    \draw (\x,\y) -- (\x+1,\y);
                }
            }
            \foreach \x in {0,1,2} {
                \foreach \y in {0,1,2} {
                    \draw (\x,\y) -- (\x,\y+1);
                }
            }

            \foreach \x in {0,1,2} {
                \foreach \y in {0,1,2,3} {
                    \fill (\x,\y) circle (1.7pt);
                }
            }
        \end{scope}

        \draw[->, very thick] (3.2,1.5) -- (4.8,1.5);
        \node at (4.0,1.8) {\large FHE};

        \begin{scope}[xshift=5.5cm, yshift=0cm, scale=1.0]
            \foreach \x in {0,1} {
                \foreach \y in {0,1,2,3} {
                    \draw (\x,\y) -- (\x+1,\y);
                }
            }
            \foreach \x in {0,1,2} {
                \foreach \y in {0,1,2} {
                    \draw (\x,\y) -- (\x,\y+1);
                }
            }

            \foreach \x in {0,1,2} {
                \foreach \y in {0,1,2,3} {
                    \fill (\x,\y) circle (1.7pt);
                }
            }

            \foreach \x/\y in {0.5/0.5, 1.5/0.5, 0.5/1.5, 1.5/1.5, 0.5/2.5, 1.5/2.5} {
                \draw[->, thick, fluxblue]
                    (\x,\y+0.24)
                    arc[start angle=90, end angle=-180, radius=0.28];
                \node at (\x,\y) {$\Phi$};
            }
        \end{scope}

    \end{tikzpicture}
    \caption{Fermionic Hamiltonian engineering (FHE) maps a trivial nearest-neighbour lattice Hamiltonian to a Harper--Hofstadter-type model with flux \(\Phi=2\pi/3\) per plaquette.}
    \label{fig:fhe_flux_cartoon}
\end{figure}

\begin{figure*}[t]
  \centering
  \includegraphics[width = \textwidth]{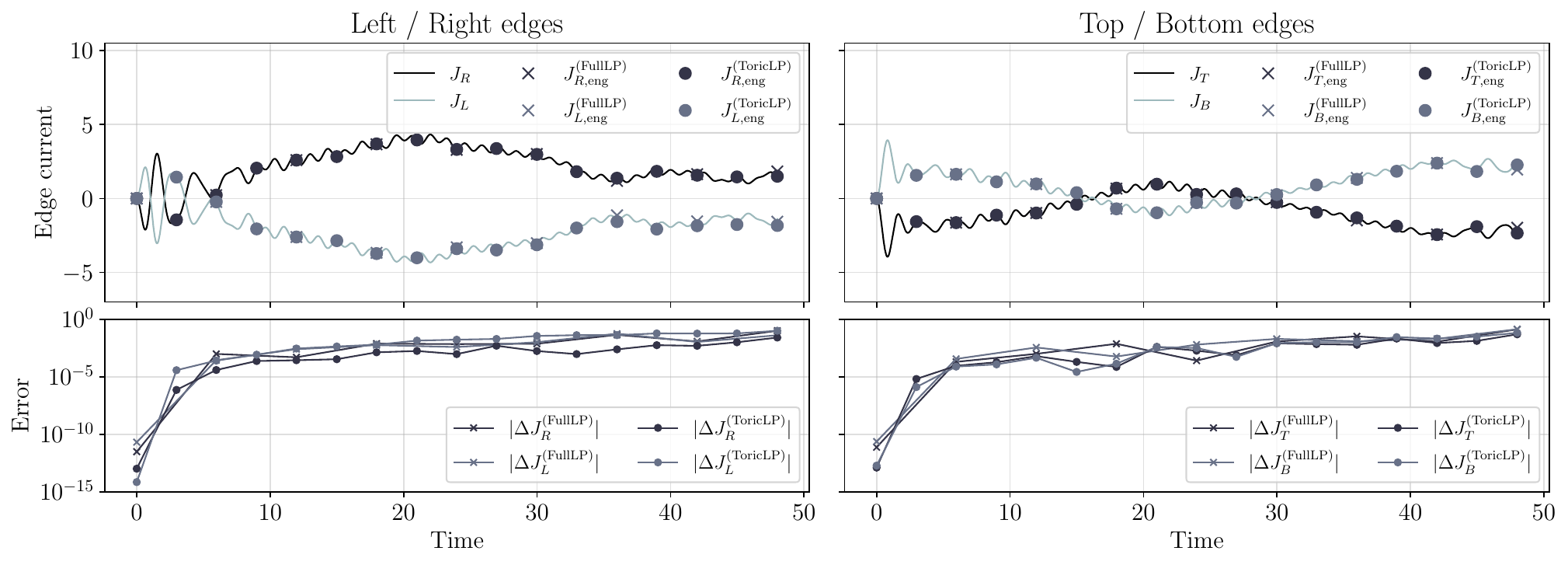}
\caption{
  Numerical simulation of the edge currents for the Harper--Hofstadter model on a $32\times34$ lattice with 1088 modes at half-filling $N=544$.
  The top plots show the edge currents for the four edges of the lattice as functions of time, where the solid lines correspond to exact evolution obtained by directly exponentiating the target Hamiltonian, and the dots and crosses correspond to the observables obtained through the engineered evolution using $\mathrm{ToricLP}$ and $\mathrm{FullLP}$ pulses with $n_T=300$ Trotter steps, respectively.
  As expected, both the left-right and the top-bottom edges are seen to start with no edge current but then exhibit the currents of the Harper--Hofstadter model.
  The bottom plots 
  show the error $|\Delta J_E| \coloneqq |J_E - J_{E,\mathrm{eng}}|/\mathrm{max}_t\left|(J_E(t))\right|$, for $E\in\{R,L,T,B\}$ for the right, left, top, and bottom edge of the lattice.
  The size of the errors indicates that the FHE method faithfully reproduces the target dynamics.
}
  \label{fig:edge_currents}
\end{figure*}

As a first example, we demonstrate that the protocol can be used to engineer artificial gauge fields by simulating time evolution under the Harper--Hofstadter model. 
We consider a quadratic system Hamiltonian defined on a two-dimensional $32\times 34$ square lattice with $n=1088$ sites, $r=2110$ edges, at half-filling $N=544$, with uniform nearest-neighbour tunnelling amplitudes
\begin{align}
  \label{eq:HH_system_hamiltonian}
  H_S = -J\sum_{x,y} \left(c_{x+1,y}^\dagger c_{x,y} + c_{x,y+1}^\dagger c_{x,y} + \mathrm{H.c.}\right).
\end{align}
Using the Hamiltonian engineering protocol, we simulate time evolution under the Harper--Hofstadter Hamiltonian, given in the Landau gauge by
\begin{align}
  \label{eq:HH_target_hamiltonian}
  H_T = -J\sum_{x,y} \left(\e^{\i 2\pi y/3} c_{x+1,y}^\dagger c_{x,y} + c_{x,y+1}^\dagger c_{x,y} + \mathrm{H.c.}\right),
\end{align}
with a uniform magnetic flux of $\Phi~=~2\pi/3$ through each plaquette of the lattice \cite{cooperTopologicalBandsUltracold2019}, as in \cref{fig:fhe_flux_cartoon}.

We prepare the ground state $\ket{\psi_0}$ of the system Hamiltonian \eqref{eq:HH_system_hamiltonian}, and simulate time evolution under the target Hamiltonian \eqref{eq:HH_target_hamiltonian}. 
This corresponds to a sudden change of the Hamiltonian governing the system and is referred to as a \emph{quantum quench}.
In general, such quenches are a natural fit for our framework.
An experimentally feasible state is prepared as an eigenstate of the system Hamiltonian (typically the ground state), which is by construction, the initial state for evolution under the engineered target Hamiltonian. 
The switch to the target evolution is then implemented with no intermediate ramp, the moment the first pulse in the sequence is applied. 

The first step of the simulation task is to solve the \ac{LP} \eqref{eq:LP_quadratic}, which takes the system and target Hamiltonians as input and returns the pulses and evolution times.
Averaged over 100 random instances of the LP with $s = 6r =12660$ columns obtained via efficient relaxation (\cref{sec:relaxation}), the classical solver took 171 seconds on a laptop, while the resulting quantum runtime $\lVert\lambda\rVert_{\ell_1}$ is $80.34 \pm 0.17$ in units of $1/J$.
Note that this enables the implementation of $U_\mathrm{exact}(t=1)=\exp{(-\i H_T)}$, and the quantum runtime required to simulate $U_\mathrm{exact}(t)$ is $t \lVert\lambda\rVert_{\ell_1}$. 
We find that the \ac{LP} finds a decomposition of \cref{eq:HH_target_hamiltonian} using 4219 pulses on average over the same sample set of randomized instances.
Note that this almost saturates the $2r = 4220$ pulses mentioned in \cref{thm:quadratic}.
Using the second-order Trotter formula \eqref{eq:engineered_evolution_operator} with $n_T=300$ Trotter steps this translates to about $2\times 2r \times n_T \approx 2.5\times 10^6$ pulses to implement the evolution. 
For the quantum resource requirements of a more general fermionic Hamiltonian engineering task, see \cref{sec:quantum_resource_requirements}.

The quantum runtime we have observed above can be reduced significantly, without changing the number of pulses that need to be implemented.
This is achieved by increasing the number of available pulses the \ac{LP} can choose from (see \cref{sec:relaxation}).
Crucially, after the \ac{LP} is run and a decomposition of the target Hamiltonian \eqref{eq:HH_target_hamiltonian} with a suitable quantum runtime is obtained, one does not need to perform the classical precomputation again, and the experiment can be repeated as many times as the experimenter desires.
Therefore, the classical cost of running the \ac{LP} for longer is usually a valuable tradeoff for better quantum runtimes.

That said, both the number of pulses required and the quantum runtime can be improved further by a more educated choice of pulses.
For example, for some Hamiltonian engineering tasks it is possible to exploit the lattice translation symmetry of the problem. 
The idea is to find a small repeating cell which can be used to tile the full lattice, solve a version of the \ac{LP} with \acp{PBC} on the small cell (identified with a torus), and extend the solution to the full lattice.
The \acp{PBC} imposed on the small lattice already take care of how neighbouring cells are glued to each other.
For the Harper--Hofstadter model the smallest cell one can use is a $3\times 3$ block, for which without loss of generality we pick the bottom left block such that the bottom left vertex is labelled by $(x,y) = (0,0)$.
With these PBCs the ``toric'' LP defined on this block has an $18\times6561$ constraint matrix ($18$ edges, the $3^9$ pulses in a cell reduce to $3^8=6561$ after removing the global-shift redundancy mentioned in \cref{sec:symmetry_reduction}), which is solved to yield an optimal solution with $6$ pulses $\{\ind{b}^{(l)}\}_{l\in[6]}$ where $\|\ind{\lambda}\|{\ell_1} = 2$ and $\lambda_{\ind{b}^{(l)}} = 1/3$ for all $l\in [6]$, without using efficient relaxation. 
We then construct an optimal solution for the full lattice by extending the pulses

\begin{align}
  \label{eq:extended_pulses}
  \tilde{b}^{(l)}(x,y) =  b^{(l)}(x\bmod 3,\ y\bmod 3). 
\end{align} 

For more details regarding the use of translation symmetry to reduce the cost of obtaining the pulses, see \cref{sec:symmetry_reduction}, for a proof of the optimality of the quantum runtime see \cref{sec:HH_optimal_pulses_appendix}.
With the same $n_T=300$ second-order Trotter steps, the total number of pulses required for this implementation is at most $2\times6\times n_T = 3600$, about three orders of magnitude fewer than the number we require for the pulses derived using the relaxed \ac{LP}.
Crucially, the effect of the pulses \eqref{eq:extended_pulses} is independent of the lattice size allowing the method to scale to even larger systems, with no extra overhead in the number of pulses required or the quantum runtime.

We benchmark both engineered dynamics $U_\mathrm{eng}^{\mathrm{(FullLP)}}(t)$  and $U_\mathrm{eng}^{\mathrm{(ToricLP)}}(t)$, obtained via the second-order Trotter formula \eqref{eq:engineered_evolution_operator} applied to the decomposition given by the \ac{LP} on the full lattice and the \ac{LP} defined on the torus obtained from the repeating cell respectively, against the exact evolution $U_\mathrm{exact}(t)$, obtained by a direct exponentiation of $H_T$. 
We quantify the performance of the simulation using the 
state fidelity
\begin{align}
  \label{eq:fidelity}
  F(\ket{\psi_\mathrm{eng}},\ket{\psi_\mathrm{exact}})\coloneqq|\bra{\psi_0}U^\dagger_{\mathrm{eng}}(t)U_\mathrm{exact}(t)\ket{\psi_0}|^2,
\end{align}
and observe that the engineered state converges to the target dynamics also for long simulation times for a suitably chosen number of Trotter steps, see \cref{fig:HH_fidelity_energy}~(a).

Fidelity provides a useful, but stringent, measure of simulation accuracy, as it compares the full many-body states and is therefore sensitive to error components that have little effect on the observables of interest.
In practice, physically relevant observables often converge substantially faster than the state fidelity and ultimately determine the utility of the simulation. 
We therefore complement the fidelity analysis with the expectation value of the system Hamiltonian, $\langle H_S\rangle(t)$, and find that the engineered dynamics closely track the exact dynamics throughout the simulated time window, as shown in \cref{fig:HH_fidelity_energy}~(b).

Another natural probe for the Harper--Hofstadter Hamiltonian is the edge currents, defined to be the sum of bond currents on a given edge of the lattice.
The current operator associated with a bond $e = \left((x,y),(x',y')\right)$ is given by
\begin{align}
  J_{e} = -\i \left(\beta_{e} c_{(x',y')}^\dagger c_{(x,y)} - \beta_{e}^* c_{(x,y)}^\dagger c_{(x',y')}\right),
\end{align}
where $\beta_{e}$ denotes the tunnelling coefficient associated with the bond.
For example, the edge current for the left edge is given by
\begin{align}
  J_L = \sum_{y=0}^{L_y-2} J_{(0,y),(0,y+1)},
\end{align}
where $L_y$ is the length of the lattice in the $y$-direction.

We compare the edge currents under exact target evolution, obtained by direct exponentiation of the target Hamiltonian, with those produced by our engineered dynamics and find close agreement, as can be seen in \cref{fig:edge_currents}.
The edge currents that are absent in the native lattice, are faithfully reproduced by the engineered dynamics.
An artificial gauge field of the kind ordinarily realized through laser-assisted tunnelling or periodic lattice modulation \cite{aidelsburgerRealizationHofstadterHamiltonian2013} is thus generated on a 1088-mode lattice, with neither continuous periodic driving nor any modification of the native couplings, constrained only by the connectivity of $H_S$.
See \cref{sec:Harper--Hofstadter_appendix} for more details on the simulation techniques used to classically simulate the Harper--Hofstadter model and \cref{sec:HH_optimal_pulses_appendix} for an explicit set of optimal pulses.

\subsubsection{Artificial gauge fields on a triangular lattice}
\label{sec:triangular}

\begin{figure*}[t]
    \centering
    \begin{minipage}{0.40\textwidth}
      \centering
      (a)
      \begin{tikzpicture}[>={Stealth[length=3mm, width=2mm]}, line cap=round, line join=round, scale=1.8]

        \coordinate (BL) at (0,0);
        \coordinate (BR) at (1,0);
        \coordinate (TM) at (0.5, 0.866);

        \draw[gray!50] (BL) -- ++(-0.6, 0);
        \draw[gray!50] (BR) -- ++(0.6, 0);
        \draw[gray!50] (BL) -- ++(-0.3, -0.520);
        \draw[gray!50] (BR) -- ++(0.3, -0.520);
        \draw[gray!50] (BL) -- ++(0.3, -0.520);
        \draw[gray!50] (BR) -- ++(-0.3, -0.520);
        \draw[gray!50] (TM) -- ++(-0.3, 0.520);
        \draw[gray!50] (TM) -- ++(0.3, 0.520);
        \draw[gray!50] (TM) -- ++(-0.6, 0);
        \draw[gray!50] (TM) -- ++(0.6, 0);

        \draw[thick] (BL) -- (BR);
        \draw[thick] (BL) -- (TM);
        \draw[thick] (BR) -- (TM);

        \draw[->, thick,fluxblue]
            (BL) to[bend right=35]
            node[midway, below=2pt] {$c^\dagger_{x+1,y} c_{x,y}$}
            (BR);

        \draw[->, thick,fluxblue]
            (BL) to[bend left=35]
            node[midway, above left=-2pt] {$c^\dagger_{x,y+1} c_{x,y}$}
            (TM);

        \draw[->, thick,fluxblue]
            (BR) to[bend right=35]
            node[midway, above right=-2pt] {$c^\dagger_{x-1,y+1} c_{x,y}$}
            (TM);

        \fill (BL) circle (1.2pt);
        \fill (BR) circle (1.2pt);
        \fill (TM) circle (1.2pt);

    \end{tikzpicture}
    \end{minipage}
    \hfill
    \begin{minipage}{0.55 \textwidth}
      \centering
      (b)
      \begin{tikzpicture}[>={Stealth[length=3mm, width=2mm]}, line cap=round, line join=round, scale=1.8]

        \begin{scope}[xshift=-0.2cm]
            \coordinate (BL) at (0,0);
            \coordinate (BR) at (1,0);
            \coordinate (TM) at (0.5, 0.866);

            \draw[gray!50] (BL) -- ++(-0.3, 0);
            \draw[gray!50] (BR) -- ++(0.3, 0);
            \draw[gray!50] (BL) -- ++(-0.15, -0.26);
            \draw[gray!50] (BR) -- ++(0.15, -0.26);
            \draw[gray!50] (BL) -- ++(0.15, -0.26);
            \draw[gray!50] (BR) -- ++(-0.15, -0.26);
            \draw[gray!50] (TM) -- ++(-0.15, 0.26);
            \draw[gray!50] (TM) -- ++(0.15, 0.26);
            \draw[gray!50] (TM) -- ++(-0.3, 0);
            \draw[gray!50] (TM) -- ++(0.3, 0);

            \draw[thick] (BL) -- (BR);
            \draw[thick] (BL) -- (TM);
            \draw[thick] (BR) -- (TM);

        \draw[->, thick,fluxblue]
            (BL) to[bend right=35]
            node[midway, below=2pt] {$-J$}
            (BR);

        \draw[->, thick,fluxblue]
            (BL) to[bend left=35]
            node[midway, above left=-2pt] {$-J$}
            (TM);

        \draw[->, thick,fluxblue]
            (BR) to[bend right=35]
            node[midway, above right=-2pt] {$-J$}
            (TM);

            \fill (BL) circle (1.2pt);
            \fill (BR) circle (1.2pt);
            \fill (TM) circle (1.2pt);
        \end{scope}

        \draw[->, very thick] (1.55,0.433) -- (2.05,0.433);
        \node at (1.8,0.2) {\small FHE};

        \begin{scope}[xshift=2.7cm]
            \coordinate (BL) at (0,0);
            \coordinate (BR) at (1,0);
            \coordinate (TM) at (0.5, 0.866);

            \draw[gray!50] (BL) -- ++(-0.3, 0);
            \draw[gray!50] (BR) -- ++(0.3, 0);
            \draw[gray!50] (BL) -- ++(-0.15, -0.26);
            \draw[gray!50] (BR) -- ++(0.15, -0.26);
            \draw[gray!50] (BL) -- ++(0.15, -0.26);
            \draw[gray!50] (BR) -- ++(-0.15, -0.26);
            \draw[gray!50] (TM) -- ++(-0.15, 0.26);
            \draw[gray!50] (TM) -- ++(0.15, 0.26);
            \draw[gray!50] (TM) -- ++(-0.3, 0);
            \draw[gray!50] (TM) -- ++(0.3, 0);

        \draw[->, thick,fluxblue]
            (BL) to[bend right=35]
            node[midway, below=2pt] {$-J\e^{\i\pi/2}$}
            (BR);

        \draw[->, thick,fluxblue]
            (BL) to[bend left=35]
            node[midway, above left=-2pt] {$-J\e^{\i\pi}$}
            (TM);

        \draw[->, thick,fluxblue]
            (BR) to[bend right=35]
            node[midway, above right=-2pt] {$-J\e^{-\i\pi/2}$}
            (TM);

            \draw[thick] (BL) -- (BR);
            \draw[thick] (BL) -- (TM);
            \draw[thick] (BR) -- (TM);
            \node[fluxblue] at (0.5, 0.289) { $\Phi = \pi$};

            \fill (BL) circle (1.2pt);
            \fill (BR) circle (1.2pt);
            \fill (TM) circle (1.2pt);
        \end{scope}
    \end{tikzpicture}
    \end{minipage}

    \caption{(a) The three nearest-neighbour tunnelling operators on a triangular lattice plaquette, with each arrow indicating the direction of hopping for one term in $H_S$. (b) Fermionic Hamiltonian engineering on a triangular plaquette. The system Hamiltonian has uniform tunnelling coefficients (left), which the protocol maps to engineered coefficients with phases $\e^{\i\pi/2}$, $\e^{\i\pi}$, $\e^{-\i\pi/2}$ on the three bonds, producing a $\pi$-flux per plaquette (right). Labels indicate the coefficient of the directed hopping along each arrow.}
    \label{fig:fhe_triangular}
\end{figure*}
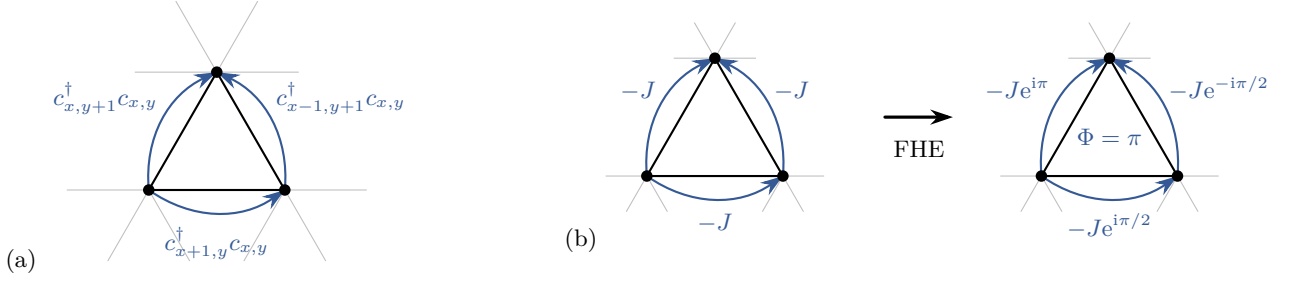

As we saw in the previous section, exploiting the structure of the target can reduce the \ac{LP} to a problem of constant size. Sometimes it is even possible to bypass the \ac{LP} altogether and construct the required pulse sequence analytically.
We exemplify this by calculating the pulses necessary to engineer artificial gauge fields on a triangular lattice with a degenerate ground state \cite{struckEngineeringIsingXYSpin2013}.
We consider a quadratic system Hamiltonian on a triangular lattice with uniform nearest-neighbour tunnelling coefficients, where the tunnelling terms are defined as in \cref{fig:fhe_triangular} (a),
\begin{align}
\begin{split}
  H_S = -J\sum_{\langle x,y \rangle} \Big(&
  c_{x+1,y}^\dagger c_{x,y}
  + c_{x,y+1}^\dagger c_{x,y}\\
  &+ c_{x-1,y+1}^\dagger c_{x,y}
  + \mathrm{H.c.}
  \Big),
\end{split}
\end{align}
and introduce the following set of local pulses
\begin{align}
  V_1 &= \prod_{x,y} \exp\!\left(-\i\frac{\pi}{2} x\, n_{x,y}\right),\\
  V_2 &= \prod_{x,y} \exp\!\left(-\i\pi y\, n_{x,y}\right),\\
  V_3 &= \prod_{x,y} \exp\!\left(-\i\pi (x+y)\, n_{x,y}\right).
\end{align}
Applying these pulses and summing the resulting conjugated Hamiltonians with $\lambda_1=\lambda_2=\lambda_3=1$ yields the effective Hamiltonian
\begin{align}
\begin{split}
H_{\mathrm{eff}}
&=\sum_{j=1}^3 V_j^\dagger H_S V_j \\
&= -J\sum_{\langle x,y \rangle} \Big(
\e^{\i\pi/2} c_{x+1,y}^\dagger c_{x,y}
+ \e^{\i\pi} c_{x,y+1}^\dagger c_{x,y} \\
&\hspace{2.2cm}
+ \e^{-\i\pi/2}c_{x-1,y+1}^\dagger c_{x,y}
+ \mathrm{H.c.}
\Big), 
\end{split}
\end{align}
which corresponds to a triangular lattice with $\pi$ fluxes in each plaquette, see \cref{fig:fhe_triangular} (b),
\begin{align}
  \Phi = \pi.
\end{align}
The dispersion of the effective Hamiltonian
\begin{align}
E_k = -2J\left[\sin(k_x) - \cos(k_y)  + \sin(k_x-k_y)\right]
\end{align} 
has two global minima, reflecting the frustration of the lattice model, demonstrating that nontrivial gauge structures can be engineered analytically, without recourse to numerical optimization, for more details see \cref{sec:triangular_appendix}.

\subsection{Interacting fermions and the Fermi--Hubbard model}
\label{sec:interactions}

We now extend our analysis to interacting Hamiltonians, and 
focus on interaction terms that are polynomials of the number operators,
i.e., of the form
\begin{equation}
I = \sum_{\emptyset \subsetneq \Omega \subseteq [n]} g_\Omega N_\Omega,
\qquad N_\Omega \coloneqq \prod_{j \in \Omega} n_j. 
\end{equation}
The resulting class of Hamiltonians includes the on-site Hubbard interaction and density-density interactions of arbitrary range. 
The local unitaries $V_j(\theta) = \e^{-i\theta n_j}$ commute with every $N_\Omega$, so the coefficients $g_\Omega$ are unchanged under conjugation by any pulse $V_\ind{b}$.
Consequently, when our engineering protocol is applied to a Hamiltonian that contains both quadratic and interaction terms, the quadratic part is modified as before, while the interaction part is only rescaled by a positive global factor.
This restriction is not an artifact of our particular choice of pulses: in \cref{sec:FHE_general} we elaborate on the fact that no set of local unitaries locally modifying the $n_j$ can do more than globally rescale interaction terms of this form.

Despite this limitation, our framework can still simulate a broad class of physically relevant interacting Hamiltonians, including those combining artificial gauge fields with on-site or density-density interactions, such as the Hofstadter--Hubbard model.
The key observation is that the global rescaling of the interaction term is itself a tunable knob in some experimental platforms, therefore the global scaling introduced by our engineering framework can be accounted for by a suitable choice of the interaction strength in the system Hamiltonian.
More precisely, if our protocol effectively results in an overall factor $\gamma > 0$ on the interaction term, we may simply set the system interaction strength to $U_S = U_T / \gamma$ to recover the desired target value $U_T$.
This approach requires a platform that allows the bare interaction strength to be tuned globally, as in optical lattices employing Feshbach resonances to control on-site interactions \cite{chinFeshbachResonancesUltracold2010a}.
Within this setting, our framework realizes interacting Hamiltonians with arbitrary complex tunnelling coefficients, bringing models such as the Hofstadter--Hubbard model within reach of programmable analogue simulation.

To see how this works, consider a system Hamiltonian of the form
\begin{align}
\label{eq:int_system_Hamiltonian}
H_S = K_S + I_S,
\end{align}
where
\begin{align}
\begin{split}
K_S &= \alpha \sum_{\{j,k\}\in E} c_j^\dagger c_k + \mathrm{H.c.}, \\
I_S &= U_S\sum_\Omega N_\Omega,
\end{split}
\end{align}
with a constant tunnelling parameter $\alpha \in \mathbb{R}$ and a globally tunable interaction strength $U_S\in \RR$.
Suppose we have identified a set of pulses $\{V_\ind{b}\}_{\ind{b}}$ and corresponding evolution times $\lambda_\ind{b}$ such that the quadratic tunnelling Hamiltonian $K_S$ is mapped to an engineered quadratic Hamiltonian
\begin{align}
K_T = \sum_{\{j,k\}\in E} \beta_{jk} c_j^\dagger c_k + \mathrm{H.c.}
\end{align}
This implies that the effective interaction term takes the form
\begin{align}
I_\mathrm{eff}
= \sum_{\ind{b}} \lambda_{\ind{b}}\, V_\ind{b}^\dagger I_S V_\ind{b}
= \|\boldsymbol{\lambda}\|_{\ell_1} I_S.
\end{align}
Applying this set of pulses to the full system Hamiltonian~\eqref{eq:int_system_Hamiltonian} therefore yields the effective Hamiltonian
\begin{align}
H_{\mathrm{eff}} = K_T + \|\boldsymbol{\lambda}\|_{\ell_1} I_S.
\end{align}
Hence, the interaction strength in the engineered model is rescaled by a global factor $\|\boldsymbol{\lambda}\|_{\ell_1}$, while the structure of the interaction term itself remains unchanged.
Therefore, by choosing the interaction strength of the system Hamiltonian suitably, i.e. $U_S = U_T/\|\boldsymbol{\lambda}\|_{\ell_1}$, we can simulate time evolution under a target Hamiltonian of the form
\begin{align}
  \label{eq:int_target_Hamiltonian}
  H_T = K_T + U_T \sum_\Omega N_\Omega.
\end{align}
In the remainder of this section, we present two applications of this method to interacting Hamiltonians.

\subsubsection{Tuning \texorpdfstring{$U/t$}{U/t} in the Fermi--Hubbard model}

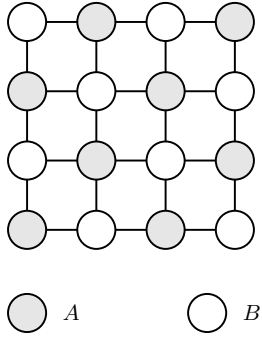
\begin{figure}[t]
\centering
\resizebox{0.4\columnwidth}{!}{%
\begin{tikzpicture}[
  node/.style={circle, draw, thick, minimum size=5.5mm, inner sep=0pt},
  A/.style={node, fill=gray!20},
  B/.style={node, fill=white},
  edge/.style={thick}
]

  \foreach \x in {0,1,2,3} {
    \foreach \y in {0,1,2,3} {

      \ifnum\x<3
        \draw[edge] (\x,\y) -- (\the\numexpr\x+1\relax,\y);
      \fi
      \ifnum\y<3
        \draw[edge] (\x,\y) -- (\x,\the\numexpr\y+1\relax);
      \fi

      \pgfmathtruncatemacro{\parity}{mod(\x+\y,2)}
      \ifnum\parity=0
        \node[A] at (\x,\y) {};
      \else
        \node[B] at (\x,\y) {};
      \fi
    }
  }

  \begin{scope}[shift={(0,-1.2)}]
    \node[A] at (0,0) {};
    \node[anchor=west] at (0.4,0) {$A$};
    \node[B] at (2.6,0) {};
    \node[anchor=west] at (3.0,0) {$B$};
  \end{scope}

\end{tikzpicture}}
\caption{A $4\times 4$ square lattice with a bipartition into $A$ and $B$ sublattices.}
\label{fig:2D_lattice}
\end{figure}
\begin{figure*}[t]
  \centering
  \begin{minipage}{0.49\textwidth}
    \includegraphics[width = 0.95\columnwidth]{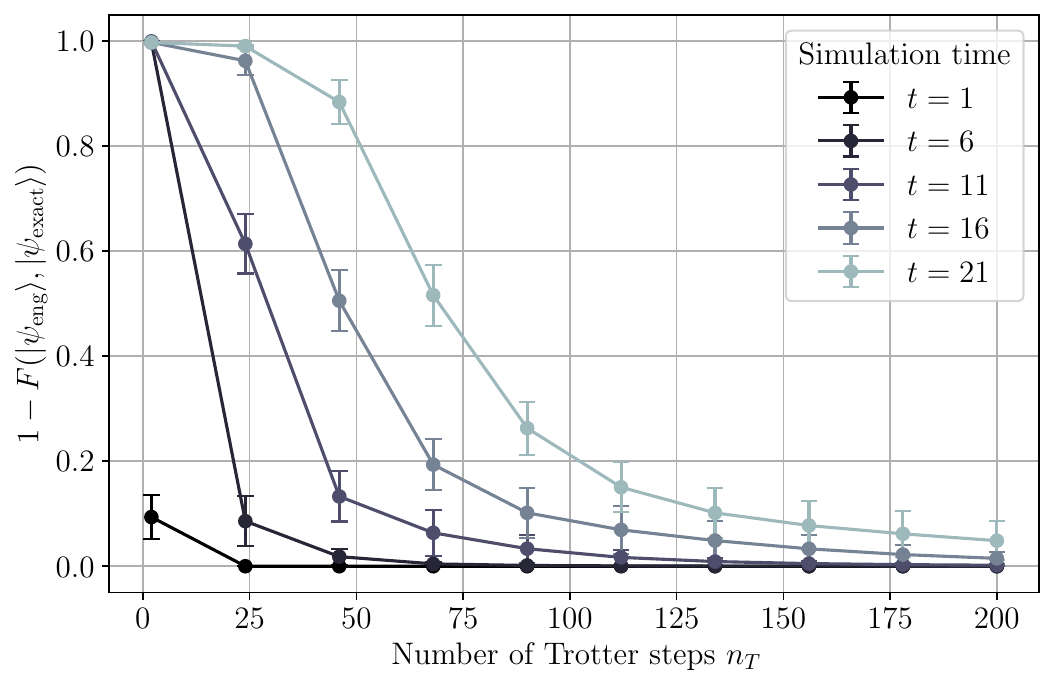}
    
    (a)
  \end{minipage}
  \begin{minipage}{0.49\textwidth}
    \includegraphics[width = 0.95\columnwidth]{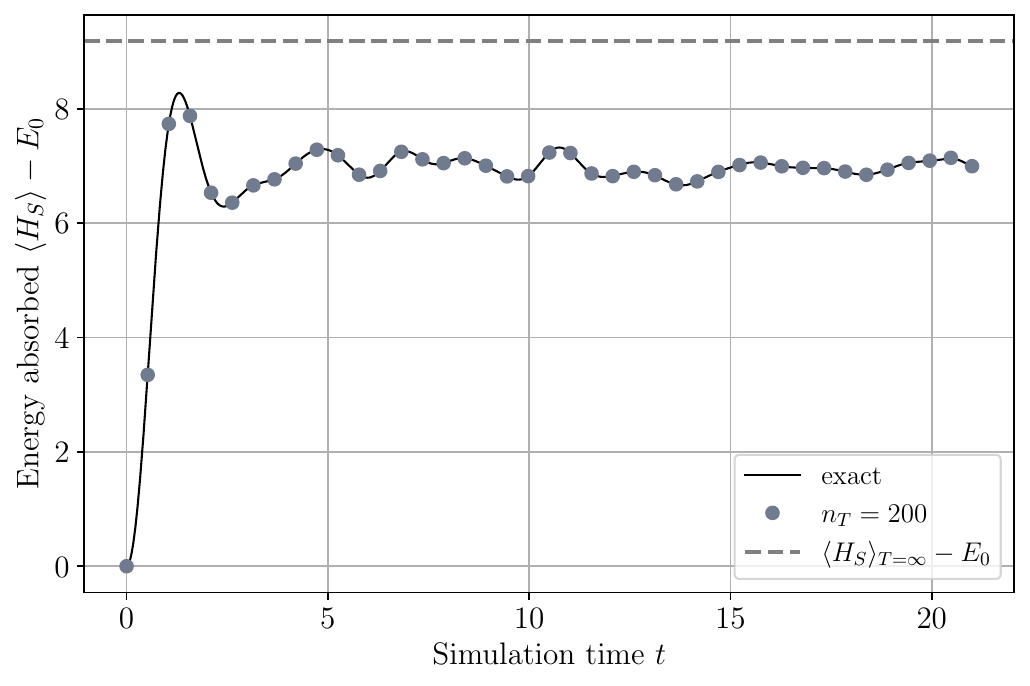}

    (b)
  \end{minipage}
\caption{
Numerical simulation of engineered Fermi--Hubbard dynamics on a one-dimensional chain with $L = 6$ sites ($12$ fermionic modes) at half filling, $N_\uparrow = N_\downarrow = 3$.
(a) Infidelity $1-F(\ket{\psi_\mathrm{eng}},\ket{\psi_\mathrm{exact}})$ between the engineered state $\ket{\psi_{\mathrm{eng}}}$ and the exact target state $\ket{\psi_{\mathrm{exact}}}$ as a function of the number of Trotter steps $n_T$.
The target Hamiltonians \eqref{eq:FH_target} contain quadratic hopping terms with complex coefficients sampled uniformly from the unit disc.
The engineered evolution converges to the exact target evolution as $n_T$ increases, over a range of simulation times $t$.
Error bars show the standard error of the mean over $20$ random target Hamiltonians.
(b) Energy relative to the ground state of the native Hamiltonian, $\langle H_S\rangle(t)-E_0$, for a representative target Hamiltonian with $n_T=200$ at each $t$.
The engineered dynamics closely follow the exact target dynamics, obtained by direct exponentiation of $H_T$, and show no drift towards the infinite-temperature value $\langle H_S\rangle_{T=\infty}-E_0$ (dashed), where $\langle H_S\rangle_{T=\infty} = U_S L/4 = 3U_S/2$.
}
  
  \label{fig:1D_sims}
\end{figure*}

As a first application of our method to interacting Hamiltonians, we consider tuning the $U/t$ ratio of the Fermi--Hubbard model on a 2D square lattice with nearest-neighbour tunnelling,
\begin{align}
  \label{eq:2D_Fermi--Hubbard}
  H_{\FH} = -t \sum_{\langle j,k\rangle,\sigma\in\{\uparrow,\downarrow\}} \left(c_{j\sigma}^\dagger c_{k\sigma} + \mathrm{H.c.}\right) + U \sum_{j} n_{j\uparrow} n_{j\downarrow}.
\end{align}
We take the pulse sequence $\{\1,V_A\}$ with evolution times $\{\lambda_0,\lambda_1\}$, where $V_A \coloneqq \prod_{j\in A} V_j(\pi)$ and $A$ denotes one of the two sublattices of the bipartite square lattice (see \cref{fig:2D_lattice}). Since $V_j(\pi)^\dagger\, c_{j\sigma}^\dagger\, V_j(\pi) = -c_{j\sigma}^\dagger$, conjugation by $V_A$ flips the sign of every nearest-neighbour hopping -- each of which connects sublattices $A$ and $B$ -- while leaving the on-site interaction invariant. The effective Hamiltonian therefore reads
\begin{align}
  \begin{split}
  H_\mathrm{eff}  &= \lambda_0 H_\FH + \lambda_1\, V_A^\dagger H_\FH V_A \\
                  &= -\underbrace{t(\lambda_0-\lambda_1)}_{\equiv\, t_T}\sum_{\langle j,k\rangle,\sigma}\left( c_{j\sigma}^\dagger c_{k\sigma} + \mathrm{H.c.}\right) \\
                  &\quad + \underbrace{U(\lambda_0+\lambda_1)}_{\equiv\, U_T}\sum_{j} n_{j\uparrow} n_{j\downarrow},
  \end{split}
\end{align}
i.e., a Fermi--Hubbard Hamiltonian with rescaled interaction strength to tunnelling ratio
\begin{align}
  \frac{U_T}{t_T} = \gamma\,\frac{U}{t}, \qquad \gamma \coloneqq \frac{\lambda_0+\lambda_1}{\lambda_0-\lambda_1}.
\end{align}

Since $\lambda_0,\lambda_1\geq 0$, we have $|\gamma|\geq 1$, meaning the construction can only increase the absolute $U/t$ ratio, never decrease it. The excluded case $\lambda_0=\lambda_1$ corresponds to the atomic limit $t_T=0$. 

Thus, the effective ratio $U_T/t_T$ is set by the evolution times rather than by the microscopic parameters that fix $U$ and $t$ in the system Hamiltonian.
In analogue platforms the two couplings are usually controlled by a common knob, e.g. the optical-lattice depth, which restricts the range of $U/t$ ratios that can be feasibly implemented in the lab.
Our construction provides a workaround.
Starting from a system Hamiltonian in a moderate, experimentally convenient regime, the effective $U/t$ can be tuned up by a factor $|\gamma|\geq 1$.

\subsubsection{Artificial gauge fields in the 1D Fermi--Hubbard chain}
\label{sec:1D_sim}
We now consider the task of simulating artificial gauge fields in an interacting system.
To this end, we consider a one-dimensional Fermi--Hubbard system Hamiltonian where, for concreteness, the uniform tunnelling coefficient is set to one,
\begin{align}
  \label{eq:FH_system}
  H_S = \sum_{j\in[L-1],\sigma} \left( c_{j+1,\sigma}^\dagger c_{j,\sigma} + \mathrm{H.c.} \right)
  + U_S \sum_{j} n_{j\uparrow} n_{j\downarrow}.
\end{align}
Using the Hamiltonian engineering protocol, we simulate time evolution under randomly generated target Hamiltonians of the form
\begin{align}
  \label{eq:FH_target}
  \begin{split}
  H_T =
  \sum_{j\in[L-1],\sigma}
  \left(
  \beta_{(j+1,j),\sigma}\,
  c_{j+1,\sigma}^\dagger c_{j,\sigma}
  + \mathrm{H.c.}
  \right)\\
  + U_T \sum_{j} n_{j\uparrow} n_{j\downarrow},
  \end{split}
\end{align}
where $\beta_{(j+1,j),\sigma}\in \mathbb{C}$ is sampled uniformly from the unit disc on the complex plane, i.e. $|\beta_{(j+1,j),\sigma}|\leq 1$, and $U_S$ is set such that $U_T =\|\ind{\lambda}\|_{\ell_1} U_S = 2$.

Since the tunnelling coefficients in \cref{eq:FH_target} are spin dependent, the pulses must address the two spin components separately, $V_{j\sigma}(\theta)=\e^{-\i\theta n_{j\sigma}}$, which requires spin-resolved local potentials.
For spin-independent targets, $\beta_{(j+1,j),\uparrow}=\beta_{(j+1,j),\downarrow}$, it suffices to consider pulses generated by $n_{j\uparrow}+n_{j\downarrow}$ and the number of modes entering the \ac{LP} is halved.

As in the Harper--Hofstadter case,
we prepare the system to be in the ground state $\ket{\psi_0}$ of the system Hamiltonian~\eqref{eq:FH_system} and quench to the target Hamiltonian~\eqref{eq:FH_target}.
We then benchmark the engineered evolution $U_\mathrm{eng}(t)$, obtained via the second-order Trotter formula \eqref{eq:engineered_evolution_operator}, against the exact evolution $U_\mathrm{exact}(t)$, obtained by direct exponentiation of $H_T$ using the state fidelity $F(\ket{\psi_\mathrm{eng}},\ket{\psi_\mathrm{exact}})$ defined in \cref{eq:fidelity}, and observe that the engineered state converges to the target dynamics for long simulation times and for a suitably chosen number of Trotter steps, see \cref{fig:1D_sims}~(a).

Quenches of this kind are a natural setting for early practical quantum advantage.
After a quench, entanglement typically grows rapidly with time, so the cost of classical simulation grows quickly and becomes impractical beyond modest system sizes and evolution times \cite{schuchEntropyScalingSimulability2008, schollwoeckDensitymatrixRenormalizationGroup2011, flanniganPropagationErrorsQuantitative2022}.
The small system sizes accessible to our exact-diagonalization benchmarks are themselves a reflection of this classical hardness, and a scaled-up analogue realization, for example of a quench under the Hofstadter--Hubbard model with the pulses discussed in \cref{sec:harper_hofstadter}, would operate in precisely the regime where classical simulation is no longer feasible.

To complement the state fidelity with a physically interpretable observable, we track the expectation value of the system Hamiltonian $\langle H_S\rangle(t)$. 
Since the quench starts in the ground state of $H_S$ with energy $E_0$, $\langle H_S\rangle-E_0$ is the energy injected into the system. 
Heating beyond the target dynamics would appear as a drift of the engineered $\langle H_S\rangle(t)$ away from the exact curve, towards $\langle H_S\rangle_{T=\infty}$.
\cref{fig:1D_sims}~(b) shows no such drift within the simulated window at $n_T=200$.

\section{Hamiltonian engineering for number-operator-free Hamiltonians}
\label{sec:FHE_general}

In this section, we characterize a general class of Hamiltonians that is naturally amenable to our Hamiltonian engineering approach with local pulses.
First, we show that conjugation-based methods using single-mode pulses cannot modify terms which are monomials of the number operators, and then focus on the class of number-operator-free fermionic Hamiltonians. 
This class includes tunnelling terms in quadratic Hamiltonians as a special case, but also accommodates higher-order terms.
The algebraic structure of these Hamiltonians closely parallels Pauli decomposition commonly used for qubit Hamiltonians, which allows many concepts from conjugation-based qubit Hamiltonian engineering \cite{basslerGeneralEfficientRobust2025a} to be carried over to the fermionic setting.
The key difference in the case of fermionic Hamiltonians is that the relevant degrees of freedom are described by complex numbers rather than real numbers. 

\subsection{No-go result for conjugating interaction terms}
\label{sec:nogo}

We begin by showing that conjugation-based methods using local (single-mode) pulses alone are not sufficient to modify interaction terms.
We define
\begin{align}
  N_\Omega \coloneqq \prod_{j\in\Omega} n_j \quad \text{where} \quad \Omega \subseteq [n],
\end{align}
to be an arbitrary monomial of the number operators with nontrivial support on a subset $\Omega$ of all fermionic modes, such as the two-body density-density interactions $n_j n_k$.
Interaction terms are polynomials of number operators
\begin{align}
  \label{eq:interaction_term}
  I = \sum_{\emptyset \subsetneq \Omega\subseteq [n]}g_\Omega N_\Omega,
\end{align}
where $g_\Omega$ are real coefficients.
As we argued before, the local unitaries $V_j(\theta)$ commute with the number operators, i.e.
\begin{align}
  [n_j,V_k(\theta)] = 0
\end{align}
for all $j, k \in [n]$ and for all $\theta \in [0,2\pi)$.
This means that all monomials $N_\Omega$, and therefore also the interaction terms $I$, are invariant under conjugation by the local unitaries 
\begin{align}
  V_k^\dagger(\theta) N_\Omega V_k(\theta) = N_\Omega.
\end{align}
Hence, the local unitaries $V_k(\theta)$ cannot be used to engineer local effective interaction strengths for interacting Hamiltonians.

Moreover, it can be seen that \emph{any} conjugation-based method which uses local unitaries cannot be used to modify interaction terms.
This is because any physical local unitary $O_j$ must obey the fermion-parity superselection rule \cite{Vidal21QuantumOperationsInAnInformationTheoryForFermions} and thus commute with the parity operator $\e^{\i\pi n_j}$.
Therefore all physical local unitaries are diagonal in the occupation basis and commute with the number operator $n_j$,
\begin{align}
  \label{eq:n_conjugation}
  O_j^\dagger n_j O_j =  n_j \quad \text{for any local unitary}\ O_j.
\end{align}

It follows that interaction terms like in \cref{eq:interaction_term} commute with all pulses of the form 
\begin{align}
  \label{eq:pulses_for_n}
  O_\ind{b} = \prod_{j\in [n]} O_j(b_j), 
\end{align}
where $O_j(b_j)$ are local unitaries parametrized by $b_j\in\ZZ$, and we have
\begin{align}
  \sum_{\ind{b}\in\ZZ^n} \lambda_{\ind{b}} O_{\ind{b}}^\dagger I O_{\ind{b}}= \sum_{\ind{b}\in\ZZ^n} \lambda_{\ind{b}} I = \|\ind{\lambda}\|_{\ell_1} I.
\end{align}
This means that conical combinations of conjugations under pulses constructed from local unitaries merely rescale the interaction terms globally, and cannot tune the strengths individually.
\subsection{Number-operator-free Hamiltonians}
\label{sec:tritstring_Hamiltonians}

Guided by our findings in the previous \cref{sec:nogo}, we now focus on the class of number-operator-free Hamiltonians.
Let $\mathbb{T} = \{-1,0,+1\}$ denote the set of \emph{balanced ternary digits}.
We define a family of functions $\{C^{(j)}\}_{j\in [n]}$, $C^{(j)}:~\mathbb{T}\to~L(\mathcal{F})$, where $L(\mathcal{F})$ is the space of linear operators on the Fock space $\mathcal{F}$, by
\begin{align}\begin{split}
  \label{eq:trit}
	C^{(j)}_a
	\coloneqq \begin{cases} 	
		\1 &\text{if } a =  \phantom{+}0,\\
		\cdag_j &\text{if } a = +1,\\
		c_j &\text{if } a = -1. 
	\end{cases}
\end{split}\end{align}
This definition enables us to identify the operators $\1, c_j^\dagger, c_j$  with the numbers $0,+1,-1$, respectively. 

Next, we define $C:~\TT^n\to~L(\mathcal{F})$ by 
\begin{align}
  \label{eq:defn_C}
	C_{\ind{a}} \coloneqq \prod_{j\in[n]} C^{(j)}_{a_j} = C^{(1)}_{a_1}C^{(2)}_{a_2}\dots C^{(n)}_{a_n},
\end{align}
which allows us to identify 
a string of fermionic annihilation and creation operators on different modes with 
an $n$-trit string $\ind{a} = (a_1,a_2,\dots, a_n) \in \TT^n$.
We often refer to strings $\ind{a}\in\TT^n$ as \emph{tritstrings}. 
For example, the operator $C_{(1,0,-1,-1)} = \cdag_1 c_3 c_4$ is associated with the tritstring~$(1,0,-1,-1)$. 
The order in which the operators on the right-hand side of \cref{eq:defn_C} appear is important because of the fermionic anticommutation relations, so the operator is defined such that the mode index increases as we go from left to right. 
A useful observation is that under the map $C$, additive inverses $-\ind{a},\ind{a}\in \TT^n$ map to operators which are (up to a sign) Hermitian conjugates of each other:
\begin{equation}
C_{-\ind{a}}= (-1)^w C_\ind{a}^\dagger 
\quad \text{with}\quad 
w=\binom{\wt(\ind{a})}{2}, 
\end{equation}
where $\wt(\ind{a}) 
\coloneqq\abs{\Set*{j\in [n]\given a_j\neq 0}}
$
is the \emph{Hamming weight} of $\ind{a}$.

Any fermionic Hamiltonian $H$ whose terms do not involve number operators can be written as a complex linear combination of the operators $C_\ind{a}$,
\begin{equation}\label{eq:number-freeH}
  H = \sum_{\ind{a}\in\TT^n} \alpha_{\ind{a}} C_\ind{a}
\end{equation}
with $\alpha_\ind{a}\in \CC$. 
Since $H$ is required to be Hermitian, the coefficients associated with inverse tritstrings are not independent, but are related by 
\begin{equation}
\alpha_{-\ind{a}} 
= 
(-1)^w\alpha_\ind{a}^*.
\end{equation} 

As a consequence, each pair $\{\ind{a},-\ind{a}\}\subset \TT^n\setminus\{0\}$ contributes only a single complex degree of freedom.
To avoid this redundancy, it is convenient to choose a bipartition
\begin{equation}\label{eq:trit_split}
  \mathcal{T}_n\sqcup \left(-\mathcal{T}_n\right) = \TT^n\setminus\{\ind{0}\},
\end{equation}
such that exactly one representative from each pair $\{\ind{a},-\ind{a}\}$ is contained in $\mathcal{T}_n$, i.e. $\ind{a}\in\mathcal{T}_n$ if and only if $-\ind{a}\notin\mathcal{T}_n$. 
This allows us to parameterize Hamiltonians using a minimal set of independent coefficients, which reduces the dimension of and redundancy in the Hamiltonian engineering problem and we can write the Hamiltonian \eqref{eq:number-freeH} as 
\begin{align}
  H = \sum_{\ind{a}\in\mathcal{T}_n}\alpha_\ind{a} C_\ind{a} + \mathrm{H.c.},
\end{align}
where we have dropped the term proportional to the identity, which corresponds to a physically unobservable constant energy shift.
For brevity, wherever it does not produce any confusion, 
we adopt the convention that the Hermitian conjugate terms are implied, but not indicated explicitly, i.e.
\begin{align}
  H = \sum_{\ind{a}\in\mathcal{T}_n}\alpha_\ind{a} C_\ind{a}.
\end{align}

For concreteness, one admissible choice of bipartition of $\TT^n\setminus\{\boldsymbol{0}\}$ is given by the set of tritstrings of length $n$, for which the first non-zero entry is ``$1$''.
As a regular expression
\begin{align}
\label{eq:regex}
\mathcal{T}_n \coloneqq \{\ind{a}\in\TT^n \ |\  \ind{a}\in 0^*1\TT^* \},
\end{align}
meaning that $\mc{T}_n$ is the set of tritstrings of length $n$ of the form $(0,0,\dots,0,1,\ind{a}')$ with a shorter tritstring $\ind a'$. 
Some elements from $\mathcal{T}_n$ and $-\mathcal{T}_n$ constructed in this way are explicitly given in \cref{tab:SminusS}.
The total number of terms in each column of \cref{tab:SminusS} is
\begin{align}
	\sum_{k=0}^{n-1}3^k =\frac{3^n-1}{2}. 
\end{align}
So, this is indeed a bipartition, since $|\TT^n\setminus\{\ind{0}\}| = 3^n-1$.
Throughout this paper, we will use the definition \eqref{eq:regex} for $\mathcal{T}_n$.

\begin{table}[t]
\centering
\setlength{\tabcolsep}{8pt}
\renewcommand{\arraystretch}{1.2}

\begin{tabular}{cc}
\toprule
$\mathcal{T}_n$ & $-\mathcal{T}_n$ \\
\midrule

$\begin{aligned}
&(0,\cdots,0,0,1),\\
&(0,\cdots,0,1,0),\\
&(0,\cdots,0,1,1),\\
&(0,\cdots,0,1,-1),
\end{aligned}$ &
$\begin{aligned}
&(0,\cdots,0,0,-1),\\
&(0,\cdots,0,-1,0),\\
&(0,\cdots,0,-1,-1),\\
&(0,\cdots,0,-1,1),
\end{aligned}$ \\

$\vdots$ & $\vdots$ \\

\bottomrule
\end{tabular}

\caption{Example elements of the sets $\mathcal{T}_n$ and $-\mathcal{T}_n$.}
\label{tab:SminusS}
\end{table}

The notation introduced in this section is useful to write down the conjugation behavior of the operators $C_\ind{a}$. Using the definitions \eqref{eq:trit} and \eqref{eq:defn_C}, and the conjugation \eqref{eq:conjugations} of annihilation and creation operators we find 
\begin{align}
  \label{eq:tritstring_conjugation}
  \begin{split}
  V_\ind{b}^\dagger(\theta) C^{(j)}_{a_j} V_\ind{b}(\theta) &=  \e^{\i a_jb_j\theta} C^{(j)}_{a_j}, 
  \\
  V_\ind{b}^\dagger(\theta) C_\ind{a} V_\ind{b}(\theta) &=  \e^{\i\theta \ind{a}\cdot\ind{b}} C_\ind{a},
  \end{split}
\end{align}
where $\ind{a}\cdot \ind{b} = \sum_{j\in[n]} a_jb_j \in \ZZ$.

We focus on local unitaries defined by \cref{eq:V_definition} with $\theta = \frac{2\pi}{3}$,
\begin{align}
  \begin{split}
      \label{eq:local_ops_3}
    	V_j&\coloneqq V_j\left(\frac{2\pi}{3}\right) = \exp{\left(-\i \frac{2\pi}{3} n_j\right)}, 
      \\
      V_j^\dagger &= V_j^{-1} = V_j \left(-\frac{2\pi}{3}\right).
  \end{split}
\end{align}
As we discuss later, the particular choice of phase $2\pi/3$ leads to the smallest member of a family of Hamiltonian engineering protocols, which can efficiently simulate target Hamiltonians with arbitrary complex coefficients. 
Applying these operators in parallel we obtain the pulses~\eqref{eq:defn_pulse}
\begin{align}
  \label{eq:pulses_23}
  V_\ind{b} &\coloneqq \prod_{j\in[n]} V_j^{b_j},
\end{align}
where $\ind{b}\in\TT^n$ and $b_j = \pm 1$ indicates we apply the local operation $V_j(\pm \frac{2\pi}{3})$ and $b_j = 0$ implies we apply the identity operator at mode $j$. 
Conveniently, both terms in number-operator-free Hamiltonians $C_\ind{a}$ and local pulses $V_\ind{b}$ are indexed by tritstrings in $\TT^n$.

In terms of the primitive third root of unity $\omega \coloneqq \e^{2\pi\i/3}$ and for this particular set of pulses, \cref{eq:tritstring_conjugation} becomes
\begin{align}
  \label{eq:conjugation}
  V_\ind{b}^\dagger C_\ind{a} V_\ind{b} = \omega^{\ind{a}\cdot{\ind{b}}}C_\ind{a}.
\end{align}
For a single mode, tabulating the phases obtained via conjugation yields the Fourier matrix  \cite{banica2024invitationhadamardmatrices} of the cyclic group of order 3
\begin{align}
	\label{eq:fourier_3}
	F^{(3)} &= \begin{pmatrix}
		1 & 1 & 1\\
		1 & \omega & \omega^2\\
		1 & \omega^2 & \omega 
	\end{pmatrix}.
\end{align}

We define $F^{(3)}\big|_{\mathcal{T}_1}$, the restriction of $F^{(3)}$ on row indices restricted to be taken from $\mathcal{T}_1=\{1\}$, 

\begin{align}
  F^{(3)}\big|_{\mathcal{T}_1} = \begin{pmatrix}
		1 & \omega & \omega^2
	\end{pmatrix}.
\end{align} 
For $n$ modes, we define the matrix $F\in\CC^{d \times c}$, by
\begin{equation}\label{eq:def:F}
  F_{\ind{a}\ind{b}}\coloneqq \omega^{\ind{a}\cdot\ind{b}},
\end{equation}
for $\ind{a}\in\mathcal{T}_n$ and $\ind{b}\in\TT^n$, where $d = |\mathcal{T}_n|$ is the number of rows and $c=3^n$ is the number of columns. 

Therefore, the conjugations \eqref{eq:conjugation} produce phases that correspond to the entries of the matrix $F$, which is a submatrix of the Fourier matrix of the group $\ZZ_3^n$ obtained by restricting the row indices to be taken from $\mathcal{T}_n$, i.e.\ $F = (F^{(3)})^{\otimes n}\big|_{\mathcal{T}_n}$. 

We recall the following definition which will be needed to state the upcoming lemma.
For a set of vectors $\mathcal{V} = \{\ind{v}_j\}_{j\in [n]}$, the space spanned by non-negative linear combinations of elements from $\mc{V}$ is called the conical hull of $\mc{V}$, which is denoted by 
\begin{align}
  \label{eq:cone_defn}
  \cone(\mc{V}) \coloneq \Set*{ \sum_j \lambda_j \ind{v}_j \given \lambda_j \ge 0 \text{ for all } j\in [n] }. 
\end{align}
Further, we denote the set of columns of the matrix $F$ by $\col(F)\subset \CC^d$.
These definitions allow us to state the following lemma, which is useful for proving our main \cref{thm:main}.
A proof of the lemma can be found in \cref{ap:lemma_proof}.

\begin{lemma}
  \label{lem:main}
  The conical hull of the set of columns of $F\in\CC^{d \times c}$ is the complex vector space $\CC^d$ with $d = \tfrac{1}{2}(3^n-1)$, i.e. $\cone(\col(F)) = \CC^d$. 
  Hence, as a map $F:\RR_{\geq 0}^c \to \CC^d$, $F$ is surjective.
\end{lemma}

\subsection{Hamiltonian engineering}
\label{sec:he_tritstring}
Our main objective is to simulate time evolution under a \emph{target Hamiltonian} $H_T$, given access to a native \emph{system Hamiltonian} $H_S$ and local pulses, as in \cref{eq:pulses_23}. 
Let the system and the target Hamiltonians be given as
\begin{align}
    \label{eq:system_hamiltonian}
  H_S = \sum_{\ind{a}\in\mathcal{T}_n} \alpha_\ind{a} C_\ind{a},
  \\
  \label{eq:target_hamiltonian}
  H_T = \sum_{\ind{a}\in\mathcal{T}_n} \beta_\ind{a} C_\ind{a}, 
\end{align}
where $\alpha_\ind{a}, \beta_\ind{a} \in \CC$. 
Frequently, we group the components for the system and target Hamiltonians in $d$-dimensional vectors, $\ind{\alpha},\ind{\beta}\in\CC^d$.
The support of a vector 
$\boldsymbol{v}=(v_\ind{a})_{\ind{a} \in \mc{T}_n}\in\CC^d$ 
is the set of indices for which the components are non-zero, i.e. $\operatorname{supp}(\boldsymbol{v})\coloneqq \{\ind{a}\in\mathcal{T}_n\mid v_\ind{a}\neq 0\}$.

We decompose the target Hamiltonian in terms of conjugations of the system Hamiltonian under the pulses $\{V_\ind{b}\}_{\ind{b}\in\TT^n}$,
  \begin{align}
  \label{eq:target_ham_constraint}
  H_T {=} \sum_{\ind{b}\in\TT^n} \lambda_\ind{b} V_\ind{b}^\dagger H_S V_\ind{b} \equiv \sum_{\ind{b}\in\TT^n} H_\ind{b},
  \end{align}
  because then,
\begin{align}
  \begin{split}
  \e^{-\i H_T t}  &= \e^{-\i t \sum_{\ind{b}}\lambda_\ind{b} V_\ind{b}^\dagger H_S V_\ind{b}}\\
                  \label{eq:approximation}
                  &\approx \prod_{\ind{b}\in\TT^n} e^{-\i t  \lambda_\ind{b} V_\ind{b}^\dagger H_S V_\ind{b}}\\
                  &= \prod_{\ind{b}\in\TT^n} V_\ind{b}^\dagger \, \e^{-\i t \lambda_\ind{b} H_S} \ V_\ind{b},
  \end{split}
\end{align}
where the last line can be implemented on the simulator, allowing the experimenter to simulate time evolution under the target Hamiltonian by only using the local pulses \eqref{eq:pulses_23} and free evolution under the system Hamiltonian $H_S$. 
The thus implemented time evolution is exact only if the terms $H_\ind{b}$ in \cref{eq:target_ham_constraint} commute. 
When they do not commute, the product on the right side of \cref{eq:approximation} is replaced by a Trotter-Suzuki type product formula \cite{trotterProductSemiGroupsOperators1959,childsTheoryTrotterError2021,dalzellQuantumAlgorithmsSurvey2023}. 
More details on the Trotter error can be found in \cref{sec:trotter}. 

\begin{theorem}
  \label{thm:main}
Let $H_T$ be a target Hamiltonian of the form \eqref{eq:target_hamiltonian}. 
Given access to a system Hamiltonian $H_S$ of the form \eqref{eq:system_hamiltonian} with $r = \abs{\operatorname{supp}(\boldsymbol{\alpha})}$ terms such that $\operatorname{supp}(\boldsymbol{\beta})\subseteq \operatorname{supp}(\boldsymbol{\alpha})$, and access to local pulses $V_\ind{b}$ as in \cref{eq:pulses_23}, there exists a set of at most $2r$ pulses $\{V_{\ind{b}_\mu}\}_{\mu\in[2r]}$ and $2r$ nonnegative numbers $\{\lambda_{\ind{b}_\mu}\}_{\mu\in[2r]}$ such that
 \begin{align}
    \label{eq:thm}
    U_\mathrm{eng}(t) = \prod_{\mu=1}^{2r} V_{\ind{b}_\mu}^\dagger \e^{-\i t H_S \lambda_{\ind{b}_\mu}} V_{\ind{b}_\mu}  \approx \e^{-\i H_T t}. 
  \end{align}
  The approximation error in operator norm $\norm{}$ is bounded by the standard Trotter scaling, i.e., by $\LandauO\bigl(\norm{H_T}_1^{2p+1}t^{2p+1}/{n_T^{2p}}\bigr)$ when using the $2p$-th order product formula, where $\norm{H_T}_1\coloneqq \sum_{\mu=1}^{2r} \norm{H_{\ind{b}_\mu}}$.
\end{theorem}

\begin{proof}
  It suffices to show that for system and target Hamiltonians which satisfy the assumptions of the theorem, a decomposition \eqref{eq:target_ham_constraint} with only $2r$ terms exists. 
  Using \cref{eq:target_ham_constraint,eq:conjugation}, 
  we obtain
  \begin{align}
    \begin{split}
    H_T= 
  \sum_{\ind{a}\in\mathcal{T}_n} \beta_\ind{a} C_\ind{a}  &= \sum_{\ind{b}\in\TT^n} \lambda_\ind{b} V_\ind{b}^\dagger \left(\sum_{\ind{a}\in\mathcal{T}_n}\alpha_\ind{a} C_\ind{a}\right) V_\ind{b}\\
                                        &= \sum_{\ind{a}\in\mathcal{T}_n} \alpha_\ind{a} \left(\sum_{\ind{b}\in\TT^n} \lambda_\ind{b} V_\ind{b}^\dagger C_\ind{a} V_\ind{b} \right)\\
                                        &= \sum_{\ind{a}\in\mathcal{T}_n} \alpha_\ind{a} \left(\sum_{\ind{b}\in\TT^n} \omega^{\ind{a}\cdot\ind{b}}\lambda_\ind{b} \right)C_\ind{a}.
    \end{split}
  \end{align}
  Therefore, a decomposition of the target Hamiltonian of the form \eqref{eq:target_ham_constraint} exists, if the following linear system of equations has a solution,
  \begin{align}
    \label{eq:elementwise_constraint}
    \alpha_\ind{a} \sum_{\ind{b}\in\TT^n} \omega^{\ind{a}\cdot\ind{b}}\lambda_\ind{b} = \beta_\ind{a},
  \end{align}
which we write in matrix language as
\begin{align}
  \label{eq:matrix_constraint}
  F\boldsymbol{\lambda} = \boldsymbol{\gamma} \equiv \boldsymbol{\beta}\oslash\boldsymbol{\alpha},
\end{align}
where $\boldsymbol{\lambda}\in\RR_{\geq 0}^{c}$ is a nonnegative vector of the evolution times and we used entrywise division $\oslash$, defined as
\begin{align}
  \label{eq:entrywise_division}
  \boldsymbol{\beta}\oslash\boldsymbol{\alpha} \coloneqq 
  \begin{cases}
    \beta_\ind{a}/\alpha_\ind{a} \quad &\text{if} \quad \alpha_\ind{a} \neq 0\\
    0 \quad &\text{if} \quad \alpha_\ind{a} = 0.
  \end{cases}  
\end{align}
Since by assumption $\supp(\boldsymbol{\beta})\subseteq \supp(\boldsymbol{\alpha})$, we have $\supp(\boldsymbol{\gamma}) = \supp(\boldsymbol{\beta})$.
By \cref{lem:main}, the columns of the Fourier submatrix $F$ conically span the full complex space $\CC^d$, and therefore
\begin{align}
  \boldsymbol{\gamma} \in\cone(\col(F)) \quad \forall\boldsymbol{\gamma}\in\CC^r\subseteq \CC^d.
\end{align}
In other words the map $F:\RR_{\geq 0}^c \to \CC^r$ is surjective.
This means that  by Carathéodory's theorem \cite{baranyGeneralizationCaratheodorysTheorem1982}, there exists a $2r$-sparse nonnegative vector $\boldsymbol{\lambda}^*\in\RR_{\geq 0}^c$ such that
\begin{align}
  F\boldsymbol{\lambda}^* = \boldsymbol{\gamma}.
\end{align}
\end{proof}

Let us now see how to best find the pulses required for the simulation task.
For this, it is important to understand how the objects in \cref{eq:matrix_constraint} relate to the physical problem at hand.
The rows of the matrix $F$ correspond to number-operator-free terms in the system Hamiltonian, and the columns correspond to pulses which we can use to conjugate the Hamiltonian with. 
The entry $F_{\ind{a}\ind{b}}=\omega^{\ind{a}\cdot\ind{b}}$ therefore corresponds to the phase acquired by the term $C_\ind{a}$ in the Hamiltonian when conjugated under the $V_\ind{b}$.
The vector $\ind{\gamma}$ on the right side of \cref{eq:matrix_constraint}, tells us how much and in which direction in the complex plane the system Hamiltonian needs to be ``scaled'' to obtain the target Hamiltonian.
A solution $\ind{\lambda}$ to the equation system corresponds to a description of an implementation of the simulation task, where the pulses that need to be applied can be read-off from the indices of the non-zero entries of $\ind{\lambda}$, and the evolution time between the conjugation layers for a given pulse, say $V_{\ind{b}_\nu}$, corresponds to the entry $\lambda_{\ind{b}_\nu}$.

It is easy to see why we require $\supp(\ind{\beta}) \subseteq \supp(\ind{\alpha})$, since otherwise a decomposition of the form \eqref{eq:target_ham_constraint} does not exist, as can be seen from \cref{eq:elementwise_constraint}; setting $\alpha_\ind{a} = 0$ requires $\beta_\ind{a} = 0$.
However, this means that the rows of $F$ which are labelled by indices outside $\supp(\ind{\alpha})$ do not play a role in determining the set of pulses required for the simulation, and may be discarded.
Therefore we may define
\begin{align}
  F' \coloneqq F\big|_{\supp(\ind{\alpha})}\in\CC^{r\times c}, \quad  \ind{\gamma}'\coloneqq \ind{\gamma}\big |_{\supp(\ind{\alpha})}\in\CC^r
\end{align}
so that the Hamiltonian engineering problem reduces to
\begin{align}
  \label{eq:reduced_matrix_constraint}
  F'\ind{\lambda} = \ind{\gamma}'.
\end{align}
This is important, since usually the number of terms in the Hamiltonian is $r\in\poly(n)$ in practice, whereas $d~=~|\mathcal{T}_n|~=~(3^n-1)/2$, so this allows us to exponentially reduce the number of rows.

In general, the equation system \eqref{eq:reduced_matrix_constraint} is underdetermined, therefore it has infinitely many solutions and thus there are infinitely many ways of simulating the target Hamiltonian.
Furthermore, since $\norm{H_T}_1=\sum_\mu\norm{H_{\ind b_\mu}}=\norm{\boldsymbol\lambda}_{\ell_1}\norm{H_S}$ in \cref{thm:main}, minimizing the quantum runtime simultaneously minimizes the leading Trotter-error bound.
Therefore, it is desirable to find the solution which minimizes the total quantum runtime $\norm{\boldsymbol{\lambda}}_{\ell_1} = \boldsymbol{1}^\top \boldsymbol{\lambda} = \sum_\ind{b} \lambda_\ind{b}$.
For this reason, we phrase the Hamiltonian engineering problem as an \ac{LP}
\begin{align}
\begin{split}
\label{eq:LP}
\text{minimize} \quad & \boldsymbol{1}^\top \boldsymbol{\lambda} \\
\text{subject to} \quad &  F' \boldsymbol{\lambda} = \boldsymbol{\gamma}', \\
                        &  \boldsymbol{\lambda} \in \RR_{\geq 0}^{c}.
\end{split}
\end{align}
In this context, the matrix $F'$ is called the \emph{constraint matrix} and a vector $\boldsymbol{\lambda}$ satisfying all the constraints of the \ac{LP} \eqref{eq:LP} is called a \emph{feasible solution}.
An \emph{optimal solution} $\boldsymbol{\lambda}^*$, marked by an asterisk, is a feasible solution that also minimizes the objective function.
If the \ac{LP} \eqref{eq:LP} has a feasible solution, then there exists a $2r$-sparse optimal solution \cite{baranyGeneralizationCaratheodorysTheorem1982},
which corresponds to a decomposition of the target Hamiltonian \eqref{eq:target_ham_constraint}, with only $2r$ terms.
Such a $2r$-sparse optimal solution can be found using the simplex algorithm which, in practice, has a runtime that scales polynomially in the size of the linear program $r \cdot c$ \cite{spielmanSmoothedAnalysisAlgorithms2004,bachOptimalSmoothedAnalysis2025}. 

Since $c=3^n$, the \ac{LP} \eqref{eq:LP} grows exponentially with the system size.
Nevertheless, it can be solved as it is for small system sizes and the optimum pulse sequence for the simulation task can be found. 
For larger system sizes, it was observed in \cite{basslerGeneralEfficientRobust2025a} for a qubit setting, that we can consider an exponentially smaller constraint matrix which allows us to find a solution to the Hamiltonian engineering problem efficiently, at the cost of suboptimal quantum runtimes.
The idea is to construct an \emph{efficient relaxation} of the original \ac{LP} by uniformly subsampling $s \sim 4r$ columns from the matrix $F'$, and use these to define a new constraint matrix $\tilde{F}\in\CC^{r\times s}$.
Crucially, the matrix $\tilde{F}$ grows polynomially with the size of the system, and thus the optimal solution $\tilde{\ind{\lambda}}^*$ to the \ac{LP} defined using $\tilde{F}$ can be found efficiently. 
As $\tilde{F}$ contains only a subset of the columns of the original constraint matrix $F'$, $\tilde{\ind{\lambda}}^*$ is a feasible solution of the original \ac{LP}, but is not optimal in general. 
However, we see that by increasing the number of columns we sample from $F'$, we can get close to the optimal quantum runtime while preserving the efficiency of the relaxation.
For example, for quadratic Hamiltonians, $r\in \LandauO(n^2)$ and therefore the efficient relaxation allows us to reduce the size of the \ac{LP} to $rs \in \LandauO(n^4)$ and thus efficiently find a set of pulses for the engineering task.
For more details on efficient relaxation see \cref{sec:relaxation}.

While our engineering method cannot \emph{generate} terms that are not already present in the system Hamiltonian, it can \emph{remove} terms by dynamically setting their coefficients to zero. 
\Cref{eq:entrywise_division} shows that this is agnostic of the original coefficient $\alpha_{\ind a}$, since $0/\alpha_{\ind a} = 0$ for any non-zero $\alpha_{\ind a}$ in $H_S$.
This is particularly useful for mitigating error terms in the system Hamiltonian when the exact coefficients in the number-operator-free decomposition are not known.

Earlier we made the seemingly arbitrary choice of phase angle $\theta = 2\pi/3$. 
As we see in \cref{ap:lemma_proof}, this choice is justified since it yields the smallest constraint matrix obtained from local pulses whose conical hull is equal to the complex space $\CC^r$. 
But we are free to choose any phase $\theta_k = 2\pi/k$, with $k\geq 3$, which gives rise to a primitive $k$-th root of unity $\omega_k = \e^{2\pi\i/k}$. 
It is easy to see that if we choose $k=2$, we always get either $\pm1$ as a phase in front of the terms of the Hamiltonian, therefore it is not possible to span the complex space $\CC^r$.
However, other roots of unity with $k\geq 4$ may be used to define the local unitaries.
Conjugation of the Hamiltonian terms under pulses generated by these local unitaries give
\begin{align}
  V_\ind{b}^\dagger\left(\theta_k\right) C_\ind{a} V_\ind{b}\left(\theta_k\right) = \omega_k^{\ind{a}\cdot{\ind{b}}}C_\ind{a},
\end{align}
and the constraint matrix we obtain is of the form
\begin{align}
  F_k' \coloneqq (F^{(k)})^{\otimes n}\big|_{\supp(\ind{\alpha})}
\end{align}
with $F^{(k)}_{\ind a\ind b}= \omega_k^{\ind a\cdot\ind b}$.
That is, the constraint matrix is obtained by restricting to the relevant rows of the Fourier matrix of the group $\ZZ_k^n$.
Importantly, the proof of \cref{lem:main_appendix} applies for all $k \geq 3$, implying that \cref{thm:main} holds in this case as well.

Choosing higher orders for the roots of unity increases the size of the linear program, thereby reducing the efficiency of the algorithm we use to find the pulses.
However, it is possible to reduce the quantum evolution times substantially for some engineering tasks when we consider local unitaries generated by different roots of unity. 
For example, the simple task of engineering 
\begin{align}
  \alpha C_{(1,-1)}\longmapsto -\alpha C_{(1,-1)}
\end{align}
requires only one pulse, $V_1(\pi)$ and an evolution time of $1$ if we choose $k=4$ (or any even number larger than $2$), but two pulses $V_1(2\pi/3)$, $V_1(-2\pi/3)$ and an evolution time of $2$ if we choose $k=3$.

\section{Error and resource analysis}
\label{sec:implementation_errors}

In this section, we analyze the main sources of error present in our framework and discuss the quantum resources required for its implementation. In particular, we consider errors arising from product-formula approximations and finite pulse times, and quantify the quantum runtime and number of pulses required to realize the engineered dynamics.

There are two main assumptions in our procedure which give rise to errors. 
The first is that the product of exponentials implemented in the experiment, can be made to converge to the exponential of the sum, 
\begin{align}
  \label{eq:trotter_target}
  \prod_{\ind{b}\in\TT^n} e^{-\i t  \lambda_\ind{b} V_\ind{b}^\dagger H_S V_\ind{b}} \approx e^{-\i t \sum_{\ind{b}}\lambda_\ind{b} V_\ind{b}^\dagger H_S V_\ind{b}}.
\end{align}
This is a standard approximation and taken care of by product formulae like the Trotter-Suzuki formula.

The second assumption is that the system Hamiltonian can be ``turned off'' while we apply a pulse, which is not the case for all analogue experiments.
For ultracold atoms in optical lattices, it is possible to turn off the tunnelling part of the Hamiltonian by increasing the lattice depth, so that the pulses do not overlap with the tunnelling part of the system Hamiltonian.
For other implementations of analogue quantum simulators, we follow \cite{Votto23UniversalQuantumProcessors,basslerGeneralEfficientRobust2025a}  and use average Hamiltonian theory techniques to mitigate the errors stemming from these so-called \emph{finite-pulse-time} effects.
\subsection{Product formulae}
\label{sec:trotter}
Given a Hamiltonian decomposition $H = \sum_{\mu\in [L]} H_\mu$ with $L\in\ZZ_+$ terms, time evolution under $H$ can be approximated using the first order product formula with $n_T$ \emph{Trotter cycles}~\cite{dalzellQuantumAlgorithmsSurvey2023}
\begin{align}
  \label{eq:first_order_trotter}
  S_1(t) \coloneqq \left(\prod_{\mu=1}^{L}\e^{-\i H_\mu t/n_T}\right)^{n_T}.
\end{align}
The error in the first-order product formula is upper bounded as
\begin{align}
  \left\|S_1(t)-\e^{-\i Ht}\right\| \leq \frac{t^2}{2n_T}\sum_{\mu=1}^{L}\left\|\sum_{\nu>\mu}^{L}[H_\mu,H_\nu]\right\|,
\end{align}
where $\norm{}$ denotes the operator norm. 
The $2p$-th order product formula is defined recursively \cite{dalzellQuantumAlgorithmsSurvey2023,childsTheoryTrotterError2021}, with Trotter error bounded by
\begin{align}
  \label{eq:trotter_bound}
  \left\|S_{2p}-\e^{-\i H t} \right\| \in \LandauO\left(\frac{\left\|H\right\|_1^{2p+1}t^{2p+1}}{n_T^{2p}}\right),
\end{align}
where $\left\|H\right\|_1\coloneqq \sum_{\mu=1}^{L} \left\|H_\mu\right\|$.
In this work, we exclusively make use of the second-order product formula
\begin{align}
  \label{eq:second_order_trotter}
  S_2(t) \coloneqq \left(\prod_{\mu=1}^{\stackrel{L}{\leftarrow}}\e^{-\i H_\mu t/(2n_T)}\prod_{\mu=1}^{\stackrel{L}{\rightarrow}}\e^{-\i H_\mu t/(2n_T)}\right)^{n_T},
\end{align}
where $\prod_{\mu=1}^{\stackrel{L}{\leftarrow}} x_\mu\coloneqq x_1\cdots x_L$ and $\prod_{\mu=1}^{\stackrel{L}{\rightarrow}}~x_\mu\coloneqq~x_L\cdots~x_1$.

\subsection{Quantum resource requirements}
\label{sec:quantum_resource_requirements}
As we have seen in the previous sections, the pulses which enable fermionic Hamiltonian engineering are found efficiently using a classical algorithm.
In this section, we provide a brief accounting of the quantum resources that the procedure entails.
Namely, we focus on the quantum runtime and the total number of pulses required to implement a desired time evolution.

The optimization of the quantum runtime is already built into the framework and we observe (see \cref{sec:relaxation}) that the efficient relaxation yields near optimal quantum runtimes within the framework of conjugation based Hamiltonian engineering, at the cost of longer classical precomputation times.
The cost of the classical precomputation is mitigated by the fact that it needs to be performed only once to determine a valid decomposition of the target Hamiltonian and that its output can be stored to simulate the same Hamiltonian as many times as one wishes.

Next, we estimate the total number of pulses.
Once a decomposition of the target Hamiltonian in terms of the system Hamiltonian $H_T = \sum_{\mu=1}^L \lambda_{\ind{b}_\mu} V_{\ind{b}_\mu}^\dagger H_S V_{\ind{b}_\mu}$ is found, where $\ind{b}_\mu\in\TT^n$, the engineered time evolution operator is given by the second-order Trotter formula
\begin{align}
  \label{eq:engineered_evolution_operator}
  U_\mathrm{eng}(t) = \left(\prod_{\mu=1}^{\stackrel{L}{\leftarrow}} V_{\ind{b}_\mu}^\dagger \e^{\frac{-\i H_S \lambda_{\ind{b}_\mu}t}{2n_T}}V_{\ind{b}_\mu} \prod_{\mu=1}^{\stackrel{L}{\rightarrow}} V_{\ind{b}_\mu}^\dagger \e^{\frac{-\i H_S \lambda_{\ind{b}_\mu}t}{2n_T}}V_{\ind{b}_\mu}\right)^{n_T}.
\end{align}
During an experimental implementation of the framework, pulses that follow each other without an intermediate free evolution in between can be combined to a single pulse to reduce the total number of pulses that need to be implemented.
For example, the pulses $V_{\ind{b}_1}$ and $V_{\ind{b}_2}^\dagger$ follow each other in the above decomposition and their product can be replaced by $V_{\ind{b}_1-\ind{b_2}}$. 
Using this, we see that the total number of pulses we need to implement \eqref{eq:engineered_evolution_operator} is at most $2Ln_T+1$.
\subsection{Average Hamiltonian theory and the Magnus expansion}
In this section, we briefly review some standard tools which will be used to mitigate finite-pulse-time errors.
Let $U(t)$ be the time evolution operator for a time-dependent Hamiltonian $H(t)$, described by the Schrödinger equation
\begin{align}
  \frac{d U(t)}{dt} = -\i H(t) U(t),
\end{align}
with the condition $U(0)=\1$.
It is often useful to consider the time-independent \emph{average Hamiltonian} defined as the time-independent Hermitian operator generating the unitary $U(\tau)$, for some predetermined evolution time $\tau$, i.e.
\begin{align}
  \e^{-\i H_{\mathrm{av}}\tau} \coloneqq U(\tau).
\end{align}
The average Hamiltonian can be expressed using the \emph{Magnus expansion} \cite{brinkmannIntroductionAverageHamiltonian2016}
\begin{align}
  H_{\mathrm{av}} = H_{\mathrm{av}}^{(1)}+H_{\mathrm{av}}^{(2)}+\dots
\end{align}
where the first and second order terms are given by
\begin{align}
  H_{\mathrm{av}}^{(1)} \coloneqq \frac{1}{\tau} \int_0^\tau  H(t)dt
\end{align}
and
\begin{align}
   H_{\mathrm{av}}^{(2)} \coloneqq \frac{1}{2\i\tau} \int_0^\tau\int_{0}^{t}[H(t),H(t')]dt'dt.
\end{align}
The Magnus expansion converges if \cite{moanConvergenceMagnusSeries2008}
\begin{align}
  \int_{0}^{\tau} \left\|H(t)\right\| dt < \pi 
\end{align}
and as a rule of thumb, it converges rapidly \cite{brinkmannIntroductionAverageHamiltonian2016} if, for any time $t \in [0,\tau]$
\begin{align}
  \left\|H(t)\right\| \tau \ll 1.
\end{align}

\subsection{Addressing finite-pulse-time errors}
In this section, we show how the fermionic Hamiltonian engineering protocol can be made robust against finite-pulse-time errors. 
Here, we focus on having a concise presentation.
For a more detailed analysis see \cref{ap:fpt_error_matrix}.

Let $S_\ind{b}$ denote the unitary corresponding to the $\ind{b}$-th \emph{step} of the fermionic Hamiltonian engineering procedure, i.e.
\begin{align}
  \label{eq:ideal_step}
  S_\ind{b}^{\mathrm{ideal}}(\lambda_\ind{b}t) = V_\ind{b}^\dagger \e^{-\i H_S \lambda_\ind{b}t} V_\ind{b}.
\end{align}
The direct implementation of this unitary requires the system Hamiltonian to be ``turned off'' while the pulses $V_{\ind b}$ and $V_{\ind b}^\dagger$ are applied. 
More precisely, assuming interaction terms are globally tunable, \cref{eq:ideal_step} only requires that we turn off the tunnelling part of the system Hamiltonian.
This follows from the same argument from \cref{sec:interactions},  that interaction terms commute with the pulses $V_\ind{b}$, and are globally scaled by the engineering protocol.

If we cannot turn off the system Hamiltonian, the pulses are applied ``on top'' of the system Hamiltonian.
The local unitaries \eqref{eq:local_ops_3} we assume access to are generated by the two generators $\left\{\frac{-2\pi}{3t_p}n_j,\frac{2\pi}{3t_p}n_j\right\}$, where $t_p$ is the constant pulse time and $n_j$ is the number operator for mode $j$.
So the \emph{generator} of a pulse $V_\ind{b}$ is given as
\begin{align}
  G_\ind{b}\coloneqq \frac{2\pi}{3t_p} \sum_{j\in[n]} b_j n_j,
\end{align}
where
\begin{align}
  V_\ind{b}(t) \coloneqq \e^{-\i G_\ind{b} t} \quad \text{and} \quad V_\ind{b}\equiv V_\ind{b}(t_p).
\end{align}
Therefore, the $\ind{b}$-th step that is implemented in reality is better described by the unitary
\begin{align}
  \label{eq:realistic_unitary}
    S_\ind{b}(\lambda_\ind{b}t) = \e^{-\i(H_S-G_\ind{b})t_p } \e^{-\i H_S \lambda_\ind{b}t} \e^{-\i(H_S+G_\ind{b})t_p }.
\end{align}

Using the first-order Magnus expansion, we find the effective Hamiltonian $K_\ind{b}$ generating the unitary \eqref{eq:realistic_unitary}, i.e.
\begin{align}
  S_\ind{b}(\lambda_\ind{b}t) = \e^{-\i K_\ind{b} T_\ind{b}},
\end{align}
where $T_\ind{b}\coloneqq \lambda_\ind{b}t + 2t_p$ is the total time of the $\ind{b}$-th step.
We find that
\begin{align}
  \label{eq:fpt_effective_hamiltonian}
  K_\ind{b} = \frac{1}{\lambda_\ind{b}t+2t_p} \left[\sum_{\ind{a}\in\mathcal{T}_n}\alpha_\ind{a}\left(E_{\ind{a}\ind{b}} + F_{\ind{a}\ind{b}}\lambda_\ind{b}t\right) C_\ind{a}\right],
\end{align}
where the \emph{error matrix} $E$ is given by
\begin{align}
  \label{eq:fpt_error_matrix}
  E_{\ind{a}\ind{b}} = 2 t_p \e^{\i\pi(\ind{a}\cdot\ind{b})/3}  \operatorname{sinc}{\left(\frac{\pi}{3}(\ind{a}\cdot\ind{b})\right)},
\end{align}
with $\operatorname{sinc}(x)\coloneqq\sin(x)/x$, and $F_\ind{ab} = \omega^{\ind{a}\cdot\ind{b}}$ as before.
Note that $\operatorname{sinc}(0) = 1$, so that $E_{\ind a\ind b}=2t_p$ for $\ind a\cdot\ind b=0$.

For the Hamiltonian engineering task we require
\begin{align}
  \prod_{\ind{b}} S_\ind{b}(\lambda_\ind{b}t) = \prod_\ind{b}\e^{-\i K_\ind{b} T_\ind{b}} \approx \e^{-\i\sum_\ind{b} K_\ind{b} T_\ind{b}} \stackrel{!}{=} \e^{-\i H_T \tau},
\end{align}
which we ensure by setting
\begin{align}
  \sum_\ind{b} T_\ind{b} K_\ind{b}  = \tau H_T .
\end{align}
Plugging in \cref{eq:fpt_effective_hamiltonian} and writing the resulting expression in matrix language we get
\begin{align}
  \label{eq:fpt_constraint_matrix}
  tF'\boldsymbol{\lambda}+E\boldsymbol{1} = \tau \boldsymbol{\gamma}',
\end{align}
which may be rewritten as
\begin{align}
  F'\ind{\lambda} = \ind{\eta},
\end{align}
where we defined $\ind{\eta}\coloneqq 1/t(\tau \ind{\gamma}' - E\ind{1})\in\CC^r$. 
Crucially, $\supp(\ind{\eta}) = \supp(\ind{\alpha})$, so \cref{thm:main} applies to this case as well and the linear program
\begin{align}
\begin{split}
\label{eq:LP_fpt}
\text{minimize} \quad & \boldsymbol{1}^\top \boldsymbol{\lambda} \\
\text{subject to} \quad &  F' \boldsymbol{\lambda} = \boldsymbol{\eta}, \\
                        &  \boldsymbol{\lambda} \in \RR_{\geq 0}^{c},
\end{split}
\end{align}
has an optimal solution, whenever $\mathrm{supp}(\boldsymbol{\beta})\subseteq\mathrm{supp}(\boldsymbol{\alpha})$.
Consequently, finite-pulse-time errors can be systematically compensated for within the Hamiltonian engineering framework by solving the modified linear program \eqref{eq:LP_fpt}.

The term $E\ind{1}$ assumes that every pulse in the considered set $\mathcal{B}$ of pulses is applied, which in general requires exponentially many pulses to be implemented $c=|\mathcal{B}|=3^n$.
In practice, this problem can be avoided by selecting the subset of pulses to be considered by efficient relaxation before the \ac{LP} \eqref{eq:LP_fpt} is run, and then applying every pulse in this subset, even those with $\lambda_\ind{b}=0$.

\section{Conclusion and outlook}
\label{sec:conclusion}

We introduced a framework for fermionic Hamiltonian engineering that synthesizes the time evolution under a target Hamiltonian by interleaving free evolution under a native system Hamiltonian with discrete sequences of local fermionic unitaries generated by number operators.
We find that the sequence of pulses necessary for the simulation task can be found efficiently by a linear program.
In particular, we have shown that translation symmetry of the Hamiltonian engineering task can be used to make the computational cost of finding the pulses independent of the number of modes, and illustrated it with the Harper--Hofstadter model, for which a six-pulse sequence is provably optimal on any lattice size.

For number-operator-free Hamiltonians, the framework realizes any target Hamiltonian with terms defined on a subset of the terms of the native Hamiltonian, with arbitrary complex coefficients on each term. 
This makes the engineering of artificial gauge fields a relevant use case for the framework, which we illustrated by mapping a uniform square lattice onto the Harper--Hofstadter model and a uniform triangular lattice onto a model with $\pi$-flux per plaquette; in the latter case the pulse sequence can be written down analytically and bypasses the linear program altogether.
We further proved a no-go result for interaction terms that are polynomials of number operators, and showed that no conjugation by local unitaries can do more than globally rescale such terms. 
However, we have shown that this restriction does not prevent the engineering of physically relevant interacting models.
Whenever the interaction strength of the native Hamiltonian is itself globally tunable -- for example via Feshbach resonances for ultracold atoms in optical lattices -- the global rescaling factor produced by the protocol can be absorbed into the choice of system parameters.
As concrete applications, we showed that the $U/t$ ratio of the Fermi--Hubbard model can be increased independently of the native lattice parameters, decoupling the interaction-to-tunnelling ratio from the lattice depth, and that one-dimensional Fermi--Hubbard chains with arbitrary complex tunnellings can be simulated to high accuracy.

We analyzed the two dominant error sources of the protocol. 
Trotter error is controlled by standard product-formula bounds.
Finite-pulse-time errors, which appear whenever the native Hamiltonian cannot be switched off during the application of a local pulse, were treated within average Hamiltonian theory.
To first order in the Magnus expansion they can be systematically compensated for by a corrective modification of the linear constraints, leaving the structure of the linear program and its solvability guarantees unchanged.

These results demonstrate that pulse-based fermionic Hamiltonian engineering is a versatile approach for analogue quantum simulation on current hardware.
It provides access to quenches to interacting Hamiltonians with complex tunnellings in the regime where classical simulation becomes infeasible, making it a promising route toward early practical quantum advantage, and motivate several directions for future work.
The most natural next step is to seek an experimental implementation combining artificial gauge fields with on-site interactions. 
The Hofstadter--Hubbard model is a particularly attractive target, since the non-interacting Hofstadter Hamiltonian has been realized with ultracold atoms~\cite{aidelsburgerRealizationHofstadterHamiltonian2013}, but reaching the strongly interacting fermionic regime has so far been challenging.
Such an experimental implementation would also motivate tailoring the protocol to real, noisy hardware, providing the setting for a rigorous theoretical and experimental comparison with Floquet engineering.
Another open question is how much the choice of phase $2\pi/3$, made to minimize the size of the \ac{LP}, costs us in quantum runtime.
Preliminary numerics suggest no improvement on average from using higher roots of unity, but it would be interesting to know how different phases could be tailored to the specific Hamiltonian engineering problem at hand to yield shorter quantum runtimes. Finally, the efficient relaxation introduced by \cite{basslerGeneralEfficientRobust2025a}, on which our efficiency claims rest, is still only a numerical observation.
It would be interesting to prove that uniformly sampling from the columns of a row-restricted Fourier submatrix like $F'$ yields a sharp phase transition analogous to the one in Wendel's theorem \cite{wendelProblemGeometricProbability1962}.

\begin{acknowledgments} 
We are grateful to Pascal Baßler, Bruno Murta, Mirko Arienzo, Marcus Meschede and Thomas Friese for fruitful discussions on the project.
This work has been funded by 
the TouQan project within the QuantERA~II Programme that has received funding from the EU's H2020 research and innovation programme under Grant Agreement No.\ 101017733,
and from the Deutsche Forschungsgemeinschaft (DFG, German Research Foundation) under the grant number 532779266;
by the German Federal Ministry of Research, Technology and Space (BMFTR) within the funding program ``quantum technologies -- from basic research to market'' via the joint project
MIQRO (grant number 13N15522); 
and 
by the Fujitsu Germany GmbH and Dataport as part of the endowed professorship ``Quantum Inspired and Quantum Optimization.''
The project is
co-financed by ERDF of the European Union and by “Fonds
of the Hamburg Ministry of Science, Research, Equalities
and Districts (BWFGB).”

Claude Opus~4.8, Claude Opus~5, Claude Opus~5.5 (Anthropic) were used to assist with the implementation, refactoring and debugging of the numerical code used in this work, and to compile summaries of related literature. 
Discussions with GPT-6 Astra (OpenAI) inspired the use of translation symmetry to reduce the computational cost of engineering Hamiltonians and assisted with the optimality proof of the quantum runtime for the Harper--Hofstadter model (\cref{sec:HH_optimal_pulses_appendix}).
Claude Fable 5.1 (Anthropic) first suggested the explicit pulses discussed in \cref{sec:HH_optimal_pulses_appendix} for the simulation of the Harper--Hofstadter model.
All AI-generated code, and technical and literature claims were confirmed by the authors, and the authors take full responsibility for the content of this manuscript.
\end{acknowledgments}

\bibliography{HE.bib,mk.bib}

\clearpage
\appendix

\onecolumngrid

\section*{Appendix}
\section{Conjugation of annihilation and creation operators via local unitaries}
\label[appendix]{sec:conjugation_of}
We make use of the  Campbell identity 
\begin{align}
  \label{eq:Campbell}
  \e^{X}Y\e^{-X}=\sum_{m=0}^{\infty} \frac{[X,Y]_m}{m!}
\end{align}
where $[X,Y]_{m}\coloneqq [X,\cdots[X,[X,Y]]\cdots]$, with the operator $X$ being repeated $m$ times, and $[X,Y]_{0}\coloneqq Y$, to show that
\begin{align}
  \begin{split}
  V_j^\dagger(\theta) \ c_k \ V_j(\theta) =& \e^{\i\theta n_j}\ c_k \ \e^{-\i\theta n_j}\\
                                        =& \sum_{m=0}^\infty \frac{(\i\theta)^m}{m!}[n_j,c_k]_m\\
                                        =& \sum_{m=0}^\infty \frac{(-\i\theta  \delta_{jk})^m}{m!}c_k\\
                                        =& \left(\e^{-\i\theta}\right)^{\delta_{jk}}c_k.
  \end{split}
\end{align}
\section{Simple effective Hamiltonians}
\label[appendix]{sec:simple_examples_appendix}
In this section, we focus on two example pulse sequences and investigate their impact on the generated effective Hamiltonians.
Note that in the examples we consider, the connectivity graphs $G_S$ and $G_T$ of the system and target Hamiltonians are the same, therefore we write $E_T = E_S = E$ and $G = (V,E)$.

\paragraph*{Example 1:}
The first example shows how to impart complex phases on tunnelling coefficients, which allows one to investigate artificial gauge fields \cite{struckTunableGaugePotential2012,cooperTopologicalBandsUltracold2019}.
To this end, we conjugate time evolution generated by the system Hamiltonian with the local unitary $V_j(\theta)$
\begin{align}
  V_j(\theta)^\dagger \e^{-\i H_S t} V_j(\theta)  =   \e^{-\i V_j(\theta)^\dagger H_S V_j(\theta) t}.
\end{align}
This unitary can be interpreted as the time evolution operator of an effective Hamiltonian 
\begin{equation}
  \begin{aligned}
  H_\mathrm{eff} &= V_j(\theta)^\dagger H_S V_j(\theta)\\
      &= \alpha\sum_{\{k,l\}\in E}  V_j(\theta)^\dagger c_k^\dagger c_l V_j(\theta) + \mathrm{H. c.}\\
      &= \alpha\sum_{\{k,l\}\in E}  \e^{\i\theta\delta_{jk}}\e^{-\i\theta\delta_{jl}}c_k^\dagger c_l + \mathrm{H. c.}\\
      &= \alpha\sum_{m\in N(j)} \left( \e^{\i\theta} c_j^\dagger c_m + \e^{-\i\theta} c_m^\dagger c_j\right)
          + \alpha\sum_{\{k,l\}\in E \colon k\neq j\neq l } \left(c_k^\dagger c_l + \mathrm{H. c.}\right),
  \end{aligned}
\end{equation}
where $N(j)\coloneqq \Set{m\in V \given  \Set{j,m}\in E}$ is the set of vertices neighboring $j$.
This means that we can impart phases $\e^{\i\theta}$ on the tunnelling terms for selected vertices, as illustrated in  \cref{fig:FHE_phase_engineering} on a simple 1D chain.
In particular for $\theta = \pi$, we obtain an effective Hamiltonian where tunnelling terms on edges incident to vertex $j$ acquire a minus sign, corresponding to an inverted tunnelling strength.

\begin{figure*}[t]
\centering

\subfloat[Phase engineering via a local unitary.\label{fig:FHE_phase_engineering}]{
\resizebox{0.47\textwidth}{!}{%
\begin{tikzpicture}[
  mode/.style={circle, draw, thick, minimum size=7mm, inner sep=0pt},
  hop/.style={thick},
  arrow/.style={thick, ->, >=stealth},
  highlight/.style={circle, draw=red, thick, minimum size=11mm}
]

  \node[mode] (L1) at (-5,0) {$1$};
  \node[mode] (L2) at (-3.5,0) {$2$};
  \node[mode] (L3) at (-2,0) {$3$};

  \node[highlight] at (L2.center) {};

  \draw[hop] (L1) -- (L2);
  \draw[hop] (L2) -- (L3);

  \draw[arrow, bend right=35] (L2) to (L1);
  \draw[arrow, bend right=35] (L3) to (L2);

  \node at (-4.25,0.6) {$\alpha$};
  \node at (-2.75,0.6) {$\alpha$};

  \draw[thick, red] (-3.5,1.2) -- (-3.7,0.8) -- (-3.4,0.9) -- (-3.6,0.5);
  \node[red] at (-3.5,1.5) {$V_2(\theta)$};

  \draw[very thick, ->] (-0.8,0) -- (0.8,0);
  \node at (0,0.6) {\textsc{FHE}};

  \node[mode] (R1) at (2,0) {$1$};
  \node[mode] (R2) at (3.5,0) {$2$};
  \node[mode] (R3) at (5,0) {$3$};

  \draw[hop] (R1) -- (R2);
  \draw[hop] (R2) -- (R3);

  \draw[arrow, bend right=35] (R2) to (R1);
  \draw[arrow, bend right=35] (R3) to (R2);

  \node at (2.75,0.6) {$\alpha e^{-i\theta}$};
  \node at (4.25,0.6) {$\alpha e^{i\theta}$};

\end{tikzpicture}
}}
\hfill
\subfloat[Independent tuning of real tunnellings.\label{fig:FHE_chain}]{
\resizebox{0.47\textwidth}{!}{%
\begin{tikzpicture}[
  mode/.style={circle, draw, thick, minimum size=7mm, inner sep=0pt},
  hop/.style={thick},
  arrow/.style={thick, ->, >=stealth}
]

  \node[mode] (L1) at (-5,0) {$1$};
  \node[mode] (L2) at (-3.5,0) {$2$};
  \node[mode] (L3) at (-2,0) {$3$};

  \draw[hop] (L1) -- (L2);
  \draw[hop] (L2) -- (L3);

  \draw[arrow, bend right=35] (L2) to (L1);
  \draw[arrow, bend right=35] (L3) to (L2);

  \node at (-4.25,0.6) {$\alpha$};
  \node at (-2.75,0.6) {$\alpha$};

  \draw[very thick, ->] (-0.8,0) -- (0.8,0);
  \node at (0,0.6) {\textsc{FHE}};

  \node[mode] (R1) at (2,0) {$1$};
  \node[mode] (R2) at (3.5,0) {$2$};
  \node[mode] (R3) at (5,0) {$3$};

  \draw[hop] (R1) -- (R2);
  \draw[hop] (R2) -- (R3);

  \draw[arrow, bend right=35] (R2) to (R1);
  \draw[arrow, bend right=35] (R3) to (R2);

  \node at (2.75,0.6) {$\beta_{12}$};
  \node at (4.25,0.6) {$\beta_{23}$};

\end{tikzpicture}
}}

\caption{Schematic illustrations of fermionic Hamiltonian engineering (FHE). (a) Conjugation by a local unitary $V_2(\theta)$ induces opposite phases on adjacent tunnellings, yielding $\alpha e^{\pm i\theta}$. (b) Pulse sequences enable independently tunable real coefficients $\beta_{12},\beta_{23}\in\mathbb{R}$ starting from uniform $\alpha\in\mathbb{R}$.}
\label{fig:FHE_examples}

\end{figure*}
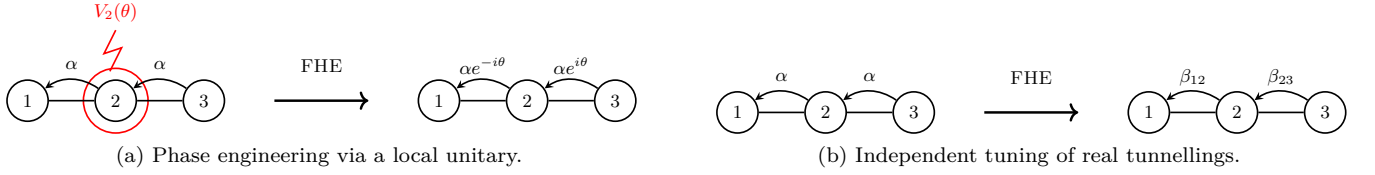
\paragraph*{Example 2:}
For the second example, we show how to tune the relative strength of two tunnelling terms in the effective Hamiltonian by making use of more than one local unitary.
This amounts to a simple demonstration of our Hamiltonian engineering method.
We consider $n=3$ fermionic modes on a one-dimensional chain, as indicated in \cref{fig:FHE_chain}, 
\begin{align}
  H_S = \alpha \left(c_1^\dagger c_2 + c_2^\dagger c_3\right) + \mathrm{H.c.}
\end{align}
with $\alpha\in\RR$.
 Implementing conjugations described by the following sequence of operators $\{\1,V_1(\pi),V_2(\pi),V_3(\pi)\}$ we obtain
\begin{align}
  \begin{split}
  & \e^{-\i H_S\lambda_0 t}\prod_{j=1}^{3}V_j(\pi)^\dagger \e^{-\i H_S \lambda_jt} V_j(\pi)  \\
  &\approx  
  \exp\left(-\i\lambda_0t H_S-\i t\sum_{j=1}^{3} \lambda_j V_j(\pi)^\dagger H_S  V_j(\pi)\right)
  \\
  &\equiv 
  \e^{-\i H_\mathrm{eff}t},
  \end{split}
\end{align}
where the approximation can be made precise with a product formula, like the Trotter-Suzuki formula \cite{trotterProductSemiGroupsOperators1959,dalzellQuantumAlgorithmsSurvey2023}.
Due to $V_j(\pi)^\dagger \cdag_k c_l V_j(\pi)= (-1)^{\delta_{jk}+\delta_{jl}}\cdag_k c_l$, 
the effective Hamiltonian is
\begin{align}
  \begin{split}
  H_\mathrm{eff} &=\alpha\Big[ \left(\lambda_0 - \lambda_1 - \lambda_2 + \lambda_3\right) c_1^\dagger c_2 + \left(\lambda_0 + \lambda_1 - \lambda_2 - \lambda_3\right) c_2^\dagger c_3  \Big] + \mathrm{H.c.}
  \end{split}
\end{align}
Therefore, by appropriately choosing the free evolution times $\lambda_j\geq 0$, we can simulate time evolution under any \emph{target Hamiltonian} of the form
\begin{align}
  H_T = \beta_1 c_1^\dagger c_2 + \beta_2 c_2^\dagger c_3 + \mathrm{H.c.}, 
\end{align}
where $\beta_1,\beta_2\in\RR$. 
A suitable choice for $\boldsymbol{\lambda} = (\lambda_0,\lambda_1,\lambda_2,\lambda_3)^\top$ is obtained by solving the linear system
\begin{align}
  \label{eq:1D_chain_constraint}
  \alpha
  \begin{pmatrix}
    1 & -1 & -1 & \phantom{+}1\\
    1 & \phantom{+}1 & -1 & -1
  \end{pmatrix}
  \boldsymbol{\lambda}
  =
  \begin{pmatrix}
    \beta_1\\
    \beta_2
  \end{pmatrix},
\end{align}
which gives us the constraint equation for the \ac{LP} \eqref{eq:LP_quadratic}.

\section{Hamiltonian engineering using spatially periodic pulses}
\label[appendix]{sec:symmetry_reduction}
Let $\{\ind{e}_j\}_{j\in[D]}$ denote the standard basis vectors of the $D$-dimensional integer lattice $\ZZ^D$, for some $D\geq1$.
We label directed nearest-neighbour edges by their site of origin and direction,
\begin{align}\label{eq:tiE}
  E \coloneqq \{(\ind{q}+\ind{e}_j,\ind{q}) \;|\; \ind{q}\in\ZZ^D, j\in[D]\}.
\end{align}
Consider particle-number-preserving quadratic system and target Hamiltonians on $\ZZ^D$ with connectivity $E$, given by
\begin{align}
  H_S = \sum_{(\ind{p},\ind{q})\in{E}} \alpha_{\ind{p},\ind{q}}c_\ind{p}^\dagger c_\ind{q} \quad \text{and} \quad H_T = \sum_{(\ind{p},\ind{q})\in{E}} \beta_{\ind{p},\ind{q}}c_\ind{p}^\dagger c_\ind{q} \quad \text{with}\quad \gamma_{\ind{p},\ind{q}} \coloneqq \beta_{\ind{p},\ind{q}}/\alpha_{\ind{p},\ind{q}} \,,
\end{align}
where we assume $\alpha_{\ind{p},\ind{q}}\neq 0$ for all $(\ind{p},\ind{q})\in E$ so that $\gamma_{\ind{p},\ind{q}}$ is well-defined.

For the remainder of this section, we will switch to the edge labelling used in the definition \eqref{eq:tiE} of $E$, and write $\gamma^j_\ind{q} \coloneqq \gamma_{\ind{q}+\ind{e}_j, \ind{q}}$.
We treat $\gamma^j: \ZZ^D \to \CC, \ind{q} \mapsto \gamma^j_\ind{q}$, for each $j \in [D]$, as a complex-valued function on the lattice.
Such functions live in a vector space on which we define some useful linear operators.
\begin{definition}[Translations and forward differences]\label[definition]{def:translations}
  For every $\ind\delta \in \ZZ^D$, we define the \emph{translation operator $T_\ind{\delta}$} acting on functions $f: \ZZ^D \to \CC$ componentwise,
  \begin{equation}
    \forall \ind{q} \in \ZZ^D: \qquad
    [T_\ind{\delta} f](\ind{q}) \coloneqq f(\ind{q}+\ind{\delta}) \,.
  \end{equation}
  Moreover, for every $j \in [D]$, we define the \emph{forward difference operator $\partial_j \coloneqq T_{\ind{e}_j} - T_\ind{0} = T_{\ind{e}_j} - \operatorname{Id}$}.
\end{definition}
It is easy to check that the $\{ T_\ind{\delta} | \ind{\delta} \in \ZZ^D \}$ form a representation of the Abelian group $\ZZ^D$, so in particular they commute pairwise.
Moreover, it also follows immediately that $T_\ind{\delta} \partial_j = \partial_j T_\ind{\delta}$ for all $\ind{\delta} \in \ZZ^D$ and $j \in [D]$.

Let $\ind{\ell} \in \ZZ_{\geq3}^D$ be an integer vector whose entries are all at least $3$ (see the proof of \Cref{thm:torus_to_global} for details), defining a subgroup of $\ZZ^D$,
\begin{equation}
  L \coloneqq \bigtimes_{i=1}^D (\ell_i \ZZ) \,.
\end{equation}
It consists of vectors whose $i$-th entry is an integer multiple of $\ell_i$.
The sublattice $L$ tessellates the full lattice $\ZZ^D$ into cells of size $n_L \coloneqq \prod_{i=1}^D \ell_i$.
We get from a site in one cell to the corresponding sites in other cells using translations $T_\ind{\delta}$ for $\ind{\delta} \in L$.
The quotient $\ZZ^D_L \coloneqq \ZZ^D / L$ is a torus whose $|\ZZ^D_L| = n_L$ elements are cosets of such translation-related families of sites.
\begin{definition}
  We say that the Hamiltonian engineering problem $H_S \to H_T$ is \emph{$L$-periodic}, if all its ratios are invariant under $L$-translations.
  That is, for all $j \in [D]$ and $\ind{\delta} \in L$
  \begin{equation}
    \gamma^j = T_\ind{\delta} \gamma^j \,,
  \end{equation}
  or in terms of components, for all $\ind{q} \in \ZZ^D$,
  \begin{equation}
    \gamma^j_\ind{q} = \gamma^j_{\ind{q}+\ind{\delta}} \,.
  \end{equation}
\end{definition}
Note that the system and target coefficients themselves need not separately have this periodicity and we require that only the ratios are periodic.

Consider $\TT^{\ZZ^D} = \{-1,0,1\}^{\ZZ^D}$, the set of labels of pulses acting on the full lattice $\ZZ^D$.
Apart from the fact that the infinite lattice $\ZZ^D$ in itself is practically impossible (and just serves as a model for large systems, see below), any experimental implementation can also only use a finite number of pulses, while the set $\TT^{\ZZ^D}$ is uncountable.
We will see in a moment, however, that the finite subset of $L$-periodic pulses,
\begin{equation}
  \mc{B}_L \coloneqq \left\{ \ind{b} \in \TT^{\ZZ^D} \;\middle|\; \forall \ind{\delta} \in L: \, T_\ind{\delta} \ind{b} = \ind{b} \right\} \,,
\end{equation}
with size $|\mc{B}_L| = |\TT|^{|\ZZ^D_L|} = 3^{n_L}$, suffices to solve arbitrary $L$-periodic Hamiltonian engineering problems.
Here, the action of translation operators is extended to $\TT$-valued functions in the obvious way, $[T_\ind{\delta} \ind{b}]_\ind{q} = b_{\ind{q}+\ind{\delta}}$.

\begin{lemma}\label[lemma]{thm:periodic_to_torus}
  The ratios of an $L$-periodic problem can be treated as class functions, $\gamma^j_{[\ind{q}]} \coloneqq \gamma^j_\ind{q}$ for all $[\ind{q}] \in \ZZ^D_L$, which is unambiguous and independent of the choice of representative due to $L$-periodicity.
  In particular, an $L$-periodic problem in $D$ dimensions only has $D n_L$ independent ratios;
  whereas a general Hamiltonian engineering problem on $\ZZ^D$ has infinitely many.

  The constraint matrix $F$ (in generalization of \cref{eq:def:F}) for the connectivity $E$ on $\ZZ^D$ has entries
  \begin{equation}
    F_{(\ind{q}, j), \ind{b}}
    = \omega^{b_{\ind{q}+\ind{e}_j} - b_\ind{q}}
    = \omega^{[\partial_j \ind{b}]_\ind{q}} \,, \qquad
    \ind{q} \in \ZZ^D \,, \quad
    j \in [D] \,, \quad
    \ind{b} \in \TT^{\ZZ^D} \,.
  \end{equation}
  For fixed $j \in [D]$ and $\tilde{\lambda}: \mc{B}_L \to \RR_{\geq0}$ (note the restriction to $L$-periodic pulses),
  \begin{equation}
    \ind{q}
    \;\longmapsto\;
    \sum_{\ind{b} \in \mc{B}_L} F_{(\ind{q}, j), \ind{b}} \tilde{\lambda}_\ind{b}
  \end{equation}
  is also $L$-periodic, and can be treated as a class function on the torus $\ZZ^D_L$.
\end{lemma}
\begin{proof}
  That periodic functions effectively live on a torus is a well-known fact, which we quickly restate.
  Let $[\ind{q}] \in \ZZ^D_L$ be some coset of sites, and let $\ind{q}, \ind{q}' \in [\ind{q}] = [\ind{q}']$ be two representatives.
  By construction, there must be some $\ind{\delta} \in L$ such that $\ind{q}' = \ind{q} + \ind{\delta}$.
  Then, $\gamma^j_{[\ind{q}]} \coloneqq \gamma^j_{\ind{q}} = \gamma^j_{\ind{q}+\ind{\delta}} = \gamma^j_{\ind{q}'} \eqqcolon \gamma^j_{[\ind{q}']}$ is independent of the choice of representative by the assumed $L$-periodicity.

  The constraint matrix $F$ has one row per operator in the Hamiltonian, so in this context one for each edge in $E$, indexed by pairs $(\ind{q}, j)$.
  For an $L$-periodic pulse $\ind{b} \in \mc{B}_L$, i.e. $T_\ind{\delta} \ind{b} = \ind{b}$ for any $\ind{\delta} \in L$, we find $T_\ind{\delta} (\partial_j \ind{b}) = \partial_j T_\ind{\delta} \ind{b} = \partial_j \ind{b}$ by commutativity of translations and forward differences.
  Thus,
  \begin{equation}
    \sum_{\ind{b} \in \mc{B}_L} F_{(\ind{q}+\ind{\delta}, j), \ind{b}} \tilde{\lambda}_\ind{b}
    = \sum_{\ind{b} \in \mc{B}_L} \omega^{[\partial_j \ind{b}]_{\ind{q}+\ind{\delta}}} \tilde{\lambda}_\ind{b}
    = \sum_{\ind{b} \in \mc{B}_L} \omega^{[T_\ind{\delta} \partial_j \ind{b}]_\ind{q}} \tilde{\lambda}_\ind{b}
    = \sum_{\ind{b} \in \mc{B}_L} \omega^{[\partial_j \ind{b}]_\ind{q}} \tilde{\lambda}_\ind{b}
    = \sum_{\ind{b} \in \mc{B}_L} F_{(\ind{q}, j), \ind{b}} \tilde{\lambda}_\ind{b} \,,
  \end{equation}
  for every $\ind{q} \in \ZZ^D$, $j \in [D]$, $\tilde{\lambda}: \mc{B}_L \to \RR_{\geq0}$, and $\ind{\delta} \in L$.
  Hence, such sums are $L$-periodic in $\ind{q}$, and admit the same conclusions as we drew for the $\gamma^j$.
\end{proof}
This allows us to construct feasible solutions to $L$-periodic Hamiltonian engineering problems.
\begin{proposition}\label[proposition]{thm:torus_to_global}
  Let $H_S \to H_T$ be an $L$-periodic Hamiltonian engineering problem on $\ZZ^D$ with coupling ratios $\ind{\gamma}^j$ for all $j \in [D]$.
  The finite ``toric'' \ac{LP}
  \begin{equation}
    \label{eq:toricLP}\begin{split}
    \text{minimize}\quad & \ind{1}^\top \ind{\tilde{\lambda}} \,, \\
    \text{subject to}\quad & \sum_{\ind{b} \in \mc{B}_L} F_{([\ind{q}],j), \ind{b}} \tilde{\lambda}_\ind{b} = \gamma^j_{[\ind{q}]} \,, \quad
    \forall [\ind{q}] \in \ZZ^D_L \,, \; j \in [D] \,, \\
    & \ind{\tilde{\lambda}} \in \RR_{\geq0}^{\mc{B}_L} \,,
  \end{split}\end{equation}
  has $D n_L$ constraints and $3^{n_L}$ variables.
  Under the assumption $\ind{\ell} \in \ZZ_{\geq3}^D$, this \ac{LP} is universally feasible and can be solved with the methods described in \Cref{ap:lemma_proof}.

  Assume we have solved the above \ac{LP} and obtained some optimal $\ind{\tilde\lambda}^\ast$.
  We extend this solution by zeros to all of $\TT^{\ZZ^D}$, i.e. we define a new
  \begin{equation}
    \lambda: \quad
    \TT^{\ZZ^D} \longrightarrow \RR_{\geq0} \,, \qquad
    \lambda_\ind{b} \coloneqq \begin{cases}
      \tilde\lambda^\ast_\ind{b} \,, & \text{if } \ind{b} \in \mc{B}_L \,, \\
      0 \,, & \text{otherwise.}
    \end{cases}
  \end{equation}
  This yields a feasible solution to the full $L$-periodic Hamiltonian engineering problem,
  \begin{equation}\begin{split}
    \text{minimize}\quad & \ind{1}^\top \ind{\lambda} \,, \\
    \text{subject to}\quad & \sum_{\ind{b} \in \TT^{\ZZ^D}} F_{(\ind{q},j), \ind{b}} \lambda_\ind{b} = \gamma^j_{\ind{q}} \,, \quad
    \forall \ind{q} \in \ZZ^D \,, \; j \in [D] \,, \\
    & \ind{\lambda} \in \RR_{\geq0}^{\TT^{\ZZ^D}} \,, \text{ finitely supported,}
  \end{split}\end{equation}
  which achieves the same objective value $\ind{1}^\top \ind{\lambda} = \ind{1}^\top \ind{\tilde{\lambda}}^\ast$ and sparsity $\| \ind{\lambda} \|_0 = \|\ind{\tilde{\lambda}}^\ast\|_0$.
\end{proposition}
\begin{proof}
  With \cref{thm:periodic_to_torus}, this result is almost immediate.
  Adjoining zeros to $\ind{\tilde{\lambda}}^\ast$ does not change non-negativity, the objective value or sparsity, and $\ind{\lambda}$ is indeed finitely supported.
  The sum in the linear constraints of the full \ac{LP} effectively breaks down to the one from the toric \ac{LP}.
  Any constraint on the lattice, for $\ind{q} \in \ZZ^D$, falls into some coset, $\ind{q} \in [\ind{q}] \in \ZZ^D_L$, and is thus satisfied by \cref{thm:periodic_to_torus} and the assumption that $\ind{\tilde{\lambda}}^\ast$ solves the toric \ac{LP};
  hence $\ind{\lambda}$ satisfies all constraints and is feasible.

  \Cref{ap:lemma_proof} explains in detail why $\ind{\ell}\geq3$ leads to a universally feasible system, so here we merely demonstrate what breaks when some $\ell_i \in \{1,2\}$.
  Fix some $i \in [D]$, and consider the choice $\ell_i=1$.
  The corresponding constraints become
  \begin{equation}
    \gamma^i_{[\ind{q}]}
    = \sum_{\ind{b} \in \mc{B}_L} \omega^{b_{\ind{q}+\ind{e}_i} - b_\ind{q}} \tilde{\lambda}_\ind{b}
    = \sum_{\ind{b} \in \mc{B}_L} \omega^{b_\ind{q} - b_\ind{q}} \tilde{\lambda}_\ind{b}
    = \sum_{\ind{b} \in \mc{B}_L} \tilde{\lambda}_\ind{b}
    \geq 0 \,,
  \end{equation}
  so the toric \ac{LP} is only feasible for non-negative ratios.
  Similarly, choose $\ell_i=2$ for some direction $i \in [D]$.
  Then,
  \begin{equation}
    \left( \gamma^i_{[\ind{q}+\ind{e}_i]} \right)^\ast
    = \left( \sum_{\ind{b} \in \mc{B}_L} \omega^{b_{\ind{q}+2\ind{e}_i} - b_{\ind{q}+\ind{e}_i}} \tilde{\lambda}_\ind{b} \right)^\ast
    = \sum_{\ind{b} \in \mc{B}_L} \omega^{-(b_\ind{q} - b_{\ind{q}+\ind{e}_i})} \tilde{\lambda}_\ind{b}
    = \sum_{\ind{b} \in \mc{B}_L} \omega^{b_{\ind{q}+\ind{e}_i} - b_\ind{q}} \tilde{\lambda}_\ind{b}
    = \gamma^i_{[\ind{q}]} \,,
  \end{equation}
  so the toric \ac{LP} is only feasible if the ratios satisfy this specific alternating-conjugate structure.
  We conclude that we need to choose $\ell_i \geq 3$ for all $i \in [D]$ to obtain a universally feasible toric \ac{LP}.
\end{proof}
Note that solutions constructed with \Cref{thm:torus_to_global} are not necessarily optimal on $\ZZ^D$, and optimality needs to be checked independently. 
See \cref{sec:HH_optimal_pulses_appendix} for a proof of the optimality of the construction for the simulation of the Harper--Hofstadter model.

In practice, the restriction $\ind{\ell} \geq 3$ does not constrain which periodic systems we can handle:
Any $1$-periodic system is also $3$-periodic, and any $2$-periodic system is also $4$-periodic;
choosing $L$ accordingly yields a feasible toric \ac{LP}, albeit slightly larger than expected at first sight for these very short periods.

Furthermore, since $H_S$ is particle number preserving, under conjugation via a pulse $V_\ind{b}$, each term in the decomposition of $H_S$ acquires a phase factor $\omega^{b_\ind{p}-b_\ind{q}}$.
Such phase factors are invariant under a constant global phase shift $\ind{b}\to\ind{b}+a\ind{1}$ as $b_\ind{p}+a-b_\ind{q}-a = b_\ind{p}-b_\ind{q}$.
For this reason, we set $b_{\ind{0}}=0$, where $\ind{0}=(0,\dots,0)\in\ZZ^D$, when considering Hamiltonian engineering for particle number preserving Hamiltonians.
This reduces the number of pulses we consider for the decomposition by a factor of $3$ to $3^{n_L-1}$.

Finally, we address the impracticality of the infinite lattice $\ZZ^D$.
Assume a finite lattice of size $N_1 \times \dots \times N_D$, which we consider as a subset of $\ZZ^D$ with the standard embedding.
The corresponding nearest-neighbour edge set $E_\mathrm{finite}$ is a subset of the edge set $E$ used above.
A Hamiltonian engineering problem $H_S \to H_T$ on this system is called $L$-periodic, if there exists an $L$-periodic problem on $\ZZ^D$ that restricts to the given finite problem.
Constraints correspond to edges, so the restriction to $E_\mathrm{finite} \subset E$ cannot add constraints;
any feasible solution to the infinite extension is also a feasible solution for the finite problem $H_S \to H_T$, using pulses that are the corresponding restrictions of the $L$-periodic pulses used by the extension.
This procedure readily extends to finite systems with periodic boundary conditions, as long as each lattice length $N_i$ is an integer multiple of the corresponding cell period $\ell_i$, and the coupling ratios on the bonds crossing the periodic boundaries agree with those of the $L$-periodic extension.

\section{Six optimal pulses to engineer the Harper--Hofstadter model}
\label[appendix]{sec:HH_optimal_pulses_appendix}

In this section, we provide an explicit description for six pulses which allow us to engineer the dynamics of the Harper--Hofstadter Hamiltonian
\begin{align}
  \label{eq:HH_target_hamiltonian_appendix}
  H_T = -J\sum_{x,y} \left(\omega^{y} c_{x+1,y}^\dagger c_{x,y} + c_{x,y+1}^\dagger c_{x,y} + \mathrm{H.c.}\right),
\end{align}
where $J\in\RR$ and $\omega = \e^{\i 2\pi/3}$, starting from the system Hamiltonian
\begin{align}
  \label{eq:HH_system_hamiltonian_appendix}
  H_S = -J\sum_{x,y} \left(c_{x+1,y}^\dagger c_{x,y} + c_{x,y+1}^\dagger c_{x,y} + \mathrm{H.c.}\right).
\end{align}
Note that these pulses were first suggested by Claude Fable 5.1, and checked both analytically and numerically by the authors. 
Furthermore, the toric \ac{LP} \eqref{eq:toricLP} restricted to the $3\times3$ lattice finds an equivalent but generally different set of six pulses which have the same quantum runtime.

The pulses we consider are labelled by indices $\ind{b}^{(m)},\ind{b}'^{(m)} \in \TT^n$ with components given by 
\begin{align}
  b_{x,y}^{(m)} = y (x+m) \mod 3 \quad \text{and} \quad  b_{x,y}'^{(m)} = xm \mod 3 \quad \text{for} \quad m\in\{0,1,2\},
\end{align}
yielding    
\begin{align}
  \label{eq:HH_analytical_pulses}
  V_{\ind{b}^{(m)}} = \prod_{x,y}\exp{\left(-\i \frac{2\pi}{3} y(x+m)\ n_{x,y}\right)} \quad \text{and} \quad V_{\ind{b}'^{(m)}} = \prod_{x,y}\exp{\left(-\i \frac{2\pi}{3} xm\ n_{x,y}\right)}.
\end{align}
Conjugating the nearest neighbour tunnelling terms which appear in the Hamiltonian by these pulses we obtain
\begin{align}
  V_{\ind{b}^{(m)}}^\dagger c_{x+1,y}^\dagger c_{x,y} V_{\ind{b}^{(m)}} &=  \omega^y c_{x+1,y}^\dagger c_{x,y}\\
  V_{\ind{b}^{(m)}}^\dagger c_{x,y+1}^\dagger c_{x,y} V_{\ind{b}^{(m)}} &=  \omega^{(x+m)} c_{x,y+1}^\dagger c_{x,y}
\end{align}
and
\begin{align}
  V_{\ind{b}'^{(m)}}^\dagger c_{x+1,y}^\dagger c_{x,y} V_{\ind{b}'^{(m)}} &=  \omega^m c_{x+1,y}^\dagger c_{x,y}\\
  V_{\ind{b}'^{(m)}}^\dagger c_{x,y+1}^\dagger c_{x,y} V_{\ind{b}'^{(m)}} &=   c_{x,y+1}^\dagger c_{x,y}.
\end{align}
Therefore for horizontal hopping terms we get
\begin{align}
  \sum_{m=0}^2\left(\lambda_m V_{\ind{b}^{(m)}}^\dagger c_{x+1,y}^\dagger c_{x,y} V_{\ind{b}^{(m)}} + \lambda_m' V_{\ind{b}'^{(m)}}^\dagger c_{x+1,y}^\dagger c_{x,y} V_{\ind{b}'^{(m)}}\right) = \sum_{m=0}^2 \left(\lambda_m \omega^y  + \lambda_m'\omega^m\right) c_{x+1,y}^\dagger c_{x,y},
\end{align}
and setting $\lambda_m = \lambda_m' = 1/3$ for all $m\in\{0,1,2\}$ we get
\begin{align}
  \sum_{m=0}^2\left(\lambda_m V_{\ind{b}^{(m)}}^\dagger c_{x+1,y}^\dagger c_{x,y} V_{\ind{b}^{(m)}} + \lambda_m' V_{\ind{b}'^{(m)}}^\dagger c_{x+1,y}^\dagger c_{x,y} V_{\ind{b}'^{(m)}}\right) = \omega^y c_{x+1,y}^\dagger c_{x,y},
\end{align}
where we used $\sum_{m=0}^2 \omega^m =0$.
Similarly, for vertical hopping terms we get
\begin{align}
  \sum_{m=0}^2\left(\lambda_m V_{\ind{b}^{(m)}}^\dagger c_{x,y+1}^\dagger c_{x,y} V_{\ind{b}^{(m)}} + \lambda_m' V_{\ind{b}'^{(m)}}^\dagger c_{x,y+1}^\dagger c_{x,y} V_{\ind{b}'^{(m)}}\right) = \sum_{m=0}^2 \left(\lambda_m \omega^{x+m}+\lambda_m'\right)c_{x,y+1}^\dagger c_{x,y} = c_{x,y+1}^\dagger c_{x,y}.
\end{align}
Thus, we see that the pulses \eqref{eq:HH_analytical_pulses} enable the simulation of \cref{eq:HH_target_hamiltonian_appendix} starting from \cref{eq:HH_system_hamiltonian_appendix}
\begin{align}\label{eq:optimal_haho_engineering}
  H_T = \frac{1}{3}\sum_{m=0}^2 \left(V_{\ind{b}^{(m)}}^\dagger H_S V_{\ind{b}^{(m)}} + V_{\ind{b}'^{(m)}}^\dagger H_S V_{\ind{b}'^{(m)}}\right),
\end{align}
and the quantum runtime is seen to be $\sum_{m=0}^2 \left(\lambda_m + \lambda_m'\right) = 2$. 
These pulses work independent of the size of the lattice, allowing the simulation of the Harper--Hofstadter model with few pulses and a low quantum runtime overhead on large lattices.

To prove optimality of the solution found above, it suffices to consider a single plaquette.
For simplicity, we choose the ``bottom-left'' $P \coloneqq \{0,1\}^2$.
We will use the fact that every tritstring $\ind{b} \in \TT^n$ labelling a pulse can be split into a part $\ind{b}_P$ supported on the plaquette $P$, and a part $\ind{b}_\Pb$ supported on its complement $\Pb$.
We denote string concatenation by juxtaposition, so $\ind{b} = \ind{b}_P \ind{b}_\Pb$.
The set of pulse labels factors accordingly as $\TT^n = \TT^P \times \TT^\Pb$.

Let
\begin{equation}\label{eq:plaquette_edges}
  (\ind{p}, \ind{q}) \in E_P \coloneqq \Big\{ \big( (1,0), (0,0) \big) \,, \big( (0,1), (0,0) \big) \,, \big( (1,1), (1,0) \big) \,, \big( (1,1), (0,1) \big) \Big\} \subset P^2
\end{equation}
be a pair of nearest-neighbour sites on our chosen plaquette.
Note that for any $\ind{b} \in \TT^n$, both $b_\ind{p}$ and $b_\ind{q}$ are entries of the part $\ind{b}_P$.
The linear constraint for $(\ind{p}, \ind{q})$ in the \ac{LP} reads
\begin{equation}
  \gamma_{(\ind{p}, \ind{q})}
  = \sum_{\ind{b} \in \TT^n} \omega^{b_\ind{p} - b_\ind{q}} \lambda_\ind{b}
  = \sum_{\ind{b}_P \in \TT^P} \omega^{b_\ind{p} - b_\ind{q}} \underbrace{\sum_{\ind{b}_\Pb \in \TT^\Pb} \lambda_{\ind{b}_P \ind{b}_\Pb}}_{\eqqcolon \hat{\lambda}_{\ind{b}_P} \geq 0}
  = \sum_{\ind{b}_P \in \TT^P} \omega^{b_\ind{p} - b_\ind{q}} \hat{\lambda}_{\ind{b}_P} \,,
\end{equation}
where we used that on $P$ the entries of the constraint matrix are independent of $\ind{b}_\Pb$.
The marginal variables $\ind{\hat{\lambda}} \in \RR_{\geq0}^{\TT^P}$ are again non-negative, because they are sums of the non-negative original variables $\ind{\lambda}$.
Note that we conventionally set $b_{(0,0)} = 0$ as mentioned in \cref{sec:symmetry_reduction}, so there are $|\TT|^{|\{(1,0),(0,1),(1,1)\}|} = 27$ non-trivial entries in $\ind{\hat{\lambda}}$.

Thus, the four nearest-neighbour edges bounding $P$ define a convex feasibility problem with $4$ linear constraints and $27$ non-negative variables.
If $\ind{\lambda}$ is a solution to the full \ac{LP}, it induces a marginal $\ind{\hat{\lambda}}$ that is necessarily feasible for the plaquette constraints and non-negative, and the induced variables achieve the same total
\begin{equation}
  \ind{1}^\top \ind{\hat{\lambda}}
  = \sum_{\ind{b}_P \in \TT^P} \hat{\lambda}_{\ind{b}_P}
  = \sum_{\ind{b}_P \in \TT^P} \sum_{\ind{b}_\Pb \in \TT^\Pb} \lambda_{\ind{b}_P \ind{b}_\Pb}
  = \sum_{\ind{b} \in \TT^n} \lambda_\ind{b}
  = \ind{1}^\top \ind{\lambda} \,.
\end{equation}
We conclude that the \emph{plaquette \ac{LP}}
\begin{equation}\label{eq:plaquette_lp}\begin{split}
  \text{minimize}\quad & \ind{1}^\top \ind{\hat{\lambda}} \,, \\
  \text{subject to}\quad & \sum_{\ind{b}_P \in \TT^P} \omega^{b_\ind{p} - b_\ind{q}} \hat{\lambda}_{\ind{b}_P} = \gamma_{(\ind{p}, \ind{q})} \,, \quad
  \forall (\ind{p}, \ind{q}) \in E_P \,, \\
  & \ind{\hat{\lambda}} \in \RR_{\geq0}^{\TT^P} \,,
\end{split}\end{equation}
has its optimum upper-bounded by the optimum of the full \ac{LP}, or conversely, the optimum of the plaquette \ac{LP} lower-bounds the full \ac{LP}.

\begin{lemma}\label[lemma]{thm:haho_dual_certificate}
  The vector
  \begin{equation}
    \ind{\tilde{\mu}}^\ast \coloneqq \left( \frac{1}{2} \,, \frac{1}{2} \,, \frac{1}{2} \,, -\frac{1}{2} \,, -\frac{1}{2\sqrt{3}} \,, \frac{1}{2\sqrt{3}} \,, -\frac{1}{2\sqrt{3}} \,, \frac{1}{2\sqrt{3}} \right)^\top \,,
  \end{equation}
  is a feasible solution of the dual \ac{LP} of the (real version of the) plaquette \ac{LP} (\ref{eq:plaquette_lp}), when edges are ordered as in \Cref{eq:plaquette_edges}.
  It achieves objective value $2$.
\end{lemma}
\begin{proof}
  We order the three non-trivial sites in $P$ as $(1,0), (0,1), (1,1)$, and the pulse labels $\TT^P$ in the ensuing lexicographic order.
  The constraint matrix of the plaquette \ac{LP} is
  \begin{minipage}{\textwidth}
    \setcounter{MaxMatrixCols}{30}
    \begin{equation}
      F_P = \begin{pmatrix}
        1 & 1 & 1 & 1 & 1 & 1 & 1 & 1 & 1 & \omega & \omega & \omega & \omega & \omega & \omega & \omega & \omega & \omega & \omega^2 & \omega^2 & \omega^2 & \omega^2 & \omega^2 & \omega^2 & \omega^2 & \omega^2 & \omega^2 \\
        1 & 1 & 1 & \omega & \omega & \omega & \omega^2 & \omega^2 & \omega^2 & 1 & 1 & 1 & \omega & \omega & \omega & \omega^2 & \omega^2 & \omega^2 & 1 & 1 & 1 & \omega & \omega & \omega & \omega^2 & \omega^2 & \omega^2 \\
        1 & \omega & \omega^2 & 1 & \omega & \omega^2 & 1 & \omega & \omega^2 & \omega^2 & 1 & \omega & \omega^2 & 1 & \omega & \omega^2 & 1 & \omega & \omega & \omega^2 & 1 & \omega & \omega^2 & 1 & \omega & \omega^2 & 1 \\
        1 & \omega & \omega^2 & \omega^2 & 1 & \omega & \omega & \omega^2 & 1 & 1 & \omega & \omega^2 & \omega^2 & 1 & \omega & \omega & \omega^2 & 1 & 1 & \omega & \omega^2 & \omega^2 & 1 & \omega & \omega & \omega^2 & 1 \\
      \end{pmatrix} \,,
    \end{equation}
  \end{minipage}
  the right-hand side vector is $\ind{\gamma}_P = (1,1,1,\omega)^\top$.
  The corresponding objects of the real version (see proof of \Cref{lem:main_appendix}) are
  \begin{equation}
    \tilde{F}_P = \begin{pmatrix}
      \Re(F_P) \\ \Im(F_P)
    \end{pmatrix} \in \RR^{8 \times 27} \,, \qquad
    \tilde{\ind{\gamma}}_P = \begin{pmatrix}
      \Re(\ind{\gamma}_P) \\ \Im(\ind{\gamma}_P)
    \end{pmatrix} \in \RR^8 \,.
  \end{equation}
  The dual \ac{LP} of the real version is
  \begin{equation}\begin{split}
    \text{maximize}\quad & \tilde{\ind{\gamma}}_P^\top \ind{\tilde{\mu}} \,, \\
    \text{subject to}\quad & \tilde{F}_P^\top \ind{\tilde{\mu}} \leq \ind{1} \,, \\
    & \ind{\tilde{\mu}} \in \RR^8 \,.
  \end{split}\end{equation}
  We leave it to the reader to verify that $\ind{\tilde{\mu}}^\ast$ from the lemma is indeed a feasible solution.
  The objective value achieved by $\ind{\tilde{\mu}}^\ast$ is
  \begin{equation}
    \tilde{\ind{\gamma}}_P^\top \ind{\tilde{\mu}}^\ast
    = \left( 1, 1, 1, -\frac{1}{2}, 0, 0, 0, \frac{\sqrt{3}}{2} \right) \ind{\tilde{\mu}}^\ast
    = \frac{1}{2} + \frac{1}{2} + \frac{1}{2} + \frac{1}{4} + 0 + 0 + 0 + \frac{1}{4}
    = 2 \,.
  \end{equation}
\end{proof}
With this dual certificate, we can assert optimality of the Harper--Hofstadter implementation found in this section.
\begin{theorem}
  The Hamiltonian engineering solution for the Harper--Hofstadter model shown in \Cref{eq:optimal_haho_engineering} is optimal.
\end{theorem}
\begin{proof}
  This is a concatenation of (in)equalities.
  Let $\ind{\lambda}^\ast$ be an optimal solution for the Hamiltonian engineering problem $H_S \to H_T$ studied in this section, and let $\ind{\hat{\lambda}}^\ast$ be an optimal solution for the plaquette \ac{LP} (\ref{eq:plaquette_lp}).
  Then, using (in this order) \Cref{thm:haho_dual_certificate}, weak duality of \acp{LP} \cite{boydConvexOptimization2004}, the argument that the plaquette optimum lower-bounds the full optimum, and the objective value witnessed by \Cref{eq:optimal_haho_engineering}, we find
  \begin{equation}
    2
    = \tilde{\ind{\gamma}}_P^\top \ind{\tilde{\mu}}^\ast
    \leq \ind{1}^\top \ind{\hat{\lambda}}^\ast
    \leq \ind{1}^\top \ind{\lambda}^\ast
    \leq 2 \,.
  \end{equation}
  We conclude that all inequalities are equalities, and \Cref{eq:optimal_haho_engineering} is indeed an optimal solution.
\end{proof}

\section{Simulating fermionic Hamiltonian engineering for particle number preserving quadratic Hamiltonians} 
\label[appendix]{sec:Harper--Hofstadter_appendix}

Suppose we are given the Hamiltonian engineering task for particle-number-preserving quadratic system and target Hamiltonians $H_S = \sum_{jk} \alpha_{jk} c_j^\dagger c_k$ and $H_T = \sum_{jk} \beta_{jk} c_j^\dagger c_k$.
As seen in \cref{sec:FHE_general}, the \ac{LP} \eqref{eq:LP} outputs a set of pulses $\{V_\ind{b}\}_\ind{b}$ and evolution times $\{\lambda_\ind{b}\}_\ind{b}$ that can be implemented on hardware.
Lacking the quantum hardware to test the method directly, we resort to numerical experiments.
For interacting systems, these numerical experiments are performed using exact diagonalization due to the transparency and the simplicity of the method, but this drastically limits the number of modes we can reach.
However, we can simulate quadratic Hamiltonians efficiently, since their evolution reduces to that of the $n \times n$ coefficient matrix $h$ for the particle number preserving case and avoids having to deal with the high-dimensional
Fock space.
It is this reduction to the single-particle picture that enables the simulation of the 1088-mode Harper--Hofstadter model.

\subsection{Particle number preserving quadratic Hamiltonians}
\label[appendix]{sec:particle_number_preserving_quadratic}

We begin by discussing some of the basics of particle number preserving quadratic Hamiltonians.
For a comprehensive review of numerical approaches to fermionic quadratic Hamiltonians see \cite{suraceFermionicGaussianStates2022}. 
Let $H$ be a particle number preserving quadratic Hamiltonian with $n$ modes, i.e.
\begin{align}
  H = \sum_{i,j\in[n]} h_{ij} c_i^\dagger c_j
\end{align}
with $n\times n$ \emph{coefficient matrix} $h\in\Herm(\CC^n)$ to ensure Hermiticity of $H$.
Let $\ind{c}\coloneqq(c_1,\dots,c_n)^\top$ be the column vector of all annihilation operators. 
Then $\ind{c}^\dagger$ is the row vector of all creation operators.
This allows us to rewrite the Hamiltonian as
\begin{align}
  H = \ind{c}^\dagger h \ind{c}.
\end{align} 
Since $h$ is Hermitian, there exists a unitary $u\in \U(\CC^n)$ that diagonalizes $h$, i.e.\ $h = u^\dagger \varepsilon u$, where $\varepsilon = \diag(\varepsilon_1,\dots,\varepsilon_n)$, and w.l.o.g. we order the eigenvalues such that $\varepsilon_1\leq \varepsilon_2\dots\leq\varepsilon_n$.
This means that the Hamiltonian can be rewritten as
\begin{align}
  H = \ind{c}^\dagger h \ind{c} = (u\ind{c})^\dagger \varepsilon (u \ind{c}) \equiv \ind{d}^\dagger \varepsilon \ind{d} = \sum_{k\in[n]} \varepsilon_k d_k^\dagger d_k, 
\end{align}
where $\ind d$ is the list of annihilation operators corresponding to the normal modes of $H$.
In particular, $d_j$ and $d_k^\dagger$ satisfy the canonical anti-commutation relations.

\subsubsection{Fermionic Gaussian states}
A state $\rho$ is said to be a \emph{fermionic Gaussian state} if there exists a quadratic Hamiltonian $H$ such that
\begin{align}
  \rho = \frac{\e^{-H}}{Z},
\end{align}
where $Z = \Tr(\e^{-H})$.
Such states are completely characterized by the \emph{correlation matrix} $\Gamma^{(c)}$ defined by $\Gamma^{(c)}_{ij}\coloneqq\langle c_j^\dagger c_i \rangle$.
To see this, consider
\begin{align}
  \rho = \frac{\e^{-H}}{Z} = \frac{\e^{-\sum_k \varepsilon_k d_k^\dagger d_k}}{Z} = \frac{1}{Z}\bigotimes_{k\in[n]} \e^{- \varepsilon_k d_k^\dagger d_k} = \bigotimes_{k\in[n]} \frac{\e^{- \varepsilon_k d_k^\dagger d_k}}{Z_k},
\end{align}
where $Z_k  \coloneqq \Tr(\e^{-\varepsilon_k d_k^\dagger d_k})$.
Therefore the state $\rho$ can be fully specified by the correlation matrix $\Gamma^{(d)}\in\CC^{n\times n}$ which is diagonal in the $d$-basis, and has the occupation numbers $\{\langle d_k^\dagger d_k \rangle\}_k$ on the diagonal.
Transforming back to the $c$-basis
\begin{align}
  \Gamma^{(c)} = u^\dagger \Gamma^{(d)} u,
\end{align}
we have the claimed result.

Assuming there is a gap at the Fermi level, i.e.\ $\varepsilon_N < \varepsilon_{N+1}$, the ground state of the Hamiltonian $H = \ind{c}^\dagger h \ind{c} = \ind{d}^\dagger \varepsilon \ind{d}$ with $N$ particles is the state where the lowest energy $N$ levels in the diagonal basis are filled,
\begin{align}
  \Gamma_0^{(d)} =
  \begin{pmatrix}
    \1_{N\times N} & 0\\
    0 & 0
  \end{pmatrix} = 
  \sum_{k\in[N]} \ketbra{k}{k}.
\end{align}
The correlation matrix in the original $c$-basis is obtained by the transformation
\begin{align}
  \label{eq:ground_slater}
  \Gamma_0^{(c)} = u^\dagger \Gamma_0^{(d)} u  = 
  \sum_{k\in[N]} u^\dagger\ketbra{k}{k}u 
  \eqqcolon\sum_{k\in[N]} \ketbra{\phi_k}{\phi_k} \eqqcolon \Phi_0 \Phi_0^\dagger, 
\end{align}
where $\Phi_0 \in \CC^{n\times N}$ is the \emph{Slater matrix}, whose columns correspond to the $N$ lowest-energy eigenvectors $\{\ket{\phi_k}\}_{k\in[N]}$ of $h$.

More generally, the Slater matrix can be used to represent any pure Gaussian state.
The reason is that any pure Gaussian state corresponds to a correlation matrix $\Gamma$ which is a projector, $\Gamma^2 = \Gamma$.
Therefore,
\begin{align}
  \Gamma = \sum_{k\in N_\mathrm{oc}\subset[n]} \ketbra{\phi_k}{\phi_k} \equiv \Phi\Phi^\dagger,
\end{align}
where $N_\mathrm{oc}\subset[n]$ is the set of occupied modes in the diagonal basis.  
The state $\ket\Phi\in\mc{F}_N$ in the $N$ particle Fock space associated to the Slater matrix $\Phi$ is given by 
\begin{align}
  \ket\Phi = d_{j_1}^\dagger d_{j_2}^\dagger \dots d_{j_N}^\dagger \ket{\Omega},
\end{align}
where $\ket{\Omega}$ is the vacuum state.
Using $d_j^\dagger = \sum_k (u^\dagger)_{kj}c_k^\dagger = \sum_k \Phi_{kj}c_k^\dagger$ and rearranging we find
\begin{align}
  \ket{\Phi} = \sum_{S}\det(\Phi_S)\ket{S},
\end{align}
where $S\in\mc{P}([n])$ such that $|S|=N$ and $s_1<s_2<\dots<s_N$ for all $\{s_1,s_2\dots,s_N\}\in S$, and $\Phi_S\in\CC^{N\times N}$ such that $(\Phi_S)_{jk}\coloneqq \Phi_{s_jk}$.
The \emph{state fidelity} can then be expressed in terms of the Slater matrices as
\begin{align}
  F(\ket{\Psi},\ket{\Phi}) \coloneqq |\braket{\Psi}{\Phi}|^2 = |\det(\Psi^\dagger \Phi)|^2.
\end{align}

\subsubsection{Time evolution of fermionic Gaussian states}
Unitary time evolution of a Gaussian state under a quadratic Hamiltonian is captured by the time evolution of the corresponding correlation matrix.
Given a correlation matrix $\Gamma(0)$ at time $t=0$, the time evolved matrix $\Gamma(t)$ 
is given by
\begin{align}
  \Gamma(t) = \e^{-\i ht}\Gamma(0)\e^{\i ht}.
\end{align}
For pure states $\Gamma(t) = \Phi(t)\Phi^\dagger(t)$, where
\begin{align}
  \label{eq:slater_time_evolution}
  \Phi(t) = \e^{-\i ht}\Phi(0),
\end{align}  
so the dynamics is fully captured by the evolution of the Slater matrix.
\subsubsection{Calculating expectation values of observables}
We are interested in evaluating the expectation values of observables using the correlation matrix.
The expectation value for the energy is 
\begin{align}
  \langle H\rangle = \sum_{i,j\in[n]} h_{ij} \langle c_i^\dagger c_j\rangle = \Tr(h\Gamma).
\end{align}
The \emph{bond current} $J_{kj}$ corresponding to the flow of particles from mode $k$ to mode $j$ is defined using the continuity equation
\begin{align}
  \frac{d n_k}{dt} \eqqcolon -\sum_{j\in[n]} J_{kj}.
\end{align} 
A straightforward calculation using the Heisenberg equation of motion shows that 
\begin{align}
  \frac{d n_k}{dt} = \i\, [H,n_k] = \sum_{j\in[n]} \i\left(h_{jk}c_j^\dagger c_k - h_{kj}c_k^\dagger c_j \right) \implies J_{kj} = \i\left(h_{kj}c_k^\dagger c_j - h_{jk}c_j^\dagger c_k \right).
\end{align}
So, the expectation value of the bond current from $k$ to $j$ is given by
\begin{align}
  \langle J_{kj} \rangle = \i \left(h_{kj} \Gamma_{jk} - h_{jk} \Gamma_{kj}\right) = 2\Im(h_{jk}\Gamma_{kj}).
\end{align}

\subsection{Translating FHE into the language of quadratic Hamiltonians}

A pulse $V_\ind{b}$ is generated by a quadratic Hamiltonian 
\begin{align}
  G_\ind{b} = \frac{2\pi}{3}\sum_{j\in[n]}b_j c_j^\dagger c_j,
\end{align}
so the coefficient matrix $g_\ind{b}$ for the generator $G_\ind{b}$ is given by
\begin{align}
  g_\ind{b} = \frac{2\pi}{3} \diag(b_1,\dots,b_n),
\end{align}
where $b_j\in\TT\equiv\{-1,0,1\}$ as before.
These matrices generate unitaries that act on the smaller $n$-dimensional space
\begin{align}
  v_{\ind{b}}
  = \e^{-\i g_{\ind{b}}}
  = \diag\!\left(
      \e^{-2\pi \i b_1/3},
      \e^{-2\pi \i b_2/3},
      \dots,
      \e^{-2\pi \i b_n/3}
    \right).
\end{align}
In the previous section, 
\cref{eq:slater_time_evolution} describes the time evolution of a pure Gaussian state.
Using the second order Trotter formula \eqref{eq:second_order_trotter}, the engineered time evolution operator in the quadratic picture corresponding to the target Hamiltonian is
\begin{align}
  \label{eq:second_order_trotter_quadratic}
  u_\mathrm{eng}(t) =\left(\prod_{\mu=1}^{\stackrel{2r}{\leftarrow}}v_{\ind{b}_\mu}^\dagger \e^{-\i h_S\lambda_{\ind{b}_\mu} t/2n_T}v_{\ind{b}_\mu}\prod_{\mu=1}^{\stackrel{2r}{\rightarrow}}v_{\ind{b}_\mu}^\dagger \e^{-\i h_S\lambda_{\ind{b}_\mu} t/2n_T}v_{\ind{b}_\mu}\right)^{n_T}\approx e^{-ih_Tt}.
\end{align}
We can therefore use \cref{eq:second_order_trotter_quadratic} to act on the initial state described by the Slater matrix $\Phi_0$ defined in \cref{eq:ground_slater} corresponding to the ground state of the system Hamiltonian to simulate the quench dynamics of interest, as described in \cref{eq:slater_time_evolution}.

\section{Artificial gauge fields on a triangular lattice}
\label[appendix]{sec:triangular_appendix}

We consider a quadratic system Hamiltonian on a triangular lattice
\begin{align}
\begin{split}
  H_S = -J\sum_{\langle x,y \rangle} \Big(c_{x+1,y}^\dagger c_{x,y} + c_{x,y+1}^\dagger c_{x,y}\\
   + c_{x-1,y+1}^\dagger c_{x,y} + \mathrm{H.c.} \Big),
\end{split}
\end{align}
and the set of pulses
\begin{align}
  V_1 &= \prod_{x,y} \exp(-\i\theta_1 x n_{x,y}),\\
  V_2 &= \prod_{x,y} \exp(-\i\theta_2 y n_{x,y}),\\
  V_3 &= \prod_{x,y} \exp(-\i\theta_3 (x+y) n_{x,y}).
\end{align}
Setting $\lambda_1 = \lambda_2 = \lambda_3 = 1$, the effective Hamiltonian
\begin{align}
  H_{\mathrm{eff}} &= \sum_{j=1}^3 \lambda_j V_j^\dagger H_S V_j
\end{align}
becomes
\begin{align}
\begin{split}
    H_\mathrm{eff} = -\sum_{\langle x,y \rangle} \Big(J_1 c_{x+1,y}^\dagger c_{x,y} + J_2 c_{x,y+1}^\dagger c_{x,y}\\
   + J_3 c_{x-1,y+1}^\dagger c_{x,y} + \mathrm{H.c.} \Big),
\end{split}
\end{align}
where
\begin{align}
  J_1/J &= \e^{\i\theta_1} + 1 + \e^{\i\theta_3}\\
  J_2/J &= 1 + \e^{\i\theta_2} + \e^{\i\theta_3}\\
  J_3/J &= \e^{-\i\theta_1} + \e^{\i\theta_2}+ 1.
\end{align}
By setting $\theta_1 = \pi/2, \theta_2 = \pi$ and $\theta_3 = \pi$ we find that
\begin{align}
  J_1/J\equiv \e^{\i\phi_1} &= \i = \e^{\i\pi/2}\\
  J_2/J\equiv \e^{\i\phi_2} &= -1  = \e^{\i\pi}\\
  J_3/J\equiv \e^{\i\phi_3} &= -\i = \e^{-\i\pi/2}.
\end{align}
So the total flux per plaquette is
\begin{align}
  \Phi = \phi_1 - \phi_2 + \phi_3 = -\pi \equiv \pi\ (\bmod\ 2\pi).
\end{align}

\section{Proof of Lemma \ref{lem:main}}
\label[appendix]{ap:lemma_proof}

In this section, we prove \cref{lem:main}, which provides the backbone of our framework.
We start by introducing some notation and showing some preliminary results. 
For the following definitions and lemmas, let $A\in\RR^{m\times n}$, $\vec{x}\in\RR^n$ and $\vec{y}\in\RR^m$.
The Moore-Penrose pseudo-inverse \cite{rohatgiGeneralizedInversesTheory1976} of the matrix $A$ is denoted by $A^+\in\RR^{n\times m}$.

\begin{definition}
  A linear system $A\vec{x}=\vec{y}$ is said to be consistent if it has at least one solution. 
\end{definition}

\begin{lemma}[\cite{rohatgiGeneralizedInversesTheory1976}]
	\label{lem:lin_system_is_consistent}
  The linear system $A\vec{x}=\vec{y}$ is consistent if and only if $AA^+\vec{y}=\vec{y}$. 
	Further, if $\rank{(A)}=m$ then $AA^+=\1_m$.
  Therefore, if $\rank{(A)} = m$ then the linear system is consistent for all $\vec{y}$.
\end{lemma} 

\begin{lemma}
    \label{lem:primitive_root}
	Let $n \geq 1$ and $k \geq 2$ be integers.
	For any $\ind a \in \ZZ_k^n$ and any primitive $k$-th root of unity $\omega_k$, we have
	\begin{equation}
		\sum_{\ind{b}\in{\ZZ_k^n}} \omega_k^{\ind{a}\cdot\ind{b}} = k^n \delta_{\ind{a},\boldsymbol{0}} .
	\end{equation}
\end{lemma}
\begin{proof}
	First, we consider $n=1$.
	Let $a = 0 \in \ZZ_k$, then
	\begin{equation}
		\sum_{b \in \ZZ_k} \omega_k^{ab}
		= \sum_{b \in \ZZ_k} \omega_k^0
		= \sum_{b \in \ZZ_k} 1
		= |\ZZ_k|
		= k .
	\end{equation}
	It is easy to check that $\omega_k^a$ is another $k$-th root of unity, and assuming $a \neq 0$, we have $\omega_k^a \neq 1$ by the assumption that $\omega_k$ is primitive.
	Then, by the summation formula for finite geometric series,
	\begin{equation}
		\sum_{b \in \ZZ_k} \omega_k^{ab}
		= \sum_{b=0}^{k-1} (\omega_k^a)^b
		= \frac{1 - (\omega_k^a)^k}{1 - \omega_k^a}
		= \frac{1 - 1}{1 - \omega_k^a}
		= 0 ,
	\end{equation}
	employing the very definition of a $k$-th root.
	
	The generalization to arbitrary $n$ is straightforward,
	\begin{equation}
		\sum_{\ind{b}\in{\ZZ_k^n}} \omega_k^{\ind{a}\cdot\ind{b}}
		= \sum_{\ind{b}\in{\ZZ_k^n}} \prod_{j \in [n]} \omega_k^{a_j b_j}
		= \prod_{j \in [n]} \sum_{b_j \in \ZZ_k} \omega_k^{a_j b_j}
		= \prod_{j \in [n]} (k \delta_{a_j,0})
		= k^n \delta_{\ind{a},\boldsymbol{0}} .
	\end{equation}
\end{proof}

Let $F_k \in \CC^{d\times c}$ be the submatrix of the Fourier matrix defined by $(F_k)_{\ind{a}\ind{b}} = \omega_k^{\ind{a}\cdot\ind{b}}$ where $\omega_k = e^{2\pi \i/k}$, and $\ind{a}\in\mathcal{T}_n\subset\TT^n$ and $\ind{b}\in \ZZ_k^n$. 
This means that we have $d = |\mathcal{T}_n| = (3^n-1)/2$ and $c = k^n$ .

\begin{lemma}
  \label{lem:main_appendix}
  The conical hull of the set of columns of $F_k$ is the complex vector space $\CC^d$, i.e. $\cone(\col(F_k)) = \CC^d$, for any integer $k\geq 3$. 
  Hence, as a map $F_k:\RR_{\geq 0}^c \to \CC^d$, $F_k$ is surjective.
\end{lemma}
\begin{proof}
  Let $\boldsymbol{\lambda}\in \RR_{\geq 0}^{c}$ be an entrywise nonnegative vector.
  We prove the lemma by showing that the linear system of equations
  \begin{align}
    F_k \boldsymbol{\lambda} = \boldsymbol{\gamma}
  \end{align}
  is consistent for all $\boldsymbol{\gamma}\in{\CC^d}$.
  We show this by successively shrinking the domain of the variable $\boldsymbol{\lambda}$ from complex, to real, and finally to nonnegative, while preserving consistency.
  
  Let $\boldsymbol{z}\in \CC^{c}$, then the complex equation system 
  \begin{align}
    \label{eq:complex_system}
    F_k\boldsymbol{z} = \boldsymbol{\gamma}
  \end{align}
  is clearly consistent for all $\boldsymbol{\gamma}\in\CC^d$ since $\rank(F_k)=d$.
  To consider the real case, we first use the decomposition $F_k = \Re(F_k)+i\Im(F_k)$ to rewrite \cref{eq:complex_system} as a real equation system,
  \begin{align}
    \label{eq:real_system}
    \begin{pmatrix}
      \Re(F_k) & -\Im(F_k) \\
      \Im(F_k) & \Re(F_k) 
    \end{pmatrix}
    \begin{pmatrix}
      \Re(\ind{z})\\
      \Im(\ind{z})
    \end{pmatrix}
    =
    \begin{pmatrix}
      \Re(\ind{\gamma})\\
      \Im(\ind{\gamma})
    \end{pmatrix}.
  \end{align}
  \Cref{eq:complex_system} has purely real solutions, if and only if \cref{eq:real_system} has solutions with $\Im(\ind{z})=0$.
  Thus, we wish to find solutions with fewer degrees of freedom on the left-hand side, but the same number of constraints and right-hand side;
  or equivalently, solutions to the smaller system
  \begin{align}
    \label{eq:reduced_system}
    \begin{pmatrix}
      \Re(F_k) \\
      \Im(F_k)  
    \end{pmatrix}
      \ind{x}
      \eqqcolon \tilde{F}_k\ind{x} = \tilde{\ind{\gamma}}
      \coloneqq
    \begin{pmatrix}
      \Re(\ind{\gamma})\\
      \Im(\ind{\gamma})
    \end{pmatrix},
  \end{align}  
  where $\ind{x}\in\RR^{c},\tilde{\ind{\gamma}}\in\RR^{2d}$ and $\tilde{F}_k\in \RR^{2d\times c}$.
  We show that the rows of $\tilde{F}_k$ are mutually orthogonal, i.e. that
\begin{align}
    \tilde{F}_k \tilde{F}_k^\top =
    \begin{pmatrix}
      \Re(F_k)\Re(F_k)^\top & \Re(F_k) \Im(F_k)^\top\\
      \Im(F_k) \Re(F_k)^\top & \Im(F_k)\Im(F_k)^\top
    \end{pmatrix}
\end{align}
is proportional to the identity.
  The real and imaginary parts of the matrix $F_k$ have entries
  \begin{align}
    (\Re(F_k))_{\boldsymbol{a}\boldsymbol{b}} &= \frac{1}{2}(\omega_k^{\boldsymbol{a}\cdot \boldsymbol{b}} + \omega_k^{-\boldsymbol{a}\cdot \boldsymbol{b}}),
    \\
		(\Im(F_k))_{\boldsymbol{a}\boldsymbol{b}} &= \frac{1}{2\i}(\omega_k^{\boldsymbol{a}\cdot \boldsymbol{b}} - \omega_k^{-\boldsymbol{a}\cdot \boldsymbol{b}}). 
  \end{align}
  Therefore,
  \begin{align}
  \label{eq:fr_2}
  \begin{split}
    \left(\Re(F_k)\Re(F_k)^\top\right)_{\ind{a}\ind{b}} &= \sum_{\ind{c}\in\ZZ_k^n} \left(\Re(F_k)\right)_{\ind{a}\ind{c}} \left(\Re(F_k)\right)_{\ind{b}\ind{c}}\\
    &= \frac{1}{4} \sum_{\ind{c}\in\ZZ_k^n}
    \left(
    \omega_k^{\boldsymbol{a}\cdot\boldsymbol{c}}\omega_k^{\boldsymbol{b}\cdot\boldsymbol{c}}
    + \omega_k^{\boldsymbol{a}\cdot\boldsymbol{c}}\omega_k^{-\boldsymbol{b}\cdot\boldsymbol{c}}
    + \omega_k^{-\boldsymbol{a}\cdot\boldsymbol{c}}\omega_k^{\boldsymbol{b}\cdot\boldsymbol{c}}
    + \omega_k^{-\boldsymbol{a}\cdot\boldsymbol{c}}\omega_k^{-\boldsymbol{b}\cdot\boldsymbol{c}}\right)\\
    &= \frac{1}{4} \sum_{\ind{c}\in\ZZ_k^n}
    \left(
    \omega_k^{(\boldsymbol{a}+\boldsymbol{b})\cdot\ind{c}} +
    \omega_k^{(\boldsymbol{a}-\boldsymbol{b})\cdot\ind{c}} +
    \omega_k^{(-\boldsymbol{a}+\boldsymbol{b})\cdot\ind{c}} +
    \omega_k^{(-\boldsymbol{a}-\boldsymbol{b})\cdot\ind{c}}\right)\\
    &= \frac{k^n}{2} \left( \delta_{\boldsymbol{a},\boldsymbol{b}} + \delta_{\boldsymbol{a},-\boldsymbol{b}}\right) = \frac{k^n}{2}\delta_{\boldsymbol{a},\boldsymbol{b}},
  \end{split}
\end{align}
where we made use of \cref{lem:primitive_root} for the first equality in the final line.
Recall the split of \cref{eq:trit_split} in the form $\ind{a}\in\mathcal{T}_n \iff -\ind{a}\notin\mathcal{T}_n$ (for $\ind{a} \neq \boldsymbol{0}$), then the row index restriction on $F_k$ gives us
\begin{equation}
	\ind{a}\in\mathcal{T}_n \wedge \ind{b}\in\mathcal{T}_n
	\quad\iff\quad
	\ind{a}\in\mathcal{T}_n \wedge -\ind{b}\notin\mathcal{T}_n
	\quad\implies\quad
	\ind{a} \neq -\ind{b}
	\quad\iff\quad
	\delta_{\boldsymbol{a},-\boldsymbol{b}} = 0,
\end{equation}
and thus the last equality above.
Analogously,
\begin{align}
  \label{eq:fi_2}
  \begin{split}
    \left(\Im(F_k)\Im(F_k)^\top\right)_{\ind{a}\ind{b}}
    = \frac{k^n}{2} \left( \delta_{\boldsymbol{a},\boldsymbol{b}} - \delta_{\boldsymbol{a},-\boldsymbol{b}}\right) = \frac{k^n}{2}\delta_{\boldsymbol{a},\boldsymbol{b}},
  \end{split}
\end{align}
and
\begin{align}
  \left(\Re(F_k) \Im(F_k)^\top \right)_{\ind{a}\ind{b}}= \left(\Im(F_k) \Re(F_k)^\top \right)_{\ind{a}\ind{b}} = 0.
\end{align}
This means that all the rows of $\tilde{F}_k$ are mutually orthogonal and $\rank(\tilde{F}_k)=2d$, therefore by \cref{lem:lin_system_is_consistent} the equation system \eqref{eq:reduced_system} is consistent for all $\tilde{\ind{\gamma}}\in\RR^{2d}$.
We also conclude that the pseudo-inverse is $\tilde{F}_k^+ = \frac{2}{k^n} \tilde{F}_k^\top$, and $\tilde{F}_k^+ \tilde{\ind\gamma}$ is a particular solution of \cref{eq:reduced_system}, and hence a particular \emph{real} solution of \cref{eq:complex_system}.

To see that the equation system remains consistent when we restrict the solutions to be nonnegative, notice that the all-ones-vector $\ind{1}\in\RR^c$ is in the kernel of $\tilde{F}_k$,
\begin{align}
  \tilde{F}_k\boldsymbol{1} = \begin{pmatrix} \Re(F_k) \boldsymbol{1}\\ \Im(F_k)\boldsymbol{1}\end{pmatrix} = \ind{0}.
\end{align}
This is because \cref{lem:primitive_root} gives
\begin{align}
  \begin{split}
  \sum_{\ind{b}\in{\ZZ_k^n}} \omega_k^{\ind{a}\cdot\ind{b}} = \sum_{\ind{b}\in{\ZZ_k^n}}\left( \Re\left(\omega_k^{\ind{a}\cdot\ind{b}}\right) + \i\Im\left(\omega_k^{\ind{a}\cdot\ind{b}}\right) \right)= k^n \delta_{\ind{a},\ind{0}} \\
  \implies \sum_{\ind{b}\in{\ZZ_k^n}} \Re\left(\omega_k^{\ind{a}\cdot\ind{b}}\right) = 0 \quad \text{and} \quad \sum_{\ind{b}\in{\ZZ_k^n}}\Im\left(\omega_k^{\ind{a}\cdot\ind{b}}\right) = 0 \quad \forall\ind{a}\neq\ind{0},
  \end{split}
\end{align}
and therefore
\begin{align}
  (\Re(F_k)\boldsymbol{1})_\ind{a} = \sum_{\ind{b}\in\ZZ_k^n} \Re\left(\omega_k^{\ind{a}\cdot\ind{b}}\right) = 0 \quad \text{and}\quad  (\Im(F)\boldsymbol{1})_\ind{a} = \sum_{\ind{b}\in\ZZ_k^n} \Im\left(\omega_k^{\ind{a}\cdot\ind{b}}\right)=0, \quad \forall \ind{a}\in\mathcal{T}_n.
\end{align}
It is a standard result in linear algebra that the solution space of a linear equation system is an affine shift of the kernel, so moving from a real solution along the positive direction of $\boldsymbol{1}$ gives more real solutions.
All entries increase along this ray, so it eventually has to become nonnegative.
Concretely, for any given $\ind{\gamma} \in \CC^d$,
\begin{align}
  \ind{\lambda} = \tilde{F}_k^+ \tilde{\ind\gamma} + \abs{\min(\tilde{F}_k^+ \tilde{\ind\gamma})}\boldsymbol{1}
\end{align}
is a nonnegative solution to \cref{eq:reduced_system}, where the minimum is taken over all entries of the vector.
\end{proof}

\Cref{lem:main} in the main text is a special case of \cref{lem:main_appendix} with $k=3$.

\section{Efficient relaxation}
\label[appendix]{sec:relaxation}
In general, \acp{LP} can be solved in polynomial time. 
However, 
the \ac{LP}
\begin{align}
\begin{split}
\label{eq:LP_appendix}
\text{minimize} \quad & \boldsymbol{1}^\top \boldsymbol{\lambda} \\
\text{subject to} \quad &  F' \boldsymbol{\lambda} = \boldsymbol{\gamma}', \\
                        &  \boldsymbol{\lambda} \in \RR_{\geq 0}^{c},
\end{split}
\end{align} 
where $F' \coloneqq F\big|_{\supp(\alpha)}\in\CC^{r\times c}$ grows exponentially with the system size, making it computationally expensive for larger systems. 
In this section, we argue that for many Hamiltonian engineering instances the linear program can be relaxed to polynomial size while still offering a reasonable but suboptimal solution.
This method was first used for the Hamiltonian engineering task for spin systems in \cite{basslerGeneralEfficientRobust2025a}. 
It is important to emphasize that the pulses found by the relaxed \ac{LP} still allow us to decompose the target Hamiltonian \emph{exactly}, it is only the quantum runtime that is suboptimal.

\begin{figure*}[t]
\centering

\begin{minipage}{0.32\textwidth}
\centering
\includegraphics[width=\linewidth]{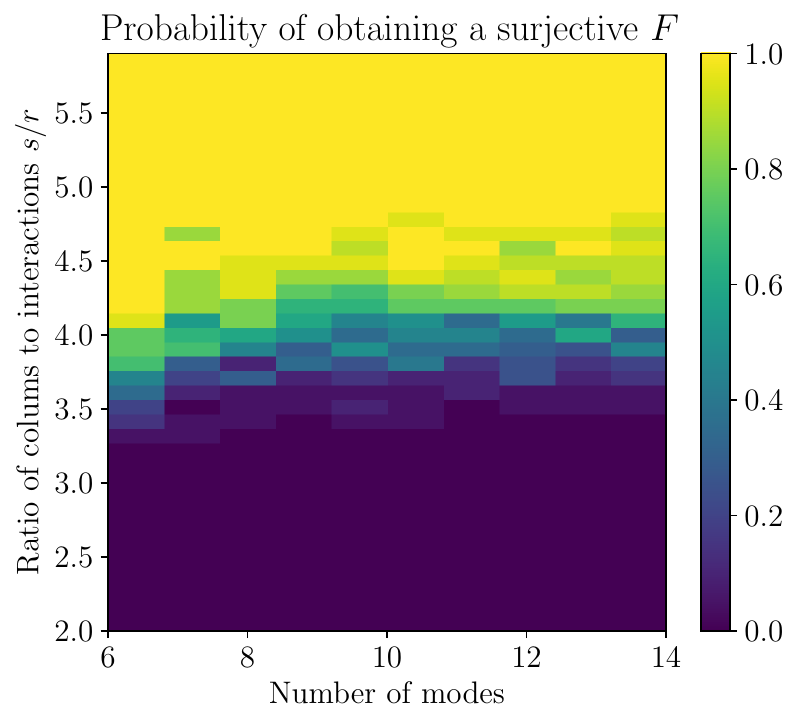}

(a)
\end{minipage}\hfill
\begin{minipage}{0.32\textwidth}
\centering
\includegraphics[width=\linewidth]{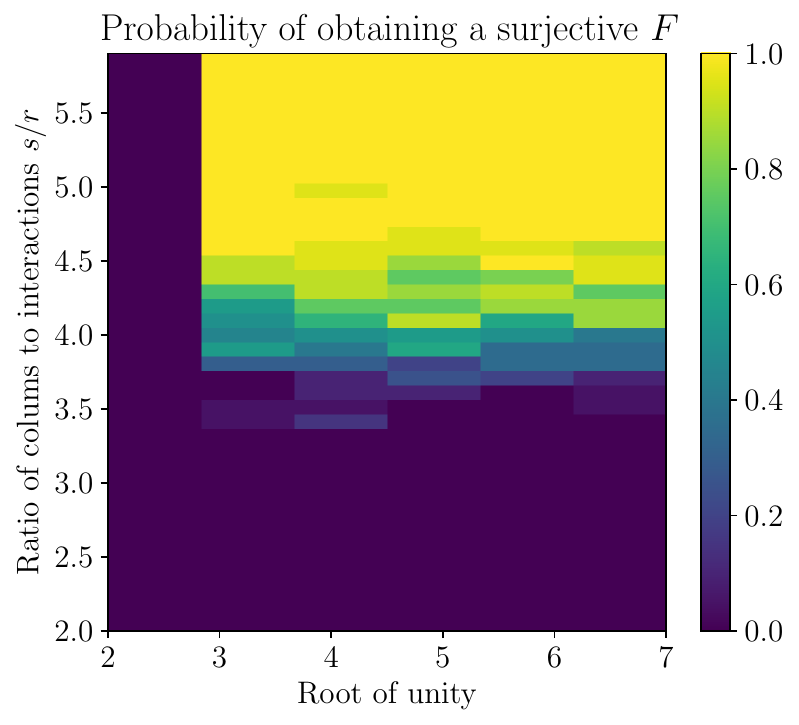}

(b)
\end{minipage}\hfill
\begin{minipage}{0.32\textwidth}
\centering
\includegraphics[width=\linewidth]{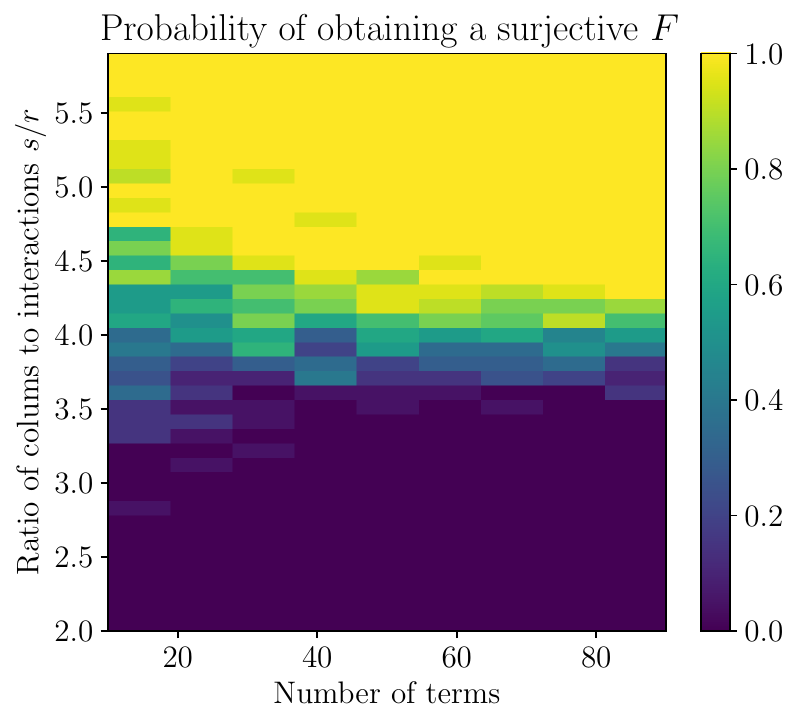}

(c)
\end{minipage}

\caption{Phase transition in the probability of obtaining a feasible matrix as a function of the ratio between the number of sampled pulses (columns) and the number of Hamiltonian terms (rows). The transition occurs when the number of pulses is roughly four times the number of terms. 
(a) Varying the number of modes while keeping the number of interactions fixed at $40$ and the root of unity fixed at $3$. 
(b) Varying the root of unity while keeping the number of interactions fixed at $40$ and the number of modes fixed at $8$. As expected, conical combinations of second roots of unity are insufficient to span the full space. 
(c) Varying the number of terms which appear in the Hamiltonian while keeping the root of unity fixed at $3$ and the number of modes fixed at $8$. The phase transition becomes sharper as the number of terms increases.
}
\label{fig:phase_transition}

\end{figure*}

The idea is to construct a new constraint matrix $\tilde{F}\in \CC^{r\times s}$ by randomly sampling $s$ columns ---corresponding to $s$ random local pulses \eqref{eq:defn_pulse}--- of the full matrix $F$, until the columns of $\tilde{F}$ conically span $\CC^r$, $\cone(\col(\tilde{F})) = \CC^r$, i.e. until $\tilde{F}$ is surjective.
This can be efficiently checked thanks to the following proposition.

\begin{proposition}
  \label{prop:origin_interior}
  Let $\col(\tilde{F})$ be the set of columns of $\tilde{F}$ interpreted as vectors in $\CC^r$.
  Then
  \begin{align}
    \label{eq:origin_interior}
  \boldsymbol{0}\in\operatorname{int}(\conv(\col(\tilde{F}))) \implies \cone(\col(\tilde{F})) = \CC^r .
  \end{align}
\end{proposition}
\begin{proof}
  Let $B_\varepsilon\coloneqq \{\ind{x}\in\CC^r\big|\ |\ind{x}|<\varepsilon\}$ denote the open ball of radius $\varepsilon$ and let $\col(\tilde{F})=\{\ind{v}_j\}_{j\in[s]}$.
  $\boldsymbol{0}~\in~\operatorname{int}(\conv(\col{\tilde{F}}))~\implies~\exists\varepsilon>0$ such that $B_\varepsilon \subset \conv(\col(\tilde{F}))$.
   Therefore every point in this open ball can be written as a convex combination of the columns of $\tilde{F}$, i.e.
   \begin{align}
    \forall \ind{x}\in B_\varepsilon,\ \exists\lambda_j\geq 0 \text{ such that } \ind{x} = \sum_{j=1}^s \lambda_j \ind{v}_j \text{ and }\sum_{j=1}^s\lambda_j = 1.
   \end{align}
   The open ball $B_\varepsilon$ is homeomorphic to the full space $B_\varepsilon\cong \CC^r$, where the homeomorphism is given by $f:B_\varepsilon\to\CC^r$, defined by
   \begin{align}
    f(\ind{x}) \coloneqq \frac{\ind{x}}{\varepsilon-|\ind{x}|}.
   \end{align}
  This means that for any point $\ind{z}\in\CC^r$ there exists $\ind{x}\in B_\varepsilon$ such that $\ind{z}= f(\ind{x})$. 
  Therefore $\ind{z}\in\cone(\col(\tilde{F}))$ since
  \begin{align}
    \ind{z} = f(\ind{x}) = \frac{\ind{x}}{\varepsilon-|\ind{x}|} = \sum_{j=1}^s \frac{\lambda_j}{\varepsilon-|\ind{x}|} \ind{v_j} \quad \text{and} \quad \frac{\lambda_j}{\varepsilon-|\ind{x}|}\geq 0.
  \end{align}
\end{proof}
\Cref{eq:origin_interior} can be checked efficiently, since if the \ac{LP}
\begin{align}
\begin{split}
\label{eq:check_LP}
\text{minimize} \quad & \boldsymbol{1}^\top \boldsymbol{\lambda} \\
\text{subject to} \quad &  \tilde{F} \boldsymbol{\lambda} = \boldsymbol{0}, \\
                        &  \boldsymbol{\lambda} \geq \boldsymbol{1} \quad \text{entrywise},
\end{split}
\end{align}
has a feasible solution, and $\tilde{F}$ has full rank, then $\tilde{F}$ is surjective.

We numerically observe that the number of columns we need to sample to conically span $\CC^r$ is given by four times the number of rows $s = 4r$, as can be seen from \cref{fig:phase_transition}.
Similar to the qubit case \cite{basslerGeneralEfficientRobust2025a}, this parallels a result by Wendel \cite{wendelProblemGeometricProbability1962}, where it is shown that when sampling $s$ vectors from a spherically symmetric distribution in $\RR^{r}$, the probability that all $s$ vectors lie in one half space is given by 
\begin{align}
  \PP_\mathrm{Wendel} [s,r] = \frac{1}{2^{s-1}} \sum_{j=0}^{r-1} \binom{s-1}{j},
\end{align}
which has a sharp phase transition from one to zero at $s=2r$. Our observations fit
\begin{align}
  \PP[\tilde{F} \mathrm{\ is\ surjective}] \approx 1-\PP_\mathrm{Wendel}[s,2r],
\end{align}
so there is a sharp transition in the probability that the conical span of the columns of $\tilde{F}$ is the full space (the negation of Wendel's statement) at $s=4r$.
The extra factor of $2$ stems from the fact that in our case the vector space to be spanned is complex, so the equivalent real dimension would be $2r$.

\begin{figure*}[t]
\centering
\subfloat[]{
\includegraphics[width=0.48\textwidth]{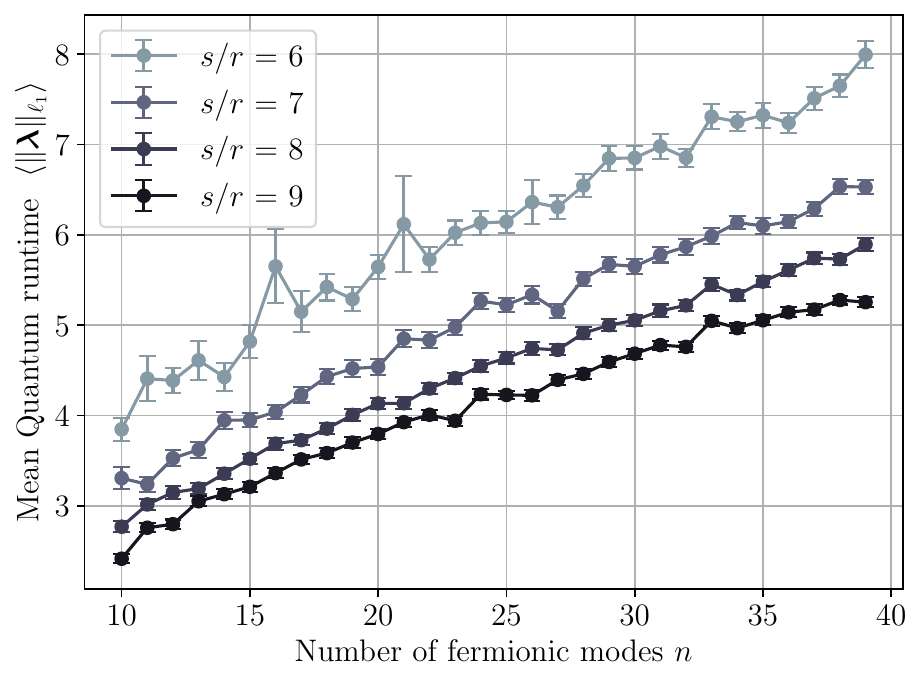}
}
\hfill
\subfloat[]{
\includegraphics[width=0.48\textwidth]{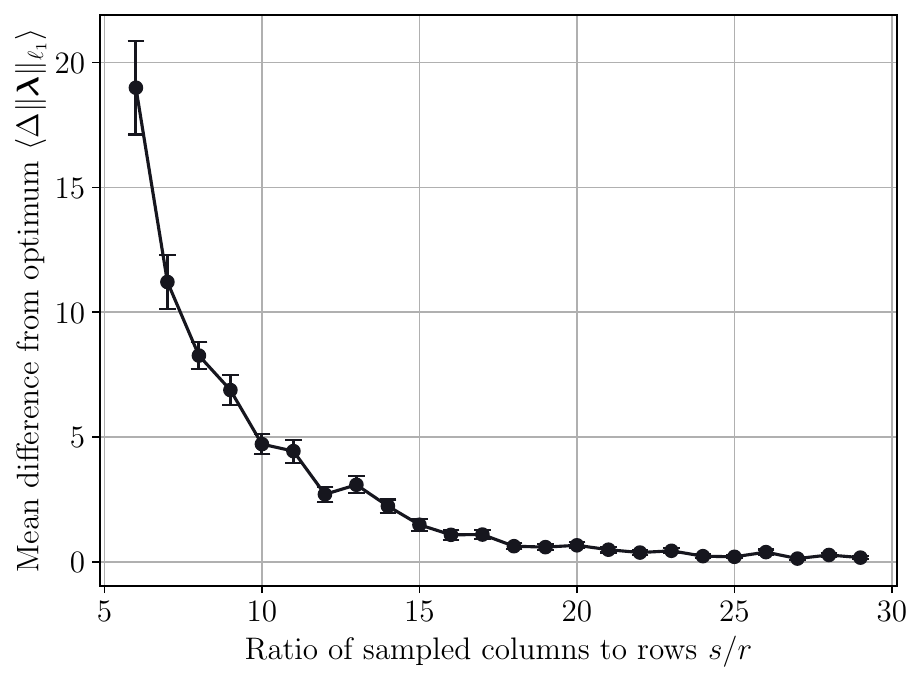}
}
\caption{(a) Average quantum runtime for the simulation of a target Hamiltonian starting from a quadratic system Hamiltonian on a 1D chain.
The averages are calculated over 100 randomly picked Hamiltonians on a chain of length $n$, with $|\beta_\ind{a}|\leq 1$, and with the same support as the system Hamiltonian.
We observe that the average runtime decreases as we sample more columns from the constraint matrix, or as the column to row ratio of the reduced constraint matrix $\tilde{F}$ increases. The error bars indicate standard error of the mean.
(b) The mean difference of the total quantum runtime from the optimum for the efficient relaxation for randomly picked target Hamiltonians on a chain of length 8. For each $s/r$ ratio, 100 random Hamiltonians are generated and the pulse sequence is constructed from the LP \eqref{eq:LP_appendix} using both the full constraint matrix $F'$ and the relaxed matrix $\tilde{F}$. As we increase the ratio of columns to rows, we approach the optimum, quantified by the difference in quantum runtime $\Delta\|\boldsymbol{\lambda}\|_{\ell_1} \coloneqq \big| \|\boldsymbol{\lambda}^*\|_{\ell_1} - \|\boldsymbol{\lambda}\|_{\ell_1}\big|$. The error bars indicate standard error of the mean. 
}
\label{fig:trade-off}
\end{figure*}

As can be seen from \cref{fig:trade-off}, by increasing the number of columns we sample, we increase the likelihood of finding a solution close to the optimum, thereby reducing quantum runtime.
However, since the size of the linear program increases as we sample more columns, the classical runtime for the linear program increases, giving us a trade-off between classical and quantum runtime.

\section{Derivation of the finite-pulse-time error matrix}
\label[appendix]{ap:fpt_error_matrix}

We start with a brief introduction to the interaction picture.
In the case when the Hamiltonian can be expressed as a sum of two Hamiltonians
\begin{align}
  H(t) = H_A(t) + H_B(t),
\end{align}
where the \emph{propagator} (or the time evolution operator) $U_A(t)$ generated by $H_A$ is assumed to be known, it is often useful to transform the Hamiltonian into what is called the \emph{interaction frame} or the \emph{interaction picture}.
The propagator $U(t)$ generated by $H(t)$ is then written as
\begin{align}
  U(t) = U_A(t)U_I(t),
\end{align}
where $U_I(t)$ is the \emph{interaction frame propagator}.
It is obtained by solving
\begin{align}
  \frac{d}{dt} U_I(t) = -\i H_I(t) U_I(t)
\end{align}
where the \emph{interaction frame Hamiltonian} $H_I(t)$ is given by
\begin{align}
  H_I(t) \coloneqq U_A^\dagger(t)H_B(t)U_A(t).
\end{align}

For an always-on Hamiltonian, the $\ind{b}$-th step of the Hamiltonian engineering protocol corresponds to the unitary
\begin{align}
    S_\ind{b}(\lambda_\ind{b}t) = \e^{-\i(H_S-G_\ind{b})t_p } \e^{-\i H_S \lambda_\ind{b}t} \e^{-\i(H_S+G_\ind{b})t_p }.
\end{align}
Let $\tilde{U}_\ind{b}^{(\pm)} = \e^{-\i(H_S\pm G_b)t_p }$ and $\tilde{H}^{(\pm)} = (H_S\pm G_b)$.
Following \cite{basslerGeneralEfficientRobust2025a}, we first use the interaction picture to find the generators of the realistic pulses $\tilde{U}_\ind{b}^{(\pm)}$.
Let us start by considering  $\tilde{U}_\ind{b}^{(+)}$ first.
The interaction frame ansatz is
\begin{align}
  \tilde{U}_\ind{b}^{(+)}(t) = V_\ind{b}(t)U_{I,\ind{b}}^{(+)}(t),
\end{align}
where $U_{I,\ind{b}}^{(+)}(t)$ is the time evolution operator for the interaction frame Hamiltonian
\begin{align}
  H_{I,\ind{b}}^{(+)} = V_\ind{b}^\dagger (t) H_S V_\ind{b}(t).
\end{align}
We use the first order Magnus expansion to find the \emph{average interaction picture Hamiltonian}
\begin{align}
  \label{eq:h_plus_integral}
  \overline{H_{I,\ind{b}}^{(+)}} = \frac{1}{t_p} \int_{0}^{t_p} dt V_\ind{b}^\dagger(t) H_S V_\ind{b}(t),
\end{align}
so that
\begin{align}
  \tilde{U}_\ind{b}^{(+)} = V_\ind{b}\ \e^{-\i \overline{H_I^{(+)}}t_p}.
\end{align}
Repeating a similar analysis for the $\tilde{H}^{(-)}$ Hamiltonian, we find
\begin{align}
  \label{eq:h_minus_integral}
  \overline{H_{I,\ind{b}}^{(-)}} = \frac{1}{t_p} \int_{0}^{t_p} dt V_\ind{b}(t) H_S V_\ind{b}^\dagger(t)
\end{align}
and
\begin{align}
  \tilde{U}_\ind{b}^{(-)} &= V_\ind{b}^\dagger\ \e^{-\i \overline{H_I^{(-)}}t_p}\\
                          &= \e^{-\i \overline{H_{I,\ind{b}}'^{(-)}}t_p}V_\ind{b}^\dagger
\end{align}
where we used $\overline{H_{I,\ind{b}}'^{(-)}}\coloneqq V_\ind{b}^\dagger\overline{H_{I,\ind{b}}^{(-)}}V_\ind{b}$.
Therefore,
\begin{align}
   S_\ind{b}(\lambda_\ind{b}t) = \e^{-\i \overline{H_{I,\ind{b}}'^{(-)}}t_p} \e^{-\i V_\ind{b}^\dagger H_S V_\ind{b} \lambda_\ind{b}t} \e^{-\i \overline{H_{I,\ind{b}}^{(+)}}t_p},
\end{align}
is the time evolution operator for the piecewise time-independent Hamiltonian 
\begin{align}
  K_\ind{b}(t)\coloneqq \begin{cases}
    \overline{H_{I,\ind{b}}^{(+)}} \quad &t\in[0,t_p],\\
    V_\ind{b}^\dagger H_S V_\ind{b} \quad &t\in[t_p,t_p+\lambda_\ind{b}t],\\
    \overline{H_{I,\ind{b}}'^{(-)}} \quad &t\in[t_p+\lambda_\ind{b}t,2t_p+\lambda_\ind{b}t].
  \end{cases}
\end{align}
Up to first order in the Magnus expansion
\begin{align}
  \overline{K_\ind{b}} = \frac{1}{\lambda_\ind{b}t+2t_p}\left[t_p\left(\overline{H_{I,\ind{b}}^{(+)}}+\overline{H_{I,\ind{b}}'^{(-)}}\right)+\lambda_\ind{b}t V_\ind{b}^\dagger H_S V_\ind{b}\right].
\end{align}
We see that the effective Hamiltonian is the desired Hamiltonian $\lambda_\ind{b}t V_\ind{b}^\dagger H_S V_\ind{b}$ plus some \emph{error Hamiltonian}
\begin{align}
  H_{E,\ind{b}} \coloneqq t_p \left(\overline{H_{I,\ind{b}}^{(+)}}+\overline{H_{I,\ind{b}}'^{(-)}}\right).
\end{align}
To find the error Hamiltonian we need to calculate the integrals \cref{eq:h_plus_integral} and \cref{eq:h_minus_integral}.
\begin{align}
  \overline{H_{I,\ind{b}}^{(+)}} &= \frac{1}{t_p}\int_{0}^{t_p} dt V_\ind{b}^\dagger (t) H_S V_\ind{b}(t) \\
 & = \frac{1}{t_p} \int_{0}^{t_p} dt \sum_{\ind{a}\in \supp(\ind{\alpha})}\alpha_\ind{a} V_\ind{b}^\dagger (t) C_\ind{a}V_\ind{b}(t)\\
 & = \frac{1}{t_p} \int_{0}^{t_p} dt \sum_{\ind{a}\in \supp(\ind{\alpha})}\alpha_\ind{a} \omega^{(\ind{a}\cdot\ind{b})t/t_p}C_\ind{a}\\
 \label{eq:int_last_step}
 & = \frac{1}{t_p} \sum_{\ind{a}\in \supp(\ind{\alpha})} \alpha_\ind{a} \left(\int_{0}^{t_p} \omega^{(\ind{a}\cdot\ind{b})t/t_p} dt \right) C_\ind{a}.
\end{align}
Consider the integral
\begin{align}
  I(\varphi) \coloneqq \frac{1}{t_p}\int_{0}^{t_p} \e^{\varphi t/t_p} dt =  \left(\frac{\e^{\varphi}-1}{\varphi}\right),
\end{align}
with this \cref{eq:int_last_step} becomes
\begin{align}
  \overline{H_{I,\ind{b}}^{(+)}}= \sum_{\ind{a}\in \supp(\ind{\alpha})} \frac{3 \alpha_\ind{a}}{2\pi\i (\ind{a}\cdot\ind{b})}\left(\omega^{\ind{a}\cdot \ind{b}}-1\right)C_\ind{a}.
\end{align}
Similarly,
\begin{align}
  \overline{H_{I,\ind{b}}^{(-)}} &= \frac{1}{t_p}\int_{0}^{t_p} dt V_\ind{b} (t) H_S V_\ind{b}^\dagger(t) \\
 & = \frac{1}{t_p} \int_{0}^{t_p} dt \sum_{\ind{a}\in \supp(\ind{\alpha})}\alpha_\ind{a} V_\ind{b} (t) C_\ind{a}V_\ind{b}^\dagger(t)\\
 & = \frac{1}{t_p} \int_{0}^{t_p} dt \sum_{\ind{a}\in \supp(\ind{\alpha})}\alpha_\ind{a} \omega^{-(\ind{a}\cdot\ind{b})t/t_p}C_\ind{a}\\
 \label{eq:int_last_step_minus}
 & = \frac{1}{t_p} \sum_{\ind{a}\in \supp(\ind{\alpha})} \alpha_\ind{a} \left(\int_{0}^{t_p} \omega^{-(\ind{a}\cdot\ind{b})t/t_p} dt \right) C_\ind{a}\\
 & = \sum_{\ind{a}\in \supp(\ind{\alpha})} \frac{3 \alpha_\ind{a}}{2\pi\i (\ind{a}\cdot\ind{b})}\left(1-\omega^{-\ind{a}\cdot \ind{b}}\right)C_\ind{a}.
\end{align}
So 
\begin{align}
  \overline{H_{I,\ind{b}}'^{(-)}} & = V_\ind{b}^\dagger\overline{H_{I,\ind{b}}^{(-)}}V_\ind{b}\\
  & =  \sum_{\ind{a}\in \supp(\ind{\alpha})} \frac{3 \alpha_\ind{a}}{2\pi\i (\ind{a}\cdot\ind{b})}\left(1-\omega^{-\ind{a}\cdot \ind{b}}\right)V_\ind{b}^\dagger C_\ind{a}V_\ind{b}\\
  & = \sum_{\ind{a}\in \supp(\ind{\alpha})} \frac{3 \alpha_\ind{a}}{2\pi\i (\ind{a}\cdot\ind{b})}\left(1-\omega^{-\ind{a}\cdot \ind{b}}\right)\omega^{\ind{a}\cdot\ind{b}} C_\ind{a}\\
  & = \sum_{\ind{a}\in \supp(\ind{\alpha})} \frac{3 \alpha_\ind{a}}{2\pi\i (\ind{a}\cdot\ind{b})}\left(\omega^{\ind{a}\cdot \ind{b}}-1\right)C_\ind{a} = \overline{H_{I,\ind{b}}^{(+)}}.
\end{align}
So the error Hamiltonian is given by
\begin{align}
  H_{E,\ind{b}} & = 2t_p \overline{H_{I,\ind{b}}^{(+)}} = \frac{3 t_p}{\pi\i}  \sum_{\ind{a}\in \supp(\ind{\alpha})} \frac{\alpha_\ind{a}}{\ind{a}\cdot\ind{b}}\left(\omega^{\ind{a}\cdot\ind{b}}-1\right)C_\ind{a}\\
  & = 2 t_p \sum_{\ind{a} \in \supp(\ind{\alpha})} \alpha_\ind{a} \e^{\i\pi(\ind{a}\cdot\ind{b})/3}\operatorname{sinc}{\left(\frac{\pi}{3}(\ind{a}\cdot\ind{b})\right)}C_\ind{a}.
\end{align}
The effective Hamiltonian for the $\ind{b}$-th step is

\begin{align}
  \overline{K_\ind{b}} & = \frac{1}{2t_p+\lambda_\ind{b}t}\left[2 t_p \sum_{\ind{a} \in \supp(\ind{\alpha})} \alpha_\ind{a} \e^{\i\pi(\ind{a}\cdot\ind{b})/3}\operatorname{sinc}{\left(\frac{\pi}{3}(\ind{a}\cdot\ind{b})\right)}C_\ind{a} + \lambda_\ind{b}t V_\ind{b}^\dagger H_S V_\ind{b}\right] \\
  & = \frac{1}{2t_p+\lambda_\ind{b}t} \left[ \sum_{\ind{a}\in \supp(\ind{\alpha})}\alpha_\ind{a}\left\{ 2 t_p \e^{\i\pi(\ind{a}\cdot\ind{b})/3}  \operatorname{sinc}{\left(\frac{\pi}{3}(\ind{a}\cdot\ind{b})\right)} +  \omega^{\ind{a}\cdot\ind{b}}\lambda_\ind{b}t\right\}C_\ind{a}\right].
\end{align}

Let $T_\ind{b}\coloneqq 2t_p+\lambda_\ind{b}t$ be the total time for the $\ind{b}$-th step.
Then we have
\begin{align}
  \prod_{\ind{b}}S_\ind{b} = \prod_{\ind{b}}\e^{-\i \overline{K_\ind{b}}T_\ind{b}}\approx \e^{-\i\sum_\ind{b}\overline{K_\ind{b}}T_b}\stackrel{!}{=}\e^{-\i H_T \tau}.
\end{align}
This means that

\begin{align}
  \tau \sum_{\ind{a}\in\mathrm{supp}(\boldsymbol{\beta})}\beta_\ind{a} C_\ind{a} & = \sum_\ind{b} \sum_{\ind{a}\in \supp(\ind{\alpha})}\alpha_\ind{a}\left\{ 2 t_p \e^{\i\pi(\ind{a}\cdot\ind{b})/3}  \operatorname{sinc}{\left(\frac{\pi}{3}(\ind{a}\cdot\ind{b})\right)} + \omega^{\ind{a}\cdot\ind{b}}\lambda_\ind{b}t \right\}C_\ind{a}\\
  & = \sum_{\ind{a}\in \supp(\ind{\alpha})}\alpha_\ind{a}\left[\left(\sum_\ind{b}\omega^{\ind{a}\cdot\ind{b}}\lambda_\ind{b}t\right)+\left(\sum_\ind{b}2 t_p \e^{\i\pi(\ind{a}\cdot\ind{b})/3}  \operatorname{sinc}{\left(\frac{\pi}{3}(\ind{a}\cdot\ind{b})\right)}\right)\right],
\end{align}

or in matrix notation
\begin{align}
  \label{eq:fpt_constraint_matrix_appendix}
  tF'\boldsymbol{\lambda}+E\boldsymbol{1} = \tau \boldsymbol{\gamma}',
\end{align}
where we defined the \emph{error matrix} $E \in \CC^{r\times c}$, with
\begin{align}
  E_{\ind{a}\ind{b}} = 2 t_p \e^{\i\pi(\ind{a}\cdot\ind{b})/3}  \operatorname{sinc}{\left(\frac{\pi}{3}(\ind{a}\cdot\ind{b})\right)}
\end{align}
 and $\boldsymbol{\gamma}' = \boldsymbol{\beta}\oslash\boldsymbol{\alpha} \big|_{\supp(\alpha)}$ and $F' = (F^{(3)})^{\otimes n}\big|_{\supp(\alpha)}$ as before. 
We may rewrite \cref{eq:fpt_constraint_matrix_appendix}
\begin{align}
  \label{eq:fpt_effective_constraint}
  F'\boldsymbol{\lambda} = \boldsymbol{\eta}
\end{align}
where we defined $\boldsymbol{\eta} \coloneqq 1/t(\tau \boldsymbol{\gamma}' - E\boldsymbol{1})\in\CC^r$.
Crucially, \cref{eq:fpt_effective_constraint} is of the same form as before, and $\mathrm{supp}(\boldsymbol{\eta}) = \supp(\ind{\alpha})$, so the proof of \cref{thm:main} applies to this case as well and the linear program
\begin{align}
\begin{split}
\text{minimize} \quad & \boldsymbol{1}^\top \boldsymbol{\lambda} \\
\text{subject to} \quad &  F' \boldsymbol{\lambda} = \boldsymbol{\eta}, \\
                        &  \boldsymbol{\lambda} \in \RR_{\geq 0}^{c},
\end{split}
\end{align}
has an optimal solution, whenever $\mathrm{supp}(\boldsymbol{\beta})\subseteq\mathrm{supp}(\boldsymbol{\alpha})$.

\end{document}